\documentclass[a4paper,USenglish,cleveref, autoref, thm-restate]{lipics-v2021}
\hideLIPIcs  %uncomment to remove references to LIPIcs series (logo, DOI, ...), e.g. when preparing a pre-final version to be uploaded to arXiv or another public repository

\title{Preprocessing to Reduce the Search Space for Odd Cycle Transversal}

\newcommand{\universityOfUtah}{School of Computing, University of Utah, USA}
\newcommand{\eindhoven}{Eindhoven University of Technology, The Netherlands}

\author{Bart M. P. Jansen}{\eindhoven}{b.m.p.jansen@tue.nl}{https://orcid.org/0000-0001-8204-1268}{}
\author{Yosuke Mizutani}{\universityOfUtah}{yos@cs.utah.edu}{https://orcid.org/0000-0002-9847-4890}{}
\author{Blair D. Sullivan}{\universityOfUtah}{sullivan@cs.utah.edu}{https://orcid.org/0000-0001-7720-6208}{Gordon \& Betty Moore Foundation under grant GBMF4560.}
\author{Ruben F. A. Verhaegh}{\eindhoven}{r.f.a.verhaegh@tue.nl}{https://orcid.org/0009-0008-8568-104X}{}

\authorrunning{B. M. P. Jansen, Y. Mizutani, B. D. Sullivan, and R. F. A. Verhaegh} %TODO mandatory. First: Use abbreviated first/middle names. Second (only in severe cases): Use first author plus 'et al.'

\Copyright{Bart M. P. Jansen, Yosuke Mizutani, Blair D. Sullivan, and Ruben F. A. Verhaegh} %TODO mandatory, please use full first names. LIPIcs license is "CC-BY";  http://creativecommons.org/licenses/by/3.0/


 %TODO mandatory; please add comma-separated list of keywords

\category{} %optional, e.g. invited paper

\relatedversion{} %optional, e.g. full version hosted on arXiv, HAL, or other respository/website
\nolinenumbers %uncomment to disable line numbering

\EventEditors{John Q. Open and Joan R. Access}
\EventNoEds{2}
\EventLongTitle{42nd Conference on Very Important Topics (CVIT 2016)}
\EventShortTitle{CVIT 2016}
\EventAcronym{CVIT}
\EventYear{2016}
\EventDate{December 24--27, 2016}
\EventLocation{Little Whinging, United Kingdom}
\EventLogo{}
\SeriesVolume{42}
\ArticleNo{23}
\input{styles/include.sty}
\input{styles/problems.sty}
\input{styles/variables.sty}
\input{styles/utilities.sty}

\begin{document}

\maketitle

\begin{abstract}
  The NP-hard \textsc{Odd Cycle Transversal} problem asks for a minimum vertex set whose removal from an undirected input graph~$G$ breaks all odd cycles, and thereby yields a bipartite graph. The problem is well-known to be fixed-parameter tractable when parameterized by the size~$k$ of the desired solution. It also admits a randomized kernelization of polynomial size, using the celebrated matroid toolkit by Kratsch and Wahlstr\"{o}m. The kernelization guarantees a reduction in the total \emph{size} of an input graph, but does not guarantee any decrease in the size of the solution to be sought; the latter governs the size of the search space for FPT algorithms parameterized by~$k$. We investigate under which conditions an efficient algorithm can detect one or more vertices that belong to an optimal solution to \textsc{Odd Cycle Transversal}. By drawing inspiration from the popular \emph{crown reduction} rule for \textsc{Vertex Cover}, and the notion of \emph{antler decompositions} that was recently proposed for \textsc{Feedback Vertex Set}, we introduce a graph decomposition called \emph{tight odd cycle cut} that can be used to certify that a vertex set is part of an optimal odd cycle transversal. While it is NP-hard to compute such a graph decomposition, we develop parameterized algorithms to find a set of at least~$k$ vertices that belong to an optimal odd cycle transversal when the input contains a tight odd cycle cut certifying the membership of~$k$ vertices in an optimal solution. The resulting algorithm formalizes when the search space for the solution-size parameterization of \textsc{Odd Cycle Transversal} can be reduced by preprocessing. To obtain our results, we develop a graph reduction step that can be used to simplify the graph to the point that the odd cycle cut can be detected via color coding.
\end{abstract}

% {\huge{Paper submission deadline: 22 April, 23:59 AoE}}

\newpage

\section{Introduction} \label{sec:intro}

% Shorter opening for IPEC.

The NP-hard \textsc{Odd Cycle Transversal} problem asks for a minimum vertex set whose removal from an undirected input graph~$G$ breaks all odd cycles, and thereby yields a bipartite graph. Finding odd cycle transversals has important applications, for example in computational biology~\cite{Huffner09,Wernicke14} and adiabatic quantum computing~\cite{GoodrichHS21,Goodrich18}. \textsc{Odd Cycle Transversal} parameterized by the desired solution size~$k$ has been studied intensively, leading to important advances such as \emph{iterative compression}~\cite{ReedSV04} and \emph{matroid-based kernelization}~\cite{KratschW14,KratschW20}. The randomized kernel due to Kratsch and Wahlstr\"{o}m~\cite[Lemma 7.11]{KratschW20} is a polynomial-time algorithm that reduces an $n$-vertex instance~$(G,k)$ of \textsc{Odd Cycle Transversal} to an instance~$(G',k')$ on~$\Oh((k \log k \log \log k)^3)$ vertices, that is equivalent to the input instance with probability at least~$2^{-n}$. Experiments with this matroid-based kernelization, however, show disappointing preprocessing results in practice~\cite{PilipczukZ18}. This formed one of the motivations for a recent line of research aimed at preprocessing that reduces the \emph{search space} explored by algorithms solving the reduced instance, rather than preprocessing aimed at reducing the \emph{encoding size} of the instance (which is captured by kernelization). To motivate our work, we present some background on this topic.

A \emph{kernelization} of size~$f \colon \mathbb{N} \to \mathbb{N}$ for a parameterized problem~$\mathcal{P}$ is a polynomial-time algorithm that reduces any parameterized instance~$(x,k)$ to an instance~$(x',k')$ with the same \textsc{yes/no} answer, such that~$|x'|, k' \leq f(k)$. It therefore guarantees that the size of the instance is reduced in terms of the complexity parameter~$k$. It does not directly ensure a reduction in the \emph{search space} of the follow-up algorithm that is employed to solve the reduced instance. Since the running times of FPT algorithms for the natural parameterization of \textsc{Odd Cycle Transversal}~\cite{Huffner09,ReedSV04,LokshtanovNRRS14} depend exponentially on the size of the sought solution, the size of the search space considered by such algorithms can be reduced significantly by a preprocessing step that finds some vertices~$S$ that belong to an optimal solution for the input graph~$G$: the search for a solution of size~$k$ on~$G$ then reduces to the search for a solution of size~$k-|S|$ on~$G-S$. Researchers therefore started to investigate in which situations an efficient preprocessing phase can guarantee finding part of an optimal solution. 

One line of inquiry in this direction aims at finding vertices that not only belong to an optimal solution, but are even required for building a $c$-approximate solution~\cite{BumpusJK22,JansenV24}; such vertices are called \emph{$c$-essential}. This has resulted in refined running time guarantees, showing that an optimal odd cycle transversal of size~$k$ can be found in time~$2.3146^{k - \ell} \cdot n^{\Oh(1)}$, where~$\ell$ is the number of vertices in the instance that are essential for making a 3-approximate solution~\cite{BumpusJK22}. Another line of research, more relevant to the subject of this paper, aims at finding vertices that belong to an optimal solution when there is a simple, locally verifiable certificate of the existence of an optimal solution containing them. So far, the latter direction has been explored for \textsc{Vertex Cover} (where a \emph{crown decomposition}~\cite{Abu-KhzamFLS07,Fellows03} forms such a certificate), and for the (undirected) \textsc{Feedback Vertex Set} problem (where an \emph{antler decomposition}~\cite{DonkersJ24}) forms such a certificate. 

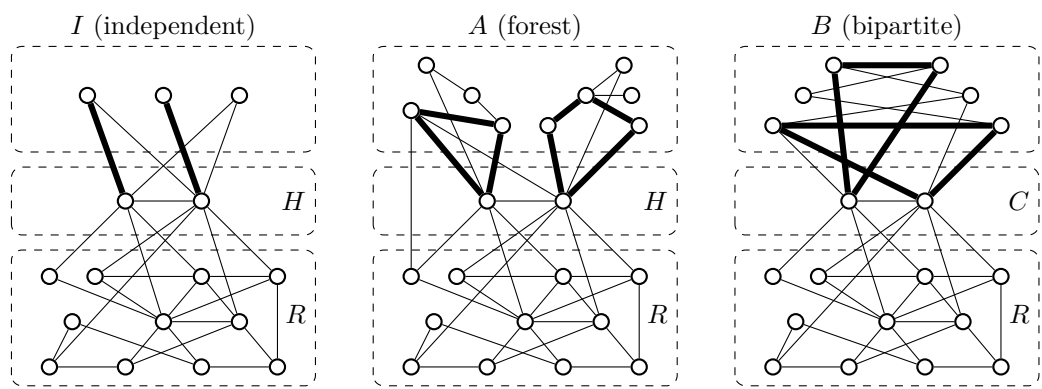
\begin{figure*}[t]
  \centering

  \tikzstyle{plain} = [circle, fill=white, text=black, draw, thick, scale=1, minimum size=0.5cm, inner sep=1.5pt]
  \tikzstyle{small} = [circle, fill=white, text=black, draw, thick, scale=1, minimum size=0.2cm, inner sep=1.5pt]
  \tikzstyle{thickedge} = [line width=0.8mm]

  \begin{minipage}[m]{.32\textwidth}
      \vspace{0pt}
      \centering
      \begin{tikzpicture}
          % vertices
          \node[small] (r1) at (1, 1.8) {};
          \node[small] (r2) at (2, 1.8) {};
          \node[small] (r3) at (3, 1.8) {};
          \node[small] (r4) at (4, 1.8) {};

          \node[small] (r5) at (1.3, 2.4) {};
          \node[small] (r6) at (2.5, 2.4) {};
          \node[small] (r7) at (3.5, 2.4) {};

          \node[small] (r8) at (1, 3) {};
          \node[small] (r9) at (1.6, 3) {};
          \node[small] (r10) at (3, 3) {};
          \node[small] (r11) at (4, 3) {};

          \node[small] (c1) at (2, 4) {};
          \node[small] (c2) at (3, 4) {};

          \node[small] (b1) at (1.5, 5.4) {};
          \node[small] (b2) at (2.5, 5.4) {};
          \node[small] (b3) at (3.5, 5.4) {};

          % edges
          \draw (r1) -- (r2);
          \draw (r3) -- (r4);
          \draw (r1) -- (r5);
          \draw (r2) -- (r6);
          \draw (r2) -- (r7);
          \draw (r3) -- (r5);
          \draw (r3) -- (r6);
          \draw (r4) -- (r7);
          \draw (r6) -- (r7);
          \draw (r6) -- (r8);
          \draw (r6) -- (r9);
          \draw (r6) -- (r10);
          \draw (r6) -- (r11);
          \draw (r7) -- (r10);
          \draw (r4) -- (r11);
          \draw (r9) -- (r10);
          \draw (r10) -- (r11);

          \draw (r8) -- (c1);
          \draw (r6) -- (c1);
          \draw (r10) -- (c1);
          \draw (r7) -- (c2);
          \draw (r9) -- (c2);
          \draw (r11) -- (c2);
          \draw (r1) -- (c2);

          \draw (c1) -- (c2);
          \draw[thickedge] (c1) -- (b1);
          \draw (c1) -- (b3);
          \draw (c2) -- (b1);
          \draw[thickedge] (c2) -- (b2);
          \draw (c2) -- (b3);

          % annotation
          \draw[rounded corners, dashed] (0.5, 1.55) rectangle ++ (4, 1.8);
          \node[anchor=east] at (4.5, 2.5) {$R$};

          \draw[rounded corners, dashed] (0.5, 3.55) rectangle ++ (4, 0.9);
          \node[anchor=east] at (4.5, 4) {$H$};

          \draw[rounded corners, dashed] (0.5, 4.65) rectangle ++ (4, 1.4);
          \node[] at (2.5, 6.3) {$I$ (independent)};
      \end{tikzpicture}
  \end{minipage}
  \hfill
  \begin{minipage}[m]{.32\textwidth}
    \vspace{0pt}
    \centering
    \begin{tikzpicture}
        % vertices
        \node[small] (r1) at (1, 1.8) {};
        \node[small] (r2) at (2, 1.8) {};
        \node[small] (r3) at (3, 1.8) {};
        \node[small] (r4) at (4, 1.8) {};

        \node[small] (r5) at (1.3, 2.4) {};
        \node[small] (r6) at (2.5, 2.4) {};
        \node[small] (r7) at (3.5, 2.4) {};

        \node[small] (r8) at (1, 3) {};
        \node[small] (r9) at (1.6, 3) {};
        \node[small] (r10) at (3, 3) {};
        \node[small] (r11) at (4, 3) {};

        \node[small] (c1) at (2, 4) {};
        \node[small] (c2) at (3, 4) {};

        \node[small] (b1) at (2.2, 5) {};
        \node[small] (b2) at (1.8, 5.4) {};
        \node[small] (b3) at (1.2, 5.8) {};
        \node[small] (b4) at (1, 5.2) {};

        \node[small] (b5) at (2.8, 5) {};
        \node[small] (b6) at (3.3, 5.4) {};
        \node[small] (b7) at (4, 5.0) {};
        \node[small] (b8) at (3.9, 5.4) {};
        \node[small] (b9) at (3.8, 5.8) {};

        % edges
        \draw (r1) -- (r2);
        \draw (r3) -- (r4);
        \draw (r1) -- (r5);
        \draw (r2) -- (r6);
        \draw (r2) -- (r7);
        \draw (r3) -- (r5);
        \draw (r3) -- (r6);
        \draw (r4) -- (r7);
        \draw (r6) -- (r7);
        \draw (r6) -- (r8);
        \draw (r6) -- (r9);
        \draw (r6) -- (r10);
        \draw (r6) -- (r11);
        \draw (r7) -- (r10);
        \draw (r4) -- (r11);
        \draw (r9) -- (r10);
        \draw (r10) -- (r11);

        \draw (r8) -- (c1);
        \draw (r6) -- (c1);
        \draw (r10) -- (c1);
        \draw (r7) -- (c2);
        \draw (r9) -- (c2);
        \draw (r11) -- (c2);
        \draw (r1) -- (c2);

        \draw (c1) -- (c2);
        \draw[thickedge] (c1) -- (b1);
        \draw (c1) -- (b3);
        \draw[thickedge] (c1) -- (b4);
        \draw (c2) -- (b4);
        \draw[thickedge] (c2) -- (b5);
        \draw[thickedge] (c2) -- (b7);
        \draw (c2) -- (b9);

        \draw (b1) -- (b2);
        \draw (b2) -- (b3);
        \draw[thickedge] (b1) -- (b4);
        \draw[thickedge] (b5) -- (b6);
        \draw[thickedge] (b6) -- (b7);
        \draw (b6) -- (b8);
        \draw (b6) -- (b9);

        \draw (b4) -- (r8);

        % annotation
        \draw[rounded corners, dashed] (0.5, 1.55) rectangle ++ (4, 1.8);
        \node[anchor=east] at (4.5, 2.5) {$R$};

        \draw[rounded corners, dashed] (0.5, 3.55) rectangle ++ (4, 0.9);
        \node[anchor=east] at (4.5, 4) {$H$};

        \draw[rounded corners, dashed] (0.5, 4.65) rectangle ++ (4, 1.4);
        \node[] at (2.5, 6.3) {$A$ (forest)};
      \end{tikzpicture}
  \end{minipage}
  \hfill
  \begin{minipage}[m]{.32\textwidth}
    \vspace{0pt}
    \centering
    \begin{tikzpicture}
        % vertices
        \node[small] (r1) at (1, 1.8) {};
        \node[small] (r2) at (2, 1.8) {};
        \node[small] (r3) at (3, 1.8) {};
        \node[small] (r4) at (4, 1.8) {};

        \node[small] (r5) at (1.3, 2.4) {};
        \node[small] (r6) at (2.5, 2.4) {};
        \node[small] (r7) at (3.5, 2.4) {};

        \node[small] (r8) at (1, 3) {};
        \node[small] (r9) at (1.6, 3) {};
        \node[small] (r10) at (3, 3) {};
        \node[small] (r11) at (4, 3) {};

        \node[small] (c1) at (2, 4) {};
        \node[small] (c2) at (3, 4) {};

        \node[small] (b1) at (1, 5) {};
        \node[small] (b2) at (1.4, 5.4) {};
        \node[small] (b3) at (1.8, 5.8) {};

        \node[small] (b4) at (4, 5) {};
        \node[small] (b5) at (3.6, 5.4) {};
        \node[small] (b6) at (3.2, 5.8) {};

        % edges
        \draw (r1) -- (r2);
        \draw (r3) -- (r4);
        \draw (r1) -- (r5);
        \draw (r2) -- (r6);
        \draw (r2) -- (r7);
        \draw (r3) -- (r5);
        \draw (r3) -- (r6);
        \draw (r4) -- (r7);
        \draw (r6) -- (r7);
        \draw (r6) -- (r8);
        \draw (r6) -- (r9);
        \draw (r6) -- (r10);
        \draw (r6) -- (r11);
        \draw (r7) -- (r10);
        \draw (r4) -- (r11);
        \draw (r9) -- (r10);
        \draw (r10) -- (r11);

        \draw (r8) -- (c1);
        \draw (r6) -- (c1);
        \draw (r10) -- (c1);
        \draw (r7) -- (c2);
        \draw (r9) -- (c2);
        \draw (r11) -- (c2);
        \draw (r1) -- (c2);

        \draw (c1) -- (c2);
        \draw (c1) -- (b1);
        \draw[thickedge] (c1) -- (b3);
        \draw[thickedge] (c1) -- (b6);
        \draw[thickedge] (c2) -- (b1);
        \draw[thickedge] (c2) -- (b4);
        \draw (c2) -- (b5);

        \draw[thickedge] (b1) -- (b4);
        \draw (b1) -- (b5);
        \draw (b2) -- (b4);
        \draw (b2) -- (b6);
        \draw (b3) -- (b5);
        \draw[thickedge] (b3) -- (b6);

        % annotation
        \draw[rounded corners, dashed] (0.5, 1.55) rectangle ++ (4, 1.8);
        \node[anchor=east] at (4.5, 2.5) {$R$};

        \draw[rounded corners, dashed] (0.5, 3.55) rectangle ++ (4, 0.9);
        \node[anchor=east] at (4.5, 4) {$C$};

        \draw[rounded corners, dashed] (0.5, 4.65) rectangle ++ (4, 1.4);
        \node[] at (2.5, 6.3) {$B$ (bipartite)};
      \end{tikzpicture}
  \end{minipage}
  \caption{%
  Examples of crown decomposition (left),
  antler decomposition for \PrbFVSLong (middle) and 
  a tight OCC for \PrbOCTLong (right).
  Packings of forbidden subgraphs are highlighted in bold.
  }
  \label{fig:example}
\end{figure*}

A \emph{crown decomposition} (see \cref{fig:example}) of a graph~$G$ consists of a partition of its vertex set into three parts: the \emph{crown}~$I$ (which is required to be a non-empty independent set), the \emph{head}~$H$ (which is required to contain all neighbors of~$I$), and the \emph{remainder}~$R = V(G) \setminus (I \cup H)$, such that the graph~$G[I \cup H]$ contains a matching~$M$ of size~$|H|$. Since~$I$ is an independent set, this matching partners each vertex of~$H$ with a private neighbor in~$I$. The existence of a crown decomposition shows that there is an optimal vertex cover (a minimum-size vertex set intersecting all edges) that contains all vertices of~$H$ and none of~$I$: any vertex cover contains at least~$|M| = |H|$ vertices from~$I \cup H$ to cover the matching~$M$, while~$H$ covers all the edges of~$G$ that can be covered by selecting vertices from~$I \cup H$. Hence a crown decomposition forms a polynomial-time verifiable certificate that there is an optimal vertex cover containing all vertices of~$H$. It facilitates a reduction in search space for \textsc{Vertex Cover}: graph~$G$ has a vertex cover of size~$k$ if and only if~$G - (I \cup H)$ has one of size~$k - |H|$. A crown decomposition can be found in polynomial time if it exists, which yields a powerful reduction rule for \textsc{Vertex Cover}~\cite{Abu-KhzamFLS07}.

% \bmpr{To be discussed: whether we want to require an antler to be non-empty, or to use the adjective everywhere.}

Inspired by this decomposition for \textsc{Vertex Cover}, Donkers and Jansen~\cite{DonkersJ24} introduced the notion of an \emph{antler decomposition} of a graph~$G$. It is a partition of the vertex set into three parts: the \emph{antler}~$A$ (which is required to induce a non-empty acyclic graph), the \emph{head}~$H$ (which is required to contain \emph{almost} all neighbors of~$A$: for each tree~$T$ in the forest~$G[A]$, there is at most one edge that connects~$T$ to a vertex outside~$H$), and the \emph{remainder}~$R = V(G) \setminus (A \cup H)$, while satisfying an additional condition in terms of an integer~$z$ that represents the \emph{order} of the antler decomposition. In its simplest form for~$z=1$ (we discuss~$z > 1$ later), the additional condition says that the graph~$G[A \cup H]$ should contain~$|H|$ vertex-disjoint cycles. Since~$G[A]$ is acyclic, each of these cycles contains exactly one vertex of~$H$. They certify that any feedback vertex set of~$G$ contains at least~$|H|$ vertices from~$A \cup H$. Since~$A$ induces an acyclic graph, and all cycles in~$G$ that enter a tree~$T$ of~$G[A]$ from~$R$ must leave~$A$ from~$H$, the set~$H$ intersects all cycles of~$G$ that contain a vertex of~$A \cup H$. Hence there is an optimal feedback vertex set containing~$H$. By finding an antler decomposition we can therefore reduce the problem of finding a size-$k$ solution in~$G$ to finding a size-($k - |H|$) solution in~$G - (A \cup H)$, and therefore reduce the search space for algorithms parameterized by solution size.

Donkers and Jansen proved that, assuming $\mathsf{P} \neq \mathsf{NP}$, there unfortunately is no polynomial-time algorithm to find an antler decomposition if one exists~\cite[Theorem 3.4]{DonkersJ24}. However, they gave a \emph{fixed-parameter tractable} preprocessing algorithm, parameterized by the size of the head. There is an algorithm that, given a graph~$G$ and integer~$k$ such that~$G$ contains an antler decomposition~$(A,H,R)$ with~$|H| = k$, runs in time~$2^{\Oh(k^5)} \cdot n^{\Oh(1)}$ and outputs a set of at least~$k$ vertices that belong to an optimal feedback vertex set. For each fixed value of~$k$, this yields a preprocessing algorithm to detect vertices that belong to an optimal solution if there is a simple certificate of their membership in an optimal solution.

In fact, Donkers and Jansen gave a more general algorithm; this is where $z$-antlers for~$z > 1$ make an appearance. Recall that for a $1$-antler decomposition~$(A,H,R)$ of a graph~$G$, the graph~$G[A \cup H]$ must contain a collection~$\mathcal{C}$ of~$|H|$ vertex-disjoint cycles. These cycles certify that the set~$H$ is an optimal feedback vertex set in the graph~$G[A \cup H]$. In fact, the feedback vertex set~$H$ in~$G[A \cup H]$ is already optimal for the subgraph~$\mathcal{C} \subseteq G[A \cup H]$, and that subgraph~$\mathcal{C}$ is structurally simple because each of its connected components (which is a cycle) has a feedback vertex set of size~$z=1$. This motivates the following definition of a $z$-antler decomposition for~$z>1$: the set~$H$ should be an optimal feedback vertex set for the subgraph~$G[A \cup H]$, and moreover, there should be a subgraph~$\mathcal{C}_z \subseteq G[A \cup H]$ such that (1)~$H$ is an optimal feedback vertex set in~$\mathcal{C}_z$, and (2)~each connected component of~$\mathcal{C}_z$ has a feedback vertex set of size at most~$z$. So for a $z$-antler decomposition~$(A,H,R)$ of a graph~$G$, there is a certificate that~$H$ is part of an optimal solution in the overall graph~$G$ that consists of the decomposition together with the subgraph~$\mathcal{C}_z \subseteq G[A \cup H]$ for which~$H$ is an optimal solution. The complexity of verifying this certificate scales with~$z$: it comes down to verifying that~$H \cap V(C)$ is indeed an optimal feedback vertex set of size at most~$z$ for each connected component of the subgraph~$\mathcal{C}_z$. Donkers and Jansen presented an algorithm that, given integers~$k \geq z \geq 0$ and a graph~$G$ that contains a $z$-antler decomposition whose head has size~$k$, outputs a set of at least~$k$ vertices that belongs to an optimal feedback vertex set in time~$2^{\Oh(k^5 z^2)} n^{\Oh(z)}$. For each fixed choice of~$k$ and~$z$, this gives a reduction rule (that can potentially be applied numerous times on an instance) to reduce the search space if the preconditions are met.

\subparagraph*{Our contribution} We investigate search-space reduction for \textsc{Odd Cycle Transversal}, thereby continuing the line of research proposed by Donkers and Jansen~\cite{DonkersJ24}. We introduce the notion of \emph{tight odd cycle cuts} to provide efficiently verifiable witnesses that a certain vertex set belongs to an optimal odd cycle transversal, and present algorithms to find vertices that belong to an optimal solution in inputs that admit such witnesses.

To be able to state our main result, we introduce the corresponding terminology. An \emph{odd cycle cut} (OCC) in an undirected graph~$G$ is a partition of its vertex set into three parts: the bipartite part~$B$ (which is required to induce a bipartite subgraph of~$G$), the cut part~$C$ (which is required to contain all neighbors of~$B$), and the rest~$R = V(G) \setminus (B \cup C)$. An odd cycle cut is called \emph{tight} if the set~$C$ forms an optimal odd cycle transversal for the graph~$G[B \cup C]$. In this case, it is easy to see that there is an optimal odd cycle transversal in~$G$ that contains all vertices of~$C$, since all odd cycles through~$B$ are intersected by~$C$. A tight OCC~$(B,C,R)$ has \emph{order~$z$} if there is a subgraph~$\mathcal{C}_z$ of~$G[B \cup C]$ for which~$C$ is an optimal odd cycle transversal, and for which each connected component of~$\mathcal{C}_z$ has an odd cycle transversal of size at most~$z$. This means that for~$z=1$, if there is such a subgraph~$\mathcal{C}_z \subseteq G[B \cup C]$, then there is one consisting of~$|C|$ vertex-disjoint odd cycles. We use the term $z$-tight OCC to refer to a tight OCC of order~$z$. Our notion of $z$-tight OCCs forms an analogue of $z$-antler decompositions. Note that the requirement that~$C$ contains \emph{all} neighbors of~$B$ is slightly more restrictive than in the \textsc{Feedback Vertex Set} case. We need this restriction for technical reasons, but discuss potential relaxations in \cref{sec:conclusion}.

Similarly to the setting of $z$-antlers for \textsc{Feedback Vertex Set}, assuming $\mathsf{P} \neq \mathsf{NP}$ there is no polynomial-time algorithm that always finds a tight OCC in a graph if one exists; not even in the case~$z=1$ (\cref{thm:1tight:occ:nphard}). We therefore develop algorithms that are efficient for small~$k$ and~$z$. The following theorem captures our main result, which is an OCT-analogue of the antler-based preprocessing algorithm for FVS. The \emph{width} of an OCC~$(B,C,R)$ is defined as~$|C|$. Our theorem shows that for constant~$z$ we can efficiently find~$k$ vertices that belong to an optimal solution, if there is a $z$-tight OCC of width~$k$.
\begin{restatable}{theorem}{thmMain}
  \label{thm:main}
  There is a deterministic algorithm that, given a graph $G$ and integers $k \geq z \geq 0$, runs in $2^{\Oh(k^{33} z^2)}\cdot n^{\Oh{(z)}}$ time and either outputs at least $k$ vertices that belong to an optimal solution for \PrbOCTLong, or concludes that $G$ does not contain a \zocc of width $k$.
\end{restatable}
One may wonder whether it is feasible to have more control over the output, by having the algorithm output a $z$-tight OCC~$(B,C,R)$ of width~$k$, if one exists. However, a small adaptation of a W[1]-hardness proof for antlers~\cite[Theorem 3.7]{DonkersJ24} shows (\cref{thm:tightocc:w1hard}) that the corresponding algorithmic task is W[1]-hard even for~$z=1$. This explains why the algorithm outputs a vertex set that belongs to an optimal solution, rather than a $z$-tight OCC.

In terms of techniques, our algorithm combines insights from the previous work on antlers~\cite{DonkersJ24} with ideas in the representative-set based kernelization~\cite{KratschW20} for \textsc{Odd Cycle Transversal}. The global idea behind the algorithm is to repeatedly simplify the graph, while preserving the structure of $z$-tight OCCs, to arrive at the following favorable situation: if there was a $z$-tight OCC of width~$k$ in the input, then the reduced graph has a $z$-tight OCC~$(B,C,R)$ of the same width that satisfies~$|B| \in k^{\Oh(1)}$. At that point, we can use color coding with a set of~$k^{\Oh(1)}$ colors to ensure that the structure~$B \cup C$ gets colored in a way that makes it tractable to identify it. The simplification steps on the graph are inspired by the kernelization for \textsc{Odd Cycle Transversal} and involve the computation of a \emph{cut covering set} of size~$k^{\Oh(1)}$ that contains a minimum three-way~$\{X,Y,Z\}$-separator
% \rfar{In our application, we specifically construct a cut covering set that contains minimum \emph{3-way} separators. Should the text (which mentions $(X, Y)$-separators) be updated to reflect this? \bmp{Done.}}
for all possible choices of sets~$\{X,Y,Z\}$ drawn from a terminal set~$T$ of size~$k^{\Oh(1)}$. The existence of such sets follows from the matroid-based tools of Kratsch and Wahlstr\"{o}m~\cite{KratschW20}. We can avoid the randomization incurred by their polynomial-time algorithm by computing a cut covering set in~$2^{\Oh(k)} \cdot n^{\Oh(1)}$ time deterministically. Compared to the kernelization for \textsc{Odd Cycle Transversal}, a significant additional challenge we face in this setting is that the size of OCTs in the graph can be arbitrarily large in terms of the parameter~$k$. Our algorithm is looking for a \emph{small} region of the graph in which a vertex set exists with a simple certificate for its membership in an optimal solution; it cannot afford to learn the structure of global OCTs in the graph. This local perspective poses a challenge when repeatedly simplifying the graph: we not only have to be careful how these operations affect the total solution size in~$G$, but also how these modifications affect the existence of simple certificates for membership in an optimal solution. This is why our reduction step works with three-way separators, rather than the two-way separators that suffice to solve or kernelize OCT.

\subparagraph*{Organization} The remainder of this work is organized as follows. The first twelve pages of the manuscript present the key statements and ideas. For statements marked $(\bigstar$), the proof is deferred to a later section for readability. After presenting preliminaries on graphs in \cref{sec:prelims}, we define (tight) OCCs in \cref{sec:oct:cut} and explore some of their properties. In \cref{sec:find:occs} we show how color coding can be used to find an OCC whose bipartite part is connected and significantly larger than its cut. Given such an OCC, we show in \cref{sec:reduce:occ} how to simplify the graph while preserving the essential structure of odd cycles in the graph. This leads to an algorithm that finds vertices belonging to an optimal solution the presence of a tight OCC in \cref{sec:tight:occs}. Sections~\ref{sec:more:prelims}--\ref{sec:tight:occs-proof} contain the deferred proofs. In \cref{sec:hardness} we give the hardness proofs mentioned above. Finally, we conclude in \cref{sec:conclusion}. 

\section{Preliminaries} \label{sec:prelims}

\subparagraph*{Graphs} We only consider finite, undirected, simple graphs. Such a graph~$G$ consists of a set~$V(G)$ of vertices and a set~$E(G) \subseteq \binom{V(G)}{2}$ of edges. For ease of notation, we write~$uv$ for an undirected edge~$\{u,v\} \in E(G)$; note that~$uv = vu$. When it is clear which graph is referenced from context, we write $n$ and $m$ to denote the number of vertices and edges in this graph respectively. For a vertex~$v \in V(G)$, its open neighborhood is~$N_G(v) := \{u \in V(G) \mid uv \in E(G)\}$ and its closed neighborhood is~$N_G[v] := N_G(v) \cup \{v\}$. For a vertex set~$S \subseteq V(G)$ we define its open neighborhood as~$N_G(S) := (\bigcup _{v \in S} N_G(v)) \setminus S$ and its closed neighborhood as~$N_G[S] := \bigcup_{v \in S} N_G[v]$. The subgraph of~$G$ induced by a vertex set~$S \subseteq V(G)$ is the graph~$G[S]$ on vertex set~$S$ with edges~$\{uv \in E(G) \mid \{u,v\} \subseteq S\}$. We use~$G-S$ as a shorthand for~$G[V(G) \setminus S]$ and write~$G - v$ instead of~$G - \{v\}$ for singletons. A \emph{walk} is a sequence of (not necessarily distinct) vertices~$(v_1, \ldots, v_k)$ such that~$v_i, v_{i+1} \in E(G)$ for each~$i \in [k-1]$. The walk is \emph{closed} if we additionally have~$v_k, v_1 \in E(G)$. A \emph{cycle} is a closed walk whose vertices are all distinct. The \emph{length} of a cycle~$(v_1, \ldots, v_k)$ is~$k$. A \emph{path} is a walk whose vertices are all distinct. The \emph{length} of a path~$(v_1, \ldots, v_k)$ is~$k-1$. The vertices~$v_1, v_k$ are the \emph{endpoints} of the path. For two (not necessarily disjoint) vertex sets~$S,T$ of a graph~$G$, we say that a path~$P = (v_1, \ldots, v_k)$ in~$G$ is an~$(S,T)$-path if~$v_1 \in S$ and~$v_k \in T$. If one (or both) of $S$ and $T$ contains only one element, we may write this single element instead of the singleton set consisting of it. 

The \emph{parity} of a path or cycle refers to the parity of its length. For a walk~$W = (v_1, \ldots, v_k)$, we refer to its vertex set as~$V(W) = \{v_1, \ldots, v_k\}$. Observe that if~$W$ is a closed walk of odd parity (a \emph{closed odd walk}), then the graph~$G[V(W)]$ contains a cycle of odd length (an \emph{odd cycle}): any edge connecting two vertices of~$V(W)$ that are not consecutive on~$W$ splits the walk into two closed subwalks, one of which has odd length.

For a positive integer~$q$, a \emph{proper $q$-coloring} of a graph~$G$ is a function~$f \colon V(G) \to \{0, \ldots, q-1\}$ such that~$f(u) \neq f(v)$ for all~$uv \in E(G)$. A graph~$G$ is \emph{bipartite} if its vertex set can be partitioned into two \emph{partite sets}~$L \dot \cup R$ such that no edge has both of its endpoints in the same partite set. It is well-known that the following three conditions are equivalent for any graph~$G$: (1)~$G$ is bipartite,~(2)~$G$ admits a proper $2$-coloring, and~(3)~there is no cycle of odd length in~$G$. An \emph{odd cycle transversal} (OCT) of a graph~$G$ is a set~$S \subseteq V(G)$ such that~$G - S$ is bipartite. An \emph{independent set} is a vertex set~$S$ such that~$G[S]$ is edgeless. We say that a vertex set~$X$ in a graph~$G$ \emph{separates} two (not necessarily) disjoint vertex sets~$S$ and~$T$ if no connected component of~$G - X$ simultaneously contains a vertex from~$S$ and a vertex from~$T$. For a collection~$\{T_1, \ldots, T_m\}$ of (not necessarily disjoint) vertex sets in a graph~$G$, we say that a vertex set~$X$ is an \emph{$\{T_1, \ldots, T_m\}$-separator} if~$X$ separates all pairs~$(T_i, T_j)$ for~$i \neq j$. Note that~$X$ is allowed to intersect~$\bigcup _{i \in [m]} T_i$.

The next lemma gives a simple sufficient condition for a graph to be bipartite.

\begin{restatable}{lemma}{lemSeparatedColoringsMakeBipartite} \label{lem:separated:colorings:make:bipartite}
	Let $G$ be a graph and let $V_L \cup V_0 \cup V_R = V(G)$ be a partition of its vertices such that $V_0$ is a $\{V_L, V_R\}$-separator. If there exist proper $2$-colorings $f_L: (V_0 \cup V_L) \rightarrow \{0, 1\}$ and $f_R: (V_0 \cup V_R) \rightarrow \{0, 1\}$ of $G[V_0 \cup V_L]$ and $G[V_0 \cup V_R]$ respectively such that $f_L(v_0) = f_R(v_0)$ for every $v_0 \in V_0$, then $G$ is bipartite.
\end{restatable}
\begin{proof}
	To show that $G$ is bipartite, we provide a proper $2$-coloring of the graph. We define this coloring $f \colon V(G) \rightarrow \{0, 1\}$ such that $f(v_0) = f_L(v_0) (= f_R(v_0))$ for every $v_0 \in V_0$, $f(v_L) = f_1(v_L)$ for every $v_L \in V_L$ and $f(v_R) = f_2(v_R)$ for every $v_R \in V_R$. To see that~$f$ is a proper $2$-coloring, we show that no edge~$e \in E(G)$ is monochromatic under~$f$. 
	
	By the assumption that~$V_0$ is a separator, each edge~$e \in E(G)$ is contained in~$G[V_0 \cup V_L]$ or~$G[V_0 \cup V_R]$ (or both). If~$e$ is an edge in the former, its endpoints are colored the same as in~$f_L$ and are therefore bichromatic. The analogous argument for~$f_R$ holds when~$e$ is an edge of the latter.
\end{proof}

The next lemma captures the main idea behind the iterative compression algorithm~\cite{ReedSV04} (cf.~\cite[\S 4.4]{CyganFKLMPPS15}) for solving \textsc{Odd Cycle Transversal}. Given a (potentially suboptimal) odd cycle transversal~$W$ of a graph, it shows that the task of finding an odd cycle transversal disjoint from~$W$ whose removal leaves a bipartite graph with~$W_0, W_1 \subseteq W$ in opposite partite sets of its bipartition is equivalent to separating two vertex sets derived from a baseline bipartition of~$G-W$. Our statement below is implied by Claim 1 in the work of Jansen and de Kroon~\cite{JansenK21}.
\begin{lemma}[{\cite[Claim 1]{JansenK21}}] \label{lem:AR-separation}
	Let~$W$ be an OCT in graph~$G$. For each partition of $W = W_0 \cup W_1$ into two independent sets, for each proper $2$-coloring $c$ of $G - W$, we have the following equivalence for each $X \subseteq V(G) \setminus W$: the graph $G - X$ has a proper $2$-coloring with $W_0$ color $0$ and $W_1$ color $1$ \emph{if and only if} the set $X$ separates $A$ from $R$ in the graph $G - W$, with:
	\begin{align*}
		A &= (N_G(W_0) \cap c^{-1}(0)) \cup (N_G(W_1) \cap c^{-1}(1)), \\ 
		R &= (N_G(W_0) \cap c^{-1}(1)) \cup (N_G(W_1) \cap c^{-1}(0)).
	\end{align*}
\end{lemma}

\subparagraph*{Multiway cuts} Let~$\mathcal{T} = (T_1, \ldots, T_s)$ be a partition of a set~$T \subseteq V(G)$ of \emph{terminal} vertices in an undirected graph~$G$. A \emph{multiway cut} of~$\mathcal{T}$ in~$G$ is a vertex set~$X \subseteq V(G)$ such that for each pair~$t_i, t_j \in T \setminus X$ that belong to different parts of partition~$\mathcal{T}$, the graph~$G-X$ does not contain a path from~$t_i$ to~$t_j$. A \emph{restricted} multiway cut of~$\mathcal{T}$ is a vertex set~$X$ that is a multiway cut for~$\mathcal{T}$ such that~$X \cap T = \emptyset$, i.e., it does not contain any terminals. 

For a positive integer~$s$, a \emph{generalized $s$-partition} of a set~$T$ is a partition~$\mathcal{T}^* = (T_0, T_1, \ldots, T_s, \linebreak[1] T_X)$ of~$T$ into~$s+2$ parts, some of which can be empty. The parts~$T_0$ and~$T_X$ play a special role, which are the \emph{free} and \emph{deleted} part of~$\mathcal{T}^*$, respectively. Let~$T' = T_1 \cup \ldots \cup T_s$. A \emph{multiway cut} of $\mathcal{T}^*$ is a (non-restricted) multiway cut in~$G - T_X$ of the partition $\mathcal{T} = (T_1, \ldots, T_s)$ of~$T'$. Hence the vertices of~$T_X$ are deleted from the graph, while no cut constraints are imposed on the vertices of~$T_0$. 

A \emph{minimum multiway cut} of a generalized $s$-partition $\mathcal{T}^*$ in a graph~$G$ is a minimum-cardinality vertex set that satisfies the requirements of a multiway cut for $\mathcal{T}^*$. We denote the size of a minimum multiway cut of~$\mathcal{T}^*$ in~$G$ by~$\mwc(G, \mathcal{T}^*)$. The following cut covering lemma by Kratsch and Wahlstr\"{o}m will be useful for our algorithm.

\begin{theorem}[{\cite[Theorem 5.14]{KratschW20}}]\label{thm:multiway-cover}
	Let~$G$ be an undirected graph on $n$ vertices with a set~$T \subseteq V(G)$ of terminal vertices, and let~$s \in \mathbb{N}$ be a constant. There is a set~$Z \subseteq V(G)$ with~$|Z| = \Oh(|T|^{s+1})$ such that~$Z$ contains a minimum multiway cut of every generalized $s$-partition~$\mathcal{T}^*$ of~$T$, and we can compute such a set in randomized polynomial time with failure probability~$\mathcal{O}(2^{-n})$.
\end{theorem}

For a generalized $s$-partition~$\mathcal{T} = (T_0, T_1, \ldots, T_s, T_X)$ of a terminal set~$T \subseteq V(G)$ in an undirected graph~$G$, we call a multiway cut~$X$ of~$\mathcal{T}$ \emph{restricted} if it satisfies~$X \cap (\bigcup_{i=1}^s T_i) = \emptyset$. Hence a restricted multiway cut does not delete any vertex that is active as a terminal in the generalized partition. A minimum \emph{restricted} multiway cut of~$\mathcal{T}$ is a restricted multiway cut whose size is minimum among all restricted multiway cuts. We denote the minimum size of a \emph{restricted} multiway cut of~$\mathcal{T}$ in~$G$ by~$\rmwc(G, \mathcal{T})$, which we define as~$+\infty$ if no such cut exists. 

The following lemma shows that the randomization in the polynomial-time algorithm by Kratsch and Wahlstr\"{o}m can be avoided by the use of a single-exponential FPT algorithm, and that the cut covering set can be adapted to work for \emph{restricted} multiway cuts as long as we have a bound on their size.

\begin{restatable}{lemma}{lemCoveringCorollary}[$(\bigstar)$]
	\label{lem:covering:corollary}
	Let~$s \in \mathbb{N}$ be a constant. There is a \emph{deterministic} algorithm that, given an undirected $n$-vertex graph~$G$ and a set~$T \subseteq V(G)$ of terminals, runs in time~$2^{\Oh(|T|)} \cdot n^{\Oh(1)}$ and computes a set~$Z \subseteq V(G)$ with~$|Z| = \Oh(|T|^{2s+2})$ with the following guarantee: for each generalized $s$-partition~$\mathcal{T}$ of~$T$, if there is a restricted multiway cut for~$\mathcal{T}$ of size at most~$|T|$ in~$G$, then the set~$Z$ contains a \emph{minimum} restricted multiway cut of~$\mathcal{T}$.
\end{restatable}
\section{Odd Cycle Cuts} \label{sec:oct:cut}
In order to extend the ``antler'' framework of~\cite{DonkersJ24} to \PrbOCTLong (\PrbOCT), 
we define a problem-specific decomposition which we term \emph{\Occs} (\occs). Our decompositions have three parts --- a bipartite induced subgraph $X_B$, a vertex separator $X_C$ (which we call the \emph{head}), and a remainder $X_R$.  

% \begin{problem}{\PrbOCTLong (\PrbOCT)}
%     \Input & A graph $G$ and an integer $k$.\\
%     \Prob & Is there a set $S \subseteq V(G)$ of size at most $k$ such that
%     $G - S$ is bipartite?
% \end{problem}
% \bmpr{I don't think it's worth stating the decision version of the problem, since our algorithms don't work on the decision version of the problem. The introduction already defines what the corresponding optimization problem for OCT asks for, and the preliminaries also give a definition of what an odd cycle transversal is. Since the threshold~$k$ on solution size is never important for us, only membership in an optimal solution, I think we can omit this problem definition box to avoid confusion while at the same time saving space.}

\begin{definition}[\Occ]
    Given a graph $G$, a partition $(X_B, X_C, X_R)$ of $V(G)$ is an \emph{\Occ (\occ)}
    if (1) $G[X_B]$ is bipartite,
    (2) there are no edges between $X_B$ and $X_R$, and
    (3) $X_C \cup X_B \neq \emptyset$.
\end{definition}

% The condition (2) is equivalent to saying that $X_C$ separates $X_B$ from $X_R$.
% Also notice that from condition (1) $G[X_B]$ is necessarily bipartite.
We say $|X_C|$ is the \emph{width} of an \occ, and 
observe that $X_C$ hits all odd cycles in $G - X_R$.
% and an \occ is \emph{empty} if $X_B=X_C=\emptyset$.
%It is clear to see that $|X_C|$ upper-bounds the OCT solution size of $G[X_C \cup X_B]$.
We denote the minimum size of an OCT in $G$ by $\oct(G)$.

\begin{observation}
    If $(X_B, X_C, X_R)$ is an \occ in $G$, then $|X_C| \geq \oct(G[X_C \cup X_B])$.
\end{observation}

Analogous to $z$-antlers \cite{DonkersJ24}, here we define a \emph{\tocc}
as a special case of an \occ.
For a graph $G$, a set $X_C \subseteq V(G)$ and an integer $z$,
an \emph{$X_C$-certificate} of order $z$ is a subgraph~$H$ of~$G$
such that $X_C$ is an optimal OCT of $H$,
and for each component $H'$ of $H$ we have $|X_C \cap V(H')| \leq z$. Throughout the paper, and starting with the following definition, we will use the convention of referring to a \tocc as $(A_B, A_C, A_R)$ to emphasize its stronger guarantees compared to an arbitrary \occ~$(X_B, X_C, X_R)$.

\begin{definition}[($z$-)\tocc] \label{def:zocc}
    An \occ $(A_B,A_C,A_R)$ of a graph $G$ is \emph{tight}
    when $|A_C| = \oct(G[A_C \cup A_B])$.
    Furthermore, $(A_B, A_C, A_R)$ is a \tocc of order $z$ (equivalently, \emph{$z$-tight OCC})
    if $G[A_C \cup A_B]$ contains an $A_C$-certificate of order $z$.
\end{definition}

Note this definition naturally implies $\oct(G) = |A_C| + \oct(G[A_R])$: the union of~$A_C$ with a minimum OCT in~$G[A_R]$ forms an OCT for~$G$ (since~$A_C$ separates~$A_B$ from~$A_R$) for which the requirement~$|A_C| = \oct(G[A_C \cup A_B])$ guarantees optimality. The main result of this section is that assuming a graph $G$ has a \zocc, there exists a \zocc $(A_B,A_C,A_R)$
such that the number of components in $G[A_B]$ is bounded in terms of $z$ and $|A_C|$.
This is an extension of \cite[Lemma 4.6]{DonkersJ24}, and we defer its proof to \Cref{sec:proofs:oct:cut}.

\begin{restatable}{lemma}{restatenumcomponents}[$(\bigstar)$]
    \label{lem:num_components}
    Let $(A_B,A_C,A_R)$ be a \zocc in a graph $G$ for some $z \geq 0$.
    There exists a set $A_B' \subseteq A_B$ such that $(A_B',A_C, A_R \cup A_B \setminus A_B')$
    is a \zocc in $G$ and $G[A_B']$ has at most $z^2|A_C|$ components.
\end{restatable}

Finally, we introduce the notion of an \emph{imposed separation problem} whose solutions naturally correspond to odd cycle transversals of
specific subgraphs. 

\begin{definition} \label{def:imposed:separation:problem}
	Let $(X_B, X_C, X_R)$ be an OCC of $G$, and let $f_B \colon X_B \rightarrow \{0, 1\}$ be a proper $2$-coloring of $G[X_B]$. Let $C_1, C_2 \subseteq X_C$ be two disjoint subsets of $X_C$ and let $f_{C} \colon C_1 \rightarrow \{0, 1\}$ be a (not necessarily proper) $2$-coloring of the vertices in $C_1$. Based on this $4$-tuple of objects $(C_1, C_2, f_C, f_B)$, we define three (potentially overlapping) subsets $A, R, N \subseteq X_B$.
  	\begin{enumerate}
		\item Let $A$ be the set of vertices $v_b \in X_B$ with a neighbor $v_c \in C_1$ such that $f_B(v_b) = f_C(v_c)$.
		\item Let $R$ be the set of vertices $v_b \in X_B$ with a neighbor $v_c \in C_1$ such that $f_B(v_b) \neq f_C(v_c)$.
		\item Finally, let $N := N_G(C_2) \cap X_B$.
	\end{enumerate} 
	We refer to the problem of finding a smallest $\{A, R, N\}$-separator in $G[X_B]$ as \emph{the $\{A, R, N\}$-separation problem imposed onto $G[X_B]$ by $(C_1, C_2, f_C, f_B)$}.
\end{definition}

To see the connection between solutions and OCTs, one may let $C_1$ and $f_B$ in this definition correspond to $W$ and $c$ respectively in \cref{lem:AR-separation}, while the color classes of $f_C$ correspond to the sets $W_0$ and $W_1$ respectively. As shown below in \cref{lem:tight:occ:cut:optimal:ARN:separator}, we can recognize parts of tight OCCs as optimal solutions to specific imposed separation problems.

Although \cref{def:imposed:separation:problem} requires $f_B$ and $f_C$ to be colorings of $X_B$ and $C_1$ respectively, we sometimes abuse the notation by providing colorings whose domains are supersets of these intended domains. In these cases, one may interpret the definition of the imposed separation problem as if given the restrictions of these colorings to their respective intended domains.

One important role of these separation problems is to allow us to characterize intersections of two OCCs when at least one is tight. Specifically, in \cref{lem:tight:occ:cut:optimal:ARN:separator}, we show that the intersection of one OCC's head with the other OCC's bipartite part forms an optimal solution to a specific $3$-way separation problem, which is even optimal for a corresponding $2$-way problem.

\begin{restatable}{lemma}{lemOCCARNSeparators}[$(\bigstar)$] \label{lem:tight:occ:cut:optimal:ARN:separator}
	Let $(X_B, X_C, X_R)$ be a (not necessarily tight) OCC in the graph $G$ and let $(A_B, A_C, A_R)$ be a tight OCC in $G$. Let $f_X \colon X_B \rightarrow \{0, 1\}$ and $f_A \colon A_B \rightarrow \{0, 1\}$ be proper $2$-colorings of $G[X_B]$ and $G[A_B]$ respectively. Let $A$, $R$ and $N$ be the three sets to be separated in the separation problem imposed onto $G[X_B]$ by $(X_C \cap A_B, X_C \cap A_R, f_A, f_B)$ and let their names correspond to their roles as defined in \cref{def:imposed:separation:problem}. Then, $A_C \cap X_B$ is both a minimum-size $\{A, R\}$-separator and a minimum-size $\{A, R, N\}$-separator in $G[X_B]$.
\end{restatable}

This will prove to be a useful property in \Cref{sec:reduce:occ} by which we are able to recognize part of a tight OCC $(A_B, A_C, A_R)$ in an arbitrary graph. We complement it with the statement below, indicating that the intersection $A_C \cap X_B$ is even bounded in size.

\begin{lemma} \label{lem:small:occ:intersection}
	Let $(X_B, X_C, X_R)$ be a (not necessarily tight) OCC in the graph $G$ and let $(A_B, A_C, A_R)$ be a tight OCC in $G$. Then $|A_C \cap X_B| \leq |X_C|$.
\end{lemma}
\begin{proof}
	Suppose for contradiction that $|A_C \cap X_B| > |S|$. Then, $A_C' := (A_C \setminus X_B) \cup (X_C \cap (A_B \cup A_C))$ is a subset of $A_B \cup A_C$ that is strictly smaller than $A_C$. Now, showing that $A_C'$ is an OCT of $G[A_B \cup A_C]$ contradicts the assumption that $A_C$ is a smallest such OCT by virtue of $(A_B, A_C, A_R)$ being a tight OCC.
	
	To show that $A_C'$ is an OCT of $G[A_B \cup A_C]$, we let $F$ be an arbitrary odd cycle in this graph and show that it intersects $A_C'$. First, if $F$ intersects $X_C$, it intersects $A_C'$ in particular, since $X_C \cap (A_B \cup A_C) \subseteq A_C'$.
	
	Otherwise, since $X_C$ separates $X_B$ and $X_R$ in $G$, $F$ is completely contained in either $G[X_B]$ or $G[X_R]$. The former is not possible, since $G[X_B]$ is bipartite by assumption, so $F$ lives in $G[X_R]$. Furthermore, since $F$ was assumed to live in $G[A_B \cup A_C]$ and $G[A_B]$ is bipartite, $F$ intersects $A_C$. In particular, as we found $F$ to live in $G[X_R]$, it intersects $A_C \cap X_R$ which is a subset of $A_C'$ by construction. Hence, $F$ intersects $A_C'$ in any case.
\end{proof}

\section{Finding \Occs} \label{sec:find:occs}

Our ultimate goal is to show that if the graph contains any \tocc $(X_B,X_C,X_R)$ with $|X_C| \leq k$, 
then we can produce a \tocc with $|X_C| \leq k$ \emph{and} $|X_B|$ upper-bounded by some function of $k$.
To achieve this, we first show that we can efficiently find some \occ where $|X_B|$ is large enough, and then (in \Cref{sec:reduce:occ}) 
that we can reduce any such cut so that $|X_B|$ is small without destroying any essential structure of the input graph.

Specifically, we say an \occ $(X_B,X_C,X_R)$ is \emph{reducible} with respect to some function $g_r$ if $|X_B| > g_r(|X_C|)$. 
Our results all hold for a specific polynomial $g_r(x)$ in $\Theta(x^{16})$, which we specify in \Cref{sec:reduction:proofs}.
We say an \occ $(X_B,X_C,X_R)$ is a \emph{single-component} \occ if $G[X_B]$ is connected.

%-------------------------------------------------------------------------------
%    Color coding
%-------------------------------------------------------------------------------

Given a graph $G$, our goal is to output a reducible \occ efficiently
assuming that $G$ contains a single-component \occ $(X_B, X_C, X_R)$
with $|X_B| > g_r(2 |X_C|)$ and~$|X_C| \leq k$.
We achieve this by color coding of the vertices in $G$ (see definitions in \cref{sec:color:coding} for details).
Consider a coloring $\chi\colon V(G) \to \{\cB, \cC\}$.
For an integer $\ell$,
an \occ $(X_B, X_C, X_R)$ with $|X_B| \geq \ell$
is \emph{$\ell$-properly colored} by $\chi$
if $X_C \subseteq \chi^{-1}(\cC)$ and there is a set of~$\ell$ vertices of~$X_B$ that are colored $\cB$ and induce a connected subgraph of~$G$.
First, we show how to construct an \occ with large $X_B$ from a proper coloring.

\begin{restatable}{lemma}{lemFindColoredOCC}[$(\bigstar)$] \label{lem:finding-occ-colored}
  Given a graph $G$, integers $k, \ell$, and a coloring $\chi \colon V(G) \to \{\cB, \cC\}$ of $V(G)$ that $\ell$-properly colors
  a single-component \occ $(X_B, X_C, X_R)$ with $|X_C| \leq k$,
  an \occ $(X_B', X_C', X_R')$ such that $|X_B'| \geq \ell$ and $|X_C'| \leq 2k$
  can be found in polynomial time.
\end{restatable}

\begin{proof}[Proof sketch]
We iterate over the connected components of $G[\chi^{-1}(\cB)]$. 
For any component which is both large ($\geq \ell$) and bipartite, we try to find an \occ of small enough width
where the component is contained in the $X_B$ side of the cut. To do this, we use the machinery of 
bipartite separations introduced in Jansen et al.~\cite{JansenKW22} (see \Cref{sec:proofs:oct:finding} for details). 
Intuitively, given a vertex set $C$ which induces a connected bipartite subgraph, they either find a set of at most $2k$ vertices
which separates $C' \supseteq C$ from the remainder of the graph so that $G[C']$ is bipartite, or certify that $C$ is not part of $X_B$ for any \occ with width $\leq k$. 
\end{proof}

%\bmpr{Since bipartite separation finding is only relevant for this section (and even more particularly, only relevant for this lemma) I took it out of the general preliminaries and moved it to \cref{sec:bipartite:separation:finding}. I suggest deferring the formal proof below to a later section while writing a proof sketch that gives the main idea while also introducing the bipartite separation finding problem at an intuitive level, referring to the later section for the details.}

%-------------------------------------------------------------------------------
%    FPT algorithm
%-------------------------------------------------------------------------------

Now, we use this coloring scheme to find a reducible \occ,
assuming that a graph $G$ has a single-component \occ $(X_B, X_C, X_R)$ with large $X_B$.
%
%Using an $(n,k)$-universal set, we obtain the following lemma.

\begin{lemma}\label{lem:finding-occ}
  There exists a $2^{\Oh(k^{16})} n^{\Oh(1)}$-time algorithm that,
  given a graph $G$ and an integer $k$,
  either determines that $G$ does not contain a single-component \occ $(X_B,X_C,X_R)$ of width at most $k$ with $|X_B|>g_r(2k)$ or
  outputs a reducible \occ in $G$.
\end{lemma}

%\bmp{[When looking at it again, perhaps it suffices to do the following: delete the text 'using an $(n,k)$-universal set' from the fill-text just before the lemma (since it has not yet been defined), and in the proof replace 'From Theorem 27 .. ' by 'Using an~$(n,k)$-universal set, which is a well-known pseudorandom object [cite the sources mentioned in the additional prelims for them] used to derandomize applications of color coding (see Theorem 27), we can construct a family of (..) many subsets~$A_1, \ldots, A_{2^{\Oh(s)} \log n}$ with the guarantee that for each set~$S \subseteq V(G)$ of size~$s$, for each subset~$S'$ of~$S$, there exists a set in the family with~$A_i \cap S = S'$. From this family, we can construct a family of colorings that is guaranteed to include one that $\ell$-properly colors a suitable \occ~$(X_B, X_C, X_R)$ if one exists. To derive a coloring~$\chi_i$ from a member~$A_i \subseteq V(G)$ of the $(n,s)$-universal set, it suffices to pick~$\chi(a \in A) = \cR$ and~$\chi(a \notin A) = \cB$.' (The latter to address the fact that currently the proof forgets to translate between sets and colorings.) (Then continue with the rest of the proof.)]}

\begin{proof}
  We will invoke the algorithm from \Cref{lem:finding-occ-colored} multiple times for $\ell=g_r(2k) + 1$.
  If we supply a coloring that $\ell$-properly colors~$(X_B, X_C, X_R)$, then the algorithm is guaranteed to find an \occ $(X_B', X_C', X_R')$
  such that $|X_B'| > g_r(2k)$ and $|X_C'| \leq 2k$,
  which is reducible as $|X_B'| > g_r(2k) \geq g_r(|X_C'|)$.
  If all relevant colorings fail to find such a reducible \occ,
  then we can conclude that $G$ does not contain a single-component \occ $(X_B, X_C, X_R)$
  with $|X_C| \leq k$ and $|X_B| \geq \ell >g_r(2k)$.

  Let $X_B' \subseteq X_B$ be an arbitrary vertex set of size $\ell$ that induces a connected subgraph of $G$.
  Since $G[X_B]$ is connected, such $X_B'$ must exist.
  Observe that we obtain an $\ell$-proper coloring if $X_C \cup X_B'$ are colored correctly.
  Let $s = |X_C \cup X_B'| = k + g_r(2k) + 1 = \Oh(k^{16})$.

  Using an~$(n,k)$-universal set, which is a well-known pseudorandom object~\cite{NaorSS95, CyganFKLMPPS15} used to derandomize applications of color coding (see \Cref{thm:universal-set}), we can construct a family of $2^{\Oh(s)}\log n$ many subsets~$A_1, \ldots, A_{2^{\Oh(s)} \log n}$ with the guarantee that for each set~$S \subseteq V(G)$ of size~$s$, for each subset~$S'$ of~$S$, there exists a set in the family with~$A_i \cap S = S'$. This can be done in $2^{\Oh(s)} n \log n = 2^{\Oh(k^{16})} n \log n$ time. From this family, we can construct a family of colorings that is guaranteed to include one that $\ell$-properly colors a suitable \occ~$(X_B, X_C, X_R)$ if one exists. To derive a coloring~$\chi_i$ from a member~$A_i \subseteq V(G)$ of the $(n,s)$-universal set, it suffices to pick~$\chi(a \in A) = \cR$ and~$\chi(a \notin A) = \cB$.

%  From \Cref{thm:universal-set}, an $(n,s)$-universal set for $V(G)$ of size $2^{\Oh(s)}\log n$
 % can be constructed in $2^{\Oh(s)} n \log n = 2^{\Oh(k^{16})} n \log n$ time.
  %
  We run the $n^{\Oh(1)}$-time algorithm from \Cref{lem:finding-occ-colored} for each coloring,
  which results in the overall runtime $2^{\Oh(k^{16})} n^{\Oh(1)}$.
\end{proof}

\section{Reducing Odd Cycle Cuts} \label{sec:reduce:occ}

Given an OCC $(X_B, X_C, X_R)$ of $G$ with $|X_B| > g_r(|X_C|)$,
the next step is to ``shrink'' $X_B$ in a way that preserves some of the structure of the input graph. In this section, we give a reduction to do this and prove that it preserves the general structure of minimum-size OCTs and of tight OCCs in the graph. The reduction starts with a marking scheme that is discussed separately in \cref{ssec:shrinking:step:1}. We give the full reduction, which includes this marking scheme as a subroutine, in \cref{ssec:shrinking:step:2}. The reduction will only affect $G[X_B]$ and the edge set between $X_B$ and $X_C$, which already ensures that an important part of the input graph is maintained.

\subsection{A marking scheme for the reduction} \label{ssec:shrinking:step:1}
The goal of the marking scheme is to mark a set $B^\ast \subseteq X_B$ of size $|X_C|^{\Oh(1)}$ as ``interesting'' vertices that the reduction should not remove or modify. Intuitively, we want this set to contain vertices which we expect might be part of the cut part of a tight OCC in $G$. More precisely, we guarantee that for every tight OCC in $G$ there is a (possibly different) tight OCC $(A_B, A_C, A_R)$ such that $A_C \cap X_B$ is contained in the marked set $B^\ast$.

As seen in \cref{lem:tight:occ:cut:optimal:ARN:separator}, for every tight OCC $(A_B, A_C, A_R)$ in the graph, the intersection $A_C \cap X_B$ forms an optimal solution to a specific imposed separation problem (\cref{def:imposed:separation:problem}). As such, it suffices if $B^\ast$ is a \emph{cut covering set} for these imposed separation problems.

Indeed, the key ingredient of the algorithm presented below is the computation of such a cut covering set. Preceding this computation is a graph reduction ensuring that the computed set covers precisely the imposed separation problems. In \cref{lem:covering:corollary}, we will show that a cut covering set can be computed in deterministic FPT time parameterized by the size of the terminal set, which leads to a total running time of $2^{\Oh(|X_C|)} n^{\Oh(1)}$ time for the marking step. % \refmark. 

\begin{marking} \label{algo:mark:B:star} 
	Consider the following algorithm.

	Input: A graph $G$ and an OCC $(X_B, X_C, X_R)$ of $G$. 
	%\rfar{Does the algorithm also work when the input OCC is not required to have $|X_B| > g_r(|X_C|)$? If yes: we could consider dropping this requirement from the input (in which case, let me know so I can update this accordingly through the remainder of \cref{sec:reduce:occ}). If no: would it be simpler to call the input OCC 'reducible' as opposed to repeating the constraint $|X_B| > g_r(|X_C|)$?}

	Output: Marked vertices $B^\ast \subseteq X_B$.
	\begin{enumerate}
		\item \label{itm:mark:coloring} Find a proper 2-coloring $f_X \colon X_B \to \{0,1\}$ of $G[X_B]$.
		\item \label{itm:mark:auxiliary} Construct an auxiliary (undirected) graph $G'$, initialized to a copy of $G[X_B]$. For each $v \in X_C$, do the following:
		\begin{itemize}
				\item Add vertices $v^{(0)}$ and $v^{(1)}$ to $G'$.
				\item For each neighbor $u \in N_G(v) \cap X_B$, add an edge $v^{(f_X(u))}u$.
		\end{itemize}
		Let $T$ be the set $\{v^{(i)} \mid v \in X_C, i \in \{0,1\}\}$.
		Note that $|T|=2|X_C|$.
		%\item Using \Cref{lem:covering:corollary},
		%find a set $B^\ast \subseteq X_B$ of size at most $c \cdot |T|^8 = c \cdot (2|X_C|)^8$
		%such that for every generalized 3-partition $\mathcal{T}^*$ of $T$,\bmp{[We currently do not have the terminology of 'generalized 3-partition' in scope, nor the constants~$c$. To handle the latter, I suggest simply omitting the size bound from this statement (its analysis can come later). We should address the former, probably by giving an in-line definition of what it means to be a restricted multiway cut for a generalized partition (to be later expanded by the full system of definitions in the later sections).]}
		%if there is a restricted multiway cut for $\mathcal{T}$ of size at most $|T|$ in $G'$,
		%then the set $B^\ast$ contains a minimum restricted multiway cut of $\mathcal{T}$.
		\item \label{itm:mark:cut:covering} Compute a cut covering set~$B^\ast \subseteq V(G')$ via \Cref{lem:covering:corollary} such that for every partition~$\mathcal{T}^* = (T_0, T_1, T_2, T_3, T_X)$ of~$T$, the set~$B^*$ contains a minimum-size solution to the following problem: 
		\begin{itemize}
			\item find a vertex set~$S \subseteq V(G') \setminus (T_1 \cup T_2 \cup T_3)$ such that~$S$ separates~$T_i$ and~$T_j$ in the graph~$G' - T_X$ for all~$1 \leq i < j \leq 3$, 
		\end{itemize}
		as long as this problem has a solution of size at most~$|T|$.
	\end{enumerate}
\end{marking}
%
%The lemma below formalizes the intuitive goal that we set out for $B^\ast$ to encompass.
%
\begin{restatable}{lemma}{lemMarkedTightOcc}[$(\bigstar)$] \label{lem:marked:tight:occ}
	Let $B^\ast$ be constructed as in \refmark when given the graph $G$ and an OCC $(X_B, X_C, X_R)$ of $G$ as input. If there exists a $z$-tight OCC $(A_B, A_C, A_R)$ in $G$, then there exists a $z$-tight OCC $(A_B^\ast, A_C^\ast, A_R^\ast)$ in $G$ with $|A_C^\ast| = |A_C|$ and with $A_C^\ast \cap X_B \subseteq B^\ast$.
\end{restatable}
\begin{proof}[Proof sketch]
	Let $f_X \colon X_B \rightarrow \{0,1\}$ be the $2$-coloring obtained in step~\ref{itm:mark:coloring} of \refmark, let $f_A \colon A_B \rightarrow \{0,1\}$ be a proper $2$-coloring of $G[A_B]$ and consider the separation problem imposed onto $G[X_B]$ by ${(X_C \cap A_B, X_C \cap A_R, f_A, f_B)}$. Let $A$, $R$ and $N$ be the three sets to be separated in this problem with their names corresponding to their roles as in \cref{def:imposed:separation:problem}.
	
	\rfar{This paragraph represents my most compact attempt at capturing the correspondence between separation problems in the original and reduced graph.}By putting the correct copy of each vertex from $A_B \cap X_C$ into $T_1$ and $T_2$ respectively, putting both copies of vertices from $A_R \cap X_C$ into $T_3$ and putting both copies of vertices from $A_C \cap X_C$ into $T_X$, we obtain a partition $(\emptyset, T_1, T_2, T_3, T_X)$ of the set $T$ defined in step~\ref{itm:mark:auxiliary}, such that the corresponding separation problem has the same solution space as the $\{A, R, N\}$-separation problem imposed onto $G[X_B]$. By construction of $B^\ast$ in step~\ref{itm:mark:cut:covering}, there is a set $S \subseteq B^\ast$ (possibly different from $A_C \cap X_B$) that is an optimal $\{A, R, N\}$-separator in $G[X_B]$. To construct the tight OCC $(A_B^\ast, A_C^\ast, A_R^\ast)$, we use this set $S$ as replacement for $A_C \cap X_B$, which is also a minimum-size $\{A, R, N\}$-separator in $G[X_B]$ by \cref{lem:tight:occ:cut:optimal:ARN:separator}.

	As such, we define $A_C^\ast := (A_C \setminus X_B) \cup S$. To define $A_B^\ast$, let $U$ be the set of vertices from $X_B \setminus S$ that are not reachable from $N$ in $G[X_B] - S$. Now, we define $A_B^\ast := (A_B \setminus X_B) \cup U$. Finally, we define $A_R^\ast := V(G) \setminus (A_B^\ast \cup A_C^\ast)$. Clearly, this $3$-partition of $V(G)$ satisfies the constraints $|A_C^\ast| = |A_C|$ and $A_C^\ast \cap X_B \subseteq B^\ast$. We proceed by showing that it satisfies the three additional properties required to be a $z$-tight OCC.
	
	First, to see that $G[A_B^\ast]$ is bipartite, we note that $A_B^\ast$ only contains vertices from $A_B$ and $X_B \setminus S$. Both are vertex sets that induce a bipartite subgraph. Then, noting the correspondence between the sets $A$ and $R$ obtained from the separation problem and the sets $A$ and $R$ as in \cref{lem:AR-separation}, we invoke this lemma on $G[(X_C \cap A_B) \cup X_B]$ with $c = f_A$ and with $W_0$ and $W_1$ being the two color classes of this coloring restricted to $X_C \cap A_B$. It follows that the vertices from $A_B^\ast$ in $X_B \setminus S$ can be properly $2$-colored by a coloring $f$ that agrees with $f_A$ on the vertex set $X_C \cap A_B$ that separates $A_B^\ast \cap X_R$ and $A_B^\ast \cap X_B$. As these two vertex sets are properly colored by $f_A$ and $f$ respectively, these colorings combine to a proper $2$-coloring of the entire graph $G[A_B^\ast]$ (see \cref{lem:separated:colorings:make:bipartite}).
	
	Secondly, a case distinction shows that there are no edges between $A_B^\ast$ and $A_R^\ast$. It combines the fact that $A_C \cap X_B$ is an $\{A, R, N\}$-separator in $G[X_B]$ --- thereby in particular separating $A_B \cap X_B$ from $N$ in $G[X_B]$ --- and the fact that $A_B^\ast$ only contains vertices that already belonged to $A_B$ and vertices from $X_B$ that are not reachable from $N$ in $G[X_B] - A_C^\ast$.
	
	Finally, it remains to show that $(A_B^\ast, A_C^\ast, A_R^\ast)$ has an  $A_C^\ast$-certificate of order $z$. To prove this, we show that the order-$z$ certificate $D$ of the original OCC $(A_B, A_C, A_R)$ is also an order-$z$ certificate in $(A_B^\ast, A_C^\ast, A_R^\ast)$. The main effort here is to prove that $D$ even lives in $A_B^\ast \cup A_C^\ast$, after which it is easy to see that it is also an order-$z$ certificate for our new OCC.
	
 	As \cref{lem:tight:occ:cut:optimal:ARN:separator} guarantees that $A_C \cap X_B$ is not only an optimal $\{A, R, N\}$-separator in $G[X_B]$ but even an optimal $\{A, R\}$-separator in this graph, it contains exactly one vertex from every path of a maximum packing $\mathcal{P}$ of pairwise vertex-disjoint $(A, R)$-paths in $G[X_B]$, due to Menger's theorem~\cite[Theorem~9.1]{Schrijver03}. Likewise, $A_C^\ast \cap X_B = S$ is also an optimal $\{A, R\}$-separator in $G[X_B]$ and hence also contains exactly one vertex from every path of $\mathcal{P}$.
 	
 	Intuitively, for any path $P \in \mathcal{P}$, a vertex on this path that stops being reachable from one endpoint of $P$ when sliding the picked vertex along the path, starts becoming reachable from the other endpoint of $P$. As both endpoints of $P$ belong to $A \cup R$ and $S$ only differs from $A_C \cap X_B$ by which vertex is picked from each path in $\mathcal{P}$, it cannot drastically alter which vertices are reachable from $A \cup R$, which in turn are all vertices that end up in $A_B^\ast$.
 	
 	Using the observation that~$A_B$ and~$A_B^\ast$ are separated from~$N$ by~$A_C$ and~$A_C^\ast$ respectively, we see that all vertices that \emph{are} disconnected from $A \cup R$ by substituting $A_C \cap X_B$ for $S$ are in particular also disconnected from $N$. Thereby, these vertices end up in $A_B^\ast$. This shows that $(A_B \cup A_C) \subseteq (A_B^\ast \cup A_C^\ast)$, which implies that the certificate $D$ also lives in the latter. 
\end{proof}

\subsection{Simplifying the graph} \label{ssec:shrinking:step:2}
Our eventual reduction starts with the Marking step from the previous section, after which the graph is modified in a way that leaves marked vertices untouched. We want the reduction to preserve the general structure of optimal OCTs and tight OCCs in the input graph. As this is governed by the locations and interactions of odd cycles in the graph, we encode this information in a more space-efficient manner using the following reduction.

\begin{reduction} \label{algo:shrink}
	Given a graph $G$ and an OCC $(X_B, X_C, X_R)$ of it, we construct a graph $G'$ as follows.
	\begin{enumerate}
		\item \label{itm:shrink:mark:b:star} Use \refmark with input $G$ and $(X_B, X_C, X_R)$ to obtain the set $B^\ast \subseteq X_B$.
		\item Initialize $G'$ as a copy of $G - (X_B \setminus B^\ast)$.
		\item \label{itm:shrink:loop} For every $u, v \in X_C \cup B^\ast$ and for every parity $p \in \{\text{even}, \text{odd}\}$, check if the subgraph $G[X_B \setminus B^\ast]$ contains the internal vertex of a $(u,v)$-path with parity $p$. If so, then:
		\begin{itemize}
			\item if $p = \text{even}$, add two new vertices $x$ and $x'$ to $G$ and connect both of them to $u$ and $v$.
			\item if $p = \text{odd}$, add four new vertices $x$, $y$, $x'$ and $y'$ to $G$ and add the edges $\{u, x\}$, $\{x, y\}$, $\{y, v\}$, $\{u, x'\}$, $\{x', y'\}$ and $\{y', v\}$.
		\end{itemize}
		Note that we explicitly allow $u = v$ in this step. 
	\end{enumerate}
\end{reduction}
Effectively, this reduction deletes the vertices~$X_B \setminus B^\ast$ from the graph. For each pair of neighbors~$u,v$ from that set, if the deleted vertices provided an odd (resp.~even) path between them, then we insert two vertex-disjoint odd (resp.~even) paths between~$u$ and~$v$. Hence we shrink the graph while preserving the parity of paths provided by the removed vertices.

As we will prove in \cref{lem:shrink:running:time} and \cref{lem:shrink:size} respectively, the reduction can be performed in $2^{\Oh(|X_C|)} \cdot n^{\Oh(1)}$ time and it is guaranteed to output a strictly smaller graph than its input graph whenever it receives an OCC that is reducible with respect to the function $g_r$ as in \cref{sec:find:occs}. To show that the reduction also preserves OCT and OCC structures, we prove that it satisfies two safety properties formalized below in \cref{lem:forward:safety,lem:backward:safety}.

\begin{restatable}{lemma}{lemForwardSafety}[$(\bigstar)$] \label{lem:forward:safety}
	Let $G$ be a graph, let $(X_B, X_C, X_R)$ be an OCC in $G$ and let $G'$ be the graph obtained by running \refreduction with these input parameters. For all $z \geq 0$, if there exists a $z$-tight OCC $(A_B, A_C, A_R)$ in $G$, then there exists a $z$-tight OCC $(A_B', A_C', A_R')$ in $G'$ with $|A_C'| = |A_C|$.
\end{restatable}

The proof of the lemma above uses \cref{lem:marked:tight:occ} to infer that, for any $z$-tight OCC $(A_B, A_C, A_R)$ of $G$, the graph~$G$ also contains a $z$-tight OCC $(A_B^\ast, A_C^\ast, A_R^\ast)$ of the same width such that $A_C^\ast \subseteq V(G) \cap V(G')$. This allows for the construction of an OCC $(A_B', A_C', A_R')$ in $G'$ with $A_C' = A_C^\ast$. Then, $A_B'$ can be defined as the union of $A_B^\ast \cap V(G) \cap V(G')$ and the set of vertices that were added during the reduction to provide a replacement connection between any two vertices from $A_C^\ast \cup (A_B^\ast \cap V(G) \cap V(G'))$. Finally, $A_R' := V(G') \setminus (A_B' \cup A_C')$.

The proof proceeds to show that the resulting partition $(A_B', A_C', A_R')$ is a $z$-tight OCC of $G'$. The two main insights used to prove this are the facts that: 
\begin{itemize}
	\item optimal OCTs of $G'$ are disjoint from the set of newly added vertices $V(G') \setminus V(G)$, and
	\item odd cycles in $G$ can be translated to very similar odd cycles in $G'$ and vice versa.
\end{itemize}
These insights are also covered in the proof sketch of the second safety property below.

\begin{restatable}{lemma}{lemBackwardSafety}[$(\bigstar)$]
\label{lem:backward:safety}
	Let $G$ be a graph, let $(X_B, X_C, X_R)$ be an OCC in $G$ and let $G'$ be the graph obtained by running \refreduction with these input parameters. If $S'$ is a minimum-size OCT of $G'$, then $S' \subseteq V(G) \cap V(G')$ and $S'$ is a minimum-size OCT of $G$.
\end{restatable}
\begin{proof}[Proof sketch]
	To see that $S' \subseteq V(G) \cap V(G')$, we show that $S'$ contains none of the newly added vertices in $V(G') \setminus V(G)$. These newly added vertices come in pairs that form degree-$2$ paths connecting the same endpoints. Consider two such paths and let $u$ and $v$ be the endpoints of both of them. Suppose for contradiction that $S'$ uses an internal vertex $p_1$ from one path to break an odd cycle $F$. Then it must also contain an internal vertex $p_2$ from the other path to break the odd cycle obtained by swapping one path for the other in $F$. As both these cycles also pass through $u$ and $v$ by construction, substituting $p_1$ and $p_2$ for one of $u$ and $v$ in $S'$ yields a strictly smaller solution. This contradicts the optimality of $S'$.
	
	To see that $S'$ is an OCT of $G$, suppose for contradiction that $G - S'$ contains an odd cycle $F$. Every subpath of $F$ that connects two vertices from $V(G) \cap V(G')$ via a path whose internal vertices lie in $G - V(G')$ can be replaced by one of the paths inserted during the construction of~$G'$, with the same endpoints and parity. Substituting every subpath of $F$ that is absent in $G'$ for such a replacement path yields a closed odd walk in $G' - S'$; but this contradicts the fact that $S'$ is an OCT of~$G'$. Hence~$S' \subseteq V(G) \cap V(G')$ is an OCT in~$G$.
	
	It remains to show that $S'$ is an OCT of minimum size. Suppose for contradiction that $T$ is a strictly smaller OCT of $G$. We start by showing how to modify $T$ into an OCT $S$ of $G$ that is at most as large and lives in $V(G) \cap V(G')$. 
	To this end, let $f \colon V(G) \setminus T \rightarrow \{0,1\}$ and $f_X \colon X_B \rightarrow \{0,1\}$ be proper $2$-colorings of $G - T$ and $G[X_B]$ respectively and consider the separation problem imposed onto $G[X_B]$ by $(X_C \setminus T, \emptyset, f, f_B)$. Let $A$, $R$, $N$ be the three sets to be separated in this problem with their names corresponding to their roles as in \cref{def:imposed:separation:problem}. Since the second argument~$C_2$ in the 4-tuple is~$\emptyset$, we obtain~$N = \emptyset$.
	
	The fact that~$N = \emptyset$ ensures that the separation problem above is merely a $2$-way separation problem between the sets $A$ and $R$ in~$G[X_B]$. These sets are defined in such a way that, for a suitable choice of input parameters to \cref{lem:AR-separation}, they coincide with the sets $A$ and $R$ in this lemma. Applying the lemma in one direction to $G[(X_C \setminus T) \cup X_B]$ with $c = f$ and with $W_0$ and $W_1$ being the two color classes of this coloring restricted to $X_C \setminus T$, yields that $T \cap X_B$ is an $\{A, R\}$-separator in $G[X_B]$. 
	Applying it in the other direction yields that the removal of \emph{any} $\{A, R\}$-separator $T'$ from $G$ allows for a proper $2$-coloring $f'$ of $G[(X_C \setminus T) \cup (X_B \setminus T')]$ that agrees with the coloring $f$ on the vertex set $X_C \setminus T$. As this set separates the subgraphs $G[X_R \setminus T]$ and $G[X_B \setminus T']$ in $G - ((T \setminus X_B) \cup T')$ and these subgraphs are properly $2$-colored by $f$ and $f'$ respectively, those two colorings combine to properly $2$-color $G - ((T \setminus X_B) \cup T')$ (see \cref{lem:separated:colorings:make:bipartite}). By construction of $B^\ast$, there is a minimum-size $\{A, R\}$-separator $T^\ast$ in $G[X_B]$ with~$T^\ast \subseteq B^\ast$. Hence, $S := (T \setminus X_B) \cup T^\ast$ is an OCT of $G$ that lives in $V(G) \cap V(G')$. Furthermore, since $T^\ast$ is a minimum-size $\{A, R\}$-separator in $G[X_B]$ and it replaces $T \cap X_B$, which is also an $\{A, R\}$-separator, we find that $|S| \leq |T| < |S'|$.
	
	Since $S'$ was assumed to be a minimum-size OCT of $G'$, the smaller set $S$ is not an OCT of $G'$. Therefore, $G' - S$ contains an odd cycle $F'$. The argument used before to convert an odd cycle in $G$ to one in $G'$ can also be used in the reverse direction to construct an odd cycle $F$ in $G - S$ from $F'$. The existence of this cycle contradicts the assumption that $S$ is an OCT of $G$, which concludes the proof by showing that $S'$ is a minimum-size OCT of $G$.
\end{proof}

\section{Finding and Removing Tight OCCs} \label{sec:tight:occs}

%-------------------------------------------------------------------------------
%    Color coding
%-------------------------------------------------------------------------------

Now we find \octmane{s} by the same color coding technique used in previous work~\cite{DonkersJ24}.
Consider a coloring 
$\chi\colon V(G) \cup E(G) \to \{\cB, \cC, \cR\}$ of the vertices and edges of a graph~$G$.
For every color $c \in \{\cB, \cC, \cR\}$, let $\chi_V^{-1}(c) = \chi^{-1}(c) \cap V(G)$.
For any integer $z \geq 0$, a \zocc $(A_B,A_C,A_R)$ is \emph{$z$-properly colored} by
a coloring $\chi$ if all the following hold:
(i)~$A_C \subseteq \chi_V^{-1}(\cC)$, (ii)~$A_B \subseteq \chi_V^{-1}(\cB)$, and (iii)~for each component $H$ of $G'=G[A_B \cup A_C] - \chi^{-1}(\cR)$ we have $\oct(H)=|A_C \cap V(H)|$ and~$|A_C \cap V(H)| \leq z$.
\bmpr{I switched the enumeration to inline, but it would be nice to turn it back into line-by-line for the final.}
%
%\begin{enumerate}[(i)]
	%\item $A_C \subseteq \chi_V^{-1}(\cC)$,
	%\item $A_B \subseteq \chi_V^{-1}(\cB)$, and
	%\item for each component $H$ of $G'=G[A_B \cup A_C] - \chi^{-1}(\cR)$
	%we have $\oct(H)=|A_C \cap V(H)|$ and~$|A_C \cap V(H)| \leq z$.
%\end{enumerate}
Note that $\chi^{-1}(\cR)$ may include both vertices and edges, so that the process of obtaining~$G'$ involves removing both the vertices and edges colored~$\cR$.
%As we show in the proof of \cref{thm:main}, if $(A_B, A_C, A_R)$ is a \zoctmane, then there exists a coloring that $z$-properly colors it. 
By a straight-forward adaptation of the color coding approach from previous work~\cite[Lemma 6.2]{DonkersJ24}, we can reconstruct a \tocc from a proper coloring.

\begin{restatable}{lemma}{lemFindAntlerColored}[$(\bigstar)$]
\label{lem:find-antler-colored}
    There is an $n^{\Oh(z)}$ time algorithm taking as input an integer $z \geq 0$,
    a graph $G$, and a coloring $\chi \colon V(G) \cup E(G) \to \{\cB, \cC, \cR\}$
    that either determines that~$\chi$ does not $z$-properly color any $z$-tight OCC,
    or outputs a $z$-tight \occ $(A_B,A_C,A_R)$ in $G$ such that
    for each OCC $(\hat{A}_B, \hat{A}_C, \hat{A}_R)$ that is $z$-properly colored by~$\chi$,
    we have $\hat{A}_B \subseteq A_B$ and $\hat{A}_C \subseteq A_C$.
\end{restatable}

%Using this algorithm combined with $(n,k)$-universal sets for generating colorings, we can extend \cite[Lemma 6.2]{DonkersJ24} as follows.

Combining all ingredients in the previous sections leads to a proof of the main theorem.

\thmMain*
\begin{proof}[Proof sketch]
Given an input graph~$G$, we repeatedly invoke \cref{lem:finding-occ} to find a reducible \occ and use \refreduction to shrink it. When we stabilize on a graph~$G'$, \cref{lem:forward:safety} guarantees that~$G'$ contains a \zocc of width~$|A'_C| = k$ if~$G$ had one. By \cref{lem:num_components}, there is such a \zocc~$(A'_B, A'_C, A'_R)$ in~$G'$ for which~$G'[A'_B]$ has at most~$z^2 k$ components. As each such component gives rise to a single-component \occ, none of them are large enough to be reducible. Hence~$|A'_C \cup A'_B| \in (zk)^{\Oh(1)}$. Hence we can deterministically construct a family of~$2^{(kz)^{\Oh(1)}} n^{\Oh(1)}$ colorings that includes one that properly colors~$(A'_B, A'_C, A'_R)$. Invoking \cref{lem:find-antler-colored} with such a coloring identifies a \zocc in~$G'$ whose head~$A^*_C$ contains~$A'_C$ and therefore has size at least~$k$. Then~$A^*_C$ is contained in an optimal OCT in~$G'$, so that \cref{lem:backward:safety} ensures~$A^*_C$ belongs to an optimal OCT in~$G$. We output~$A^*_C$.
\end{proof}

\section{Additional preliminaries} \label{sec:more:prelims}
In this section we collect some additional preliminaries that are necessary for the upcoming proofs, but which were not relevant for the first part of the paper. 

\subsection{Multiway cuts}
The parameterized problem of computing a restricted multiway cut is well-known to be fixed-parameter tractable parameterized by solution size~\cite{CyganPPW13}.

\begin{theorem}[{\cite[Corollary 2.6]{CyganPPW13}}] \label{thm:multiway:marx}
There is an algorithm that, given an undirected $n$-vertex graph~$G$, terminal set~$T \subseteq V(G)$, and integer~$k$, runs in time~$2^{\Oh(k)} \cdot n^{\Oh(1)}$ and either outputs the minimum size of a vertex set~$U \subseteq V(G) \setminus T$ for which all vertices of~$T$ belong to a different connected component of~$G - U$, or determines that any such multiway cut has size larger than~$k$. 
\end{theorem}

While the original formulation of the theorem is stated in terms of simply solving the decision problem (is there a solution of size at most~$k$?), by running this algorithm for all possible~$k'$ in the range of~$\{0, \ldots, k\}$ we can identify the size of an optimal solution with only a polynomial factor overhead. Next, we show that with a small modification we can use the same algorithm to compute the size of a minimum multiway cut of an $s$-partition~$\mathcal{T} = (T_1, \ldots, T_s)$ of the terminal set~$T$ (or determine its size is larger than~$k$). To express the necessary modification on the graph, we need the following concept. The operation of \emph{identifying} a vertex set~$X$ in an undirected graph~$G$ yields the graph~$G'$ obtained from~$G$ by removing~$X$ while inserting a single new vertex~$x$ with~$N_{G'}(x) = N_G(X)$.

\begin{observation} \label{obs:mwc:partition:or:set}
Let~$G$ be an undirected graph with terminal set~$T$ and let~$\mathcal{T} = (T_1, \ldots, T_s)$ be an $s$-partition of~$T$. The following two conditions are equivalent for each vertex set~$U \subseteq V(G) \setminus T$:
\begin{itemize}
	\item The set~$U$ is a multiway cut of the partition~$\mathcal{T}$ in~$G$;
	\item No connected component of~$G' - U$ contains more than one vertex of~$T'$, where~$G'$ is the graph obtained from~$G$ by identifying each non-empty vertex set~$T_i$ for~$i \in [s]$ into a single vertex~$t_i$, and~$T'$ is the set containing the vertices~$t_i$ resulting from these identifications. (The order in which the sets are identified does not affect the final result.)
\end{itemize}
\end{observation}

The forward implication is trivial; the converse follows from the fact that any path between two distinct vertices of~$T'$ contains a subpath connecting vertices belonging to two different sets of the partition~$\mathcal{T}$, which must therefore be broken by any multiway cut for the partition~$\mathcal{T}$ in~$G$.

\begin{lemma} \label{lem:multiway:alg}
There is an algorithm that, given an undirected $n$-vertex graph~$G$, an integer~$s$, an $s$-partition~$\mathcal{T} = (T_1, \ldots, T_s)$ of a terminal set~$T \subseteq V(G)$, and integer~$k$, runs in time~$2^{\Oh(k)} \cdot n^{\Oh(1)}$ and either outputs the minimum size of a restricted multiway cut for~$\mathcal{T}$, or outputs~$\bot$ to indicate that any such restricted multiway cut has size larger than~$k$. 
\end{lemma}
\begin{proof}
Apply the algorithm of \cref{thm:multiway:marx} to the graph~$G'$ and terminal set~$T'$ obtained by identifying each set~$T_i$ into a new terminal vertex~$t_i \in T'$, with a budget of~$k$. Correctness follows from \cref{obs:mwc:partition:or:set}.
\end{proof}

We will use \cref{lem:multiway:alg} to construct a \emph{deterministic} FPT-algorithm to compute a cut covering set later. For a vertex set~$U$ in an undirected graph~$G$, the operation of \emph{bypassing} vertex set~$U$ results in the graph~$\bypass_U(G)$ that is obtained as follows: starting from~$G - U$, add an edge $uv$ between every pair of distinct vertices $u,v$ which appear together in the same connected component of $G[U]$? Hence $\bypass_U(G)$ is a graph on vertex set~$V(G) \setminus U$ that has an edge~$uv$ whenever~$uv \in E(G)$ or there is a path from~$u$ to~$v$ for which all internal vertices belong to~$U$.
	
	\begin{observation} \label{obs:multiway:bypass}
		Let~$G$ be an undirected graph with terminal set~$T \subseteq V(G)$, and let~$(T_1, \ldots, T_s)$ be a partition of~$T$. For any set~$R \subseteq V(G) \setminus T$ of bypassed vertices and any vertex set~$U \subseteq V(G) \setminus R$, the set~$U$ is a restricted multiway cut of~$(T_1, \ldots, T_s)$ in~$G$ if and only if~$U$ is a restricted multiway cut of~$(T_1, \ldots, T_s)$ in~$\bypass_R(G)$.
	\end{observation}
	
	This observation implies that bypassing a set of non-terminal vertices can only make it \emph{harder} to form a multiway cut. 
	
\subsection{Covering multiway cuts for generalized partitions} \label{sec:coveringg:multiway}
\lemCoveringCorollary*
\begin{proof}
% \yo{Can we say $Z \subseteq V(G) \setminus T$?}\bmpr{At this level of generality, no we cannot say that~$Z \subseteq V(G) \setminus T$. A generalized $s$-partition may assign vertices of~$T$ to the~$T_0$ part of the partition, so they are 'in the graph and available to be deleted by the multiway cut', and they may in fact be required to make an optimal multiway cut. But when you use a generalized $s$-partition with~$T_0 = \emptyset$ then the multiway cut that you get will be a subset of~$V(G) \setminus T$ since it is restricted.}
%
We present the algorithm first, which will be followed by its correctness proof and analysis. Given the graph~$G$ with terminal set~$T$ as input, the algorithm proceeds as follows.

\begin{enumerate}
	\item Number the vertices of~$V(G) \setminus T$ as~$v_1, \ldots, v_\ell$. Initialize a set~$R$ of \emph{redundant} vertices, initially empty.
	\item For each~$i \in [\ell]$, do the following:
	\begin{itemize}
		\item Initialize a Boolean variable~$E[i]$ to keep track of whether~$v_i$ is essential, initially $\mathsf{false}$.
		\item For each of the~$(s+2)^{|T|}$ distinct generalized $s$-partitions~$\mathcal{T}$ of the set~$T$, do the following.
		\begin{enumerate}
				\item Let~$X_{\mathcal{T}} \in \mathbb{N} \cup \{\bot\}$ be the result of invoking the algorithm from \cref{lem:multiway:alg} with parameter value~$k := |T| + 1$ on the graph~$\bypass_R(G - T_X)$ with partition~$\mathcal{T}$. 
				\item Let~$Y_{\mathcal{T}} \in \mathbb{N} \cup \{\bot\}$ be the result of invoking the algorithm from \cref{lem:multiway:alg} with parameter value~$k := |T| + 1$ on the graph~$\bypass_{R \cup \{v_i\}}(G - T_X)$ with partition~$\mathcal{T}$. 
				\item If~$X_\mathcal{T} \neq Y_{\mathcal{T}}$, then set~$E[i]$ to $\mathsf{true}$. We say that~$\mathcal{T}$ is a \emph{witness} to the fact that~$v_i$ is essential.
		\end{enumerate}
		\item If~$E[i]$ is still $\mathsf{false}$ after having considered all~$|T|^{s+2}$ generalized $s$-partitions of~$T$: add~$v_i$ to~$R$.
	\end{itemize}
	\item Let~$R^*$ denote the contents of the set~$R$ after completing the previous loop. 
	\item Output~$Z := V(G) \setminus R^*$ as the cut covering set. (The vertices of~$T$ remain in the cut covering set despite the set covering restricted multiway cuts; their presence in a cut may be needed when they are placed in the ``free''~$T_0$ part of a generalized $s$-partition.)
\end{enumerate}

In summary, the algorithm tries the non-terminal vertices one at a time, testing for each of them whether there is a generalized $s$-partition that is affected by bypassing the vertex. We note the following important property of the algorithm above: for each generalized $s$-partition~$\mathcal{T}$ of~$T$ that admits a restricted multiway cut of size at most~$|T|$ in~$G$, we have~$\rmwc(G,\mathcal{T}) = \rmwc(\bypass_{R^*}(G), \mathcal{T}) = \rmwc(\bypass_R(G), \mathcal{T})$ for all~$R \subseteq R^*$. Conversely, any~$\mathcal{T}$ that does not admit a restricted multiway cut of size at most~$|T|$ in~$G$, does not admit such a cut in~$\bypass_{R^*}(G)$ either.

We now analyze the algorithm and its output. The bound on the running time follows easily: the algorithm consists of~$\ell \leq n$ iterations. In each iteration, we consider all~$(s+2)^{|T|} \in 2^{\Oh(|T|)}$ generalized $s$-partitions of~$\mathcal{T}$. For each of them, we make two calls to the algorithm of \cref{lem:multiway:alg} that runs in time~$2^{\Oh(|T|)} \cdot n^{\Oh(1)}$. Hence the total running time is~$2^{\Oh(|T|)} \cdot n^{\Oh(1)}$. Next, we prove that the output~$Z$ covers restricted minimum multiway cuts in the claimed way.

\begin{claim}
For each generalized $s$-partition~$\mathcal{T}$ of~$T$, if there is a restricted multiway cut for~$\mathcal{T}$ of size at most~$|T|$ in~$G$, then the computed set~$Z$ contains a \emph{minimum} restricted multiway cut of~$\mathcal{T}$.
\end{claim}
\begin{claimproof}
Consider a generalized $s$-partition~$\mathcal{T}$ of~$T$ that admits a restricted multiway cut of size at most~$|T|$ in~$G$. Let~$U$ be a minimum restricted multiway cut of~$\mathcal{T}$. Since the procedure has considered this partition~$\mathcal{T}$ in all its iterations, every time it added a vertex~$v_i$ to the set~$R$ we know that bypassing vertex~$v_i$ did not change the size of a restricted multiway cut for~$\mathcal{T}$ if that size was at most~$|T|$. Hence the graph~$\bypass_{R^*}(G)$ contains a restricted multiway cut~$U$ for~$\mathcal{T}$ of size~$\rmwc(G, \mathcal{T}) \leq |T|$. Note that~$U \subseteq V(G) \setminus R^*$ since~$\bypass_{R^*}(G)$ does not contain vertex set~$R^*$. By Observation~\ref{obs:multiway:bypass} and the fact that~$R^* \cap T = \emptyset$, the set~$U$ is also a restricted multiway cut for~$\mathcal{T}$ in graph~$G$, and since its size is at most~$\rmwc(G, \mathcal{T})$ it is minimum.
\end{claimproof}

To complete the proof of \cref{lem:covering:corollary}, it suffices to argue that~$|Z| \leq \Oh(|T|^{2s+s})$. To prove this bound, we will relate restricted multiway cuts in~$G$ of size at most~$|T|$ to non-restricted multiway cuts in an auxiliary graph~$G'$ with a somewhat larger terminal set~$T'$ that contains multiple copies of each terminal to prevent it from being deleted. We will then show that the algorithm by Kratsch and Wahlstr\"{o}m (\cref{thm:multiway-cover}) would have included every vertex of~$Z \setminus T$ in its cut covering set of~$G'$, which allows us to bound the size of~$Z \setminus T$ in terms of the output guarantee of \cref{thm:multiway-cover} with terminal set~$T'$.

%The proof combines two ideas. First of all, by making~$|T|+1$ copies of each vertex of~$T$ we can effectively ensure that multiway cuts of size at most~$|T|$ do not contain any vertex of~$T$ and are therefore restricted. This will allow us to infer the existence of a set~$Z$ of size~$\Oh(|T|^{s+2})$ with the required properties from \Cref{thm:multiway-cover}. To find such a set deterministically, we will replace the use of a randomized polynomial-time algorithm by the use of a deterministic algorithm whose running time is fixed-parameter tractable in~$|T|$. The time bound allows us to consider all possible generalized $s$-partitions of~$\mathcal{T}$, and for each of them, to test whether a certain vertex of the graph can be avoided while making an optimal restricted multiway cut for it. A filtering step based on this insight will then allow us to compute a set~$Z$ as desired.

We proceed to make these notions precise. Let~$G'$ be the undirected graph obtained as follows:
\begin{enumerate}
	\item Initialize~$G'$ as the graph~$\bypass_{R^*}(G)$ obtained from~$G$ by bypassing all vertices of~$R^*$.
	\item For each terminal vertex~$t \in T$, insert~$|T|$ additional false-twin copies~$t^{(1)}, \ldots, t^{(|T|)}$ into~$G'$. Hence for each for each~$i \in [|T|]$, the set~$N_{G'}(t^{(i)})$ contains the (copies of) the neighbors of vertex~$t$ in~$\bypass_{R^*}(G)$. 
	\item Let~$T' = T \cup \{ t^{(i)} \mid t \in T, i \in [|T|]\}$ denote the set of original terminals together with their copies. 
\end{enumerate}
Note that~$|T'| \in \Oh(|T|^2)$.

Any generalized $s$-partition~$\mathcal{T} = (T_0, T_1, \ldots, T_s, T_X)$ of~$T$ maps to a generalized $s$-partition~$\mathcal{T}' = (T'_0, T'_1, \ldots, T'_s, T'_X)$ of~$T'$. Starting from the generalized $s$-partition~$\mathcal{T}$, we obtain~$\mathcal{T}'$ by allocating the copies~$t^{(i)} \in T' \setminus T$ of a vertex~$t \in T$ to sets of~$\mathcal{T}'$ as follows:
\begin{itemize}
	\item for each~$t \in T_0$, its copies~$t^{(i)}$ for~$i \in [|T|]$ belong to~$T'_X$;
	\item for each~$t \in T_X$, its copies~$t^{(i)}$ for~$i \in [|T|]$ belong to~$T'_X$; and
	\item for each~$t \in T_j$ with~$j \in [s]$, its copies~$t^{(i)}$ for~$i \in [|T|]$ belong to~$T'_j$.
\end{itemize}
We say that~$\mathcal{T}'$ is the generalized $s$-partition of~$T'$ that corresponds to~$\mathcal{T}$. The definition ensures that each vertex that is involved in a separation requirement is represented~$|T|+1$ times, so that such vertices cannot be contained in a multiway cut of size at most~$|T|$. Note that for vertices~$t \in T$ belonging to~$T_0$, which are allowed to be deleted by a restricted multiway cut, the original vertex~$t$ remains in~$T'_0$ but the copies are placed in~$T'_X$. This ensures that the existence of the copies does not make it more costly to delete such vertices.

The following claim shows how restricted multiway cuts in~$\bypass_{R^*}(G)$ correspond to (non-restricted) multiway cuts in~$G'$.
\begin{claim} \label{claim:multiway:restricted}
Let~$\mathcal{T} = (T_0, T_1, \ldots, T_s, T_X)$ be a generalized $s$-partition of~$T$ and let~$\mathcal{T}' = (T'_0, T'_1, \ldots, T'_s, T'_X)$ be the corresponding generalized $s$-partition of~$T'$. The following holds:
\begin{enumerate}
	\item Any restricted multiway cut~$U$ of~$\mathcal{T}$ in~$\bypass_{R^*}(G)$ is a multiway cut of~$\mathcal{T}'$ in~$G'$.
	\item If~$U'$ is a multiway cut of~$\mathcal{T}'$ in~$G'$ of size at most~$|T|$, then~$U := U' \setminus (\bigcup _{i=1}^s T_i)$ is a restricted multiway cut of~$\mathcal{T}$ in~$\bypass_{R^*}(G)$.
\end{enumerate}
\end{claim}
\begin{claimproof}
For the first item, consider an arbitrary restricted multiway cut~$U$ of~$\mathcal{T} = (T_0, T_1, \ldots,$ $ T_s, \linebreak[0] T_X)$ in~$\bypass_{R^*}(G)$. This means that~$U$ is a multiway cut for~$(T_1, \ldots, T_s)$ in the graph~$\bypass_{R^*}(G) - T_X$. Our definition of~$T'_X$ ensures that the graph~$G' - T'_X$ can be obtained from~$\bypass_{R^*}(G) - T_X$ by making false-twin copies of the vertices in~$\bigcup _{j \in [s]} T_j$; here we exploit the fact that the copies of vertices in~$T_0$ end up in~$T'_X$. No vertex of~$\bigcup _{j \in [s]} T_j$ is contained in~$U$ since it is a \emph{restricted} multiway cut for~$\mathcal{T}$. As making false-twin copies of a vertex does not affect which pairs of the remaining vertices can reach each other, it follows that~$U$ is also a multiway cut for~$(T_1, \ldots, T_s)$ in~$G' - T'_X$. Since~$(T'_1, \ldots, T'_s)$ is obtained from~$(T_1, \ldots, T_s)$ by making false-twin copies and placing them in the same set that holds the original, the set~$U$ is also a multiway cut for~$(T'_1, \ldots, T'_s)$ in~$G' - T'_X$. Hence~$U$ is a multiway cut for~$\mathcal{T}'$ in~$G'$, as required.

Next, we establish the second item of the claim. Consider a multiway cut~$U'$ of~$\mathcal{T}'$ in~$G'$ of size at most~$|T|$. Hence~$U'$ is a multiway cut for~$(T'_1, \ldots, T'_s)$ in~$G' - T'_X$. We prove that~$U := U' \setminus (\bigcup _{i=1}^s T_i)$ is a restricted multiway cut for~$\mathcal{T}$ in~$\bypass_{R^*}(G)$. Assume for a contradiction that this is not the case. Since~$U$ is disjoint from~$\bigcup_{i=1}^s T_i$ by construction, it does not violate the requirement that the multiway cut is restricted. Hence~$U$ violates the separation requirement: there exist two different terminal sets~$T_i, T_j$ with~$1 \leq i < j \leq s$ and terminals~$t_i \in T_i, t_j \in T_j$ such that~$(\bypass_{R^*}(G) - T_X) - U$ contains a path~$P$ from~$t_i$ to~$t_j$. We show that~$P$ can be transformed into a path~$P'$ in~$(G' - T'_X) - U'$ that connects a vertex~$t'_i \in T'_i$ to a vertex~$t'_j \in T'_j$, which will contradict the assumption that~$U'$ is a multiway cut for~$\mathcal{T}'$ in~$G'$. To obtain~$P'$ from~$P$, it suffices to replace any vertex~$t$ on~$P$ that belongs to~$U' \cap (\bigcup _{i=1}^s T_i)$, by a copy~$t^{(i)}$ of~$t$ with~$t^{(i)} \notin U'$. Such a copy exists since~$|U'| \leq |T|$ by assumption, so that~$U'$ cannot contain all~$|T|$ copies of~$t$ together with the original vertex~$t \in U' \cap T$. The correspondence between~$\mathcal{T}$ and~$\mathcal{T}'$ ensures that the vertices~$t'_i, t'_j$ used to replace the endpoints of~$P$ (if needed) belong to distinct sets~$T'_i, T'_j$ of the generalized $s$-partition~$\mathcal{T}'$ with~$1 \leq i < j \leq s$. Hence~$P'$ contradicts the assumption that~$U'$ is a multiway cut for~$\mathcal{T}'$ in~$G'$, which concludes the proof.
\end{claimproof}

Using the preceding claim to translate between restricted and non-restricted multiway cuts, we now analyze the output of the algorithm.

\begin{claim} \label{claim:z:essential:g}
For each vertex~$v_i \in Z \setminus T$, there is a generalized $s$-partition~$\mathcal{T}$ of~$T$ such that:
\begin{enumerate}
	\item there is a restricted multiway cut~$U$ of size at most~$|T|$ for~$\mathcal{T}$ in~$\bypass_{R^*}(G)$, and
	\item vertex~$v_i$ belongs to \emph{all} minimum restricted multiway cuts for~$\mathcal{T}$ in~$\bypass_{R^*}(G)$.
\end{enumerate}
\end{claim}
\begin{claimproof}
Consider an arbitrary vertex~$v_i \in Z \setminus T$ and recall that~$R^*$ denotes the contents of variable~$R$ upon termination. By definition of~$Z$, we have~$v_i \in V(G) \setminus R^*$. The fact that~$v_i \notin T$ means the algorithm performed an iteration in which it considered this vertex~$v_i$ as a candidate to add to~$R$. Let~$\hat{R}$ be the vertex set representing the contents of variable~$R$ just before the iteration for~$v_i$. Since~$v_i \notin R^*$, vertex~$v_i$ was not added to~$R$ in its iteration. Hence there is a witnessing partition $\mathcal{T} = (T_0, T_1, \ldots, T_s, T_X)$ of~$T$ for which the two calls to \cref{lem:multiway:alg} yielded different outcomes~$X_{\mathcal{T}} \neq Y_{\mathcal{T}}$. Let~$\mathsf{opt}_{\mathcal{T}}$ denote the minimum size of a restricted multiway cut for~$\mathcal{T}$ in~$\bypass_{\hat{R}}(G - T_X)$. Since bypassing a non-terminal vertex can only \emph{increase} the cost of a restricted multiway cut, the fact that~$\mathcal{T}$ is a witness for~$v_i$ implies that~$X_\mathcal{T} = \mathsf{opt}_{\mathcal{T}} \leq |T|$ and that either~$Y_{\mathcal{T}} > X_{\mathcal{T}}$ or~$Y_{\mathcal{T}} = \bot$.

As noted above, the procedure for adding vertices to~$R$ is such that for any generalized $s$-partition that admits a restricted multiway cut of size at most~$|T|$, the size of a minimum restricted multiway cut is preserved whenever vertices are added to~$R$. Hence we have~$\rmwc(G, \mathcal{T}) = \rmwc(\bypass_{\hat{R}}(G), \mathcal{T}) = \rmwc(\bypass_{R^*}(G), \mathcal{T}) = \mathsf{opt}_{\mathcal{T}} \leq |T|$. Now assume for a contradiction that there is a minimum restricted multiway cut~$U$ of~$\mathcal{T}$ in~$\bypass_{R^*}(G)$ that does not contain~$v_i$; note that~$U$ is a subset of the vertices of~$\bypass_{R^*}(G)$, so that~$U \subseteq V(G) \setminus R^* = Z$. By the previous equality, we have~$|U| = \mathsf{opt}_{\mathcal{T}} \leq |T|$. Observation~\ref{obs:multiway:bypass} ensures that~$U$ is also a restricted multiway cut of~$\mathcal{T}$ in~$\bypass_{\hat{R}}(G)$. Since~$v_i \notin U$, another application of Observation~\ref{obs:multiway:bypass} shows that~$U$ is also a valid restricted multiway cut of~$\mathcal{T}$ in~$\bypass_{\hat{R} \cup \{v_i\}}(G)$. But then the graph~$\bypass_{\hat{R} \cup \{v_i\}}(G)$ has a restricted multiway cut of size~$\mathsf{opt}_{\mathcal{T}} = X_{\mathcal{T}} \leq |T|$ that the algorithm should have found; a contradiction to the assumption that~$Y_{\mathcal{T}} \neq X_{\mathcal{T}}$. This proves the claim.
% So in this proof we use the fact that if you can avoid a vertex to make a small optimal solution after all the bypass operations, then you could also avoid the vertex halfway the procedure, together with the fact that bypassing operations preserve OPT for all partitions that have a small solution.
\end{claimproof}

The following claim translates the essentiality of vertices to restricted multiway cuts in~$\bypass_{R^*}(G)$, into the essentiality for unrestricted multiway cuts in~$G'$.

\begin{claim} \label{claim:z:essential:gprime}
For each vertex~$v_i \in Z \setminus T$, there is a generalized $s$-partition~$\mathcal{T}'$ of~$T'$ such that vertex~$v_i$ belongs to \emph{all} minimum (non-restricted) multiway cuts of~$\mathcal{T}'$ in~$G'$.
\end{claim}
\begin{claimproof}
Let~$\mathcal{T} = (T_0, T_1, \ldots, T_s, T_X)$ be a generalized $s$-partition of~$T$ such that~$v_i$ belongs to all minimum restricted multiway cuts of~$\mathcal{T}$ in~$\bypass_{R^*}(G)$, and for which such minimum cuts have size at most~$|T|$. \cref{claim:z:essential:g} guarantees the existence of such~$\mathcal{T}$. Let~$\mathcal{T}' = (T'_0, T'_1, \ldots, T'_s, T'_X)$ be the generalized $s$-partition of~$T'$ corresponding to~$\mathcal{T}$. We argue that~$v_i$ belongs to all minimum (non-restricted) multiway cuts of~$\mathcal{T}'$ in~$G'$. 

Assume for a contradiction that~$U' \not \ni v_i$ is a minimum multiway cut of~$\mathcal{T}'$ in~$G'$. Since there is a restricted multiway cut of~$\mathcal{T}$ in~$\bypass_{R^*}(G)$ of size at most~$|T|$, while bypassing a set of non-terminals can only increase the size of a cut, there is a restricted multiway cut of~$\mathcal{T}$ in~$G$ of size at most~$|T|$. Hence the procedure for computing~$Z$ considered~$\mathcal{T}$ in each iteration, so that~$\mathsf{opt}_{\mathcal{T}} := \rmwc(G,\mathcal{T}) = \rmwc(\bypass_{R^*}(G), \mathcal{T}) \leq |T|$. The first part of \cref{claim:multiway:restricted} then ensures that the unrestricted counterpart in~$G^*$ satisfies~$\mwc(G^*, \mathcal{T}') \leq \mathsf{opt}_{\mathcal{T}} \leq |T|$. Since~$U'$ is minimum, we have~$|U'| \leq \mathsf{opt}_{\mathcal{T}} \leq |T|$. Hence we can invoke the second part of \cref{claim:multiway:restricted} to infer that~$U := U' \setminus (\bigcup_{i=1}^s T_i)$ is a restricted multiway cut of~$\mathcal{T}$ in~$\bypass_{R^*}(G)$. Since~$v_i \notin U'$, we have~$v_i \notin U$. We have arrived at a contradiction:~$U$ is a restricted multiway cut of~$\mathcal{T}$ avoiding~$v_i$ of size at most~$\mathsf{opt}_{\mathcal{T}} = \rmwc(\bypass_{R^*}(G), \mathcal{T})$ and is therefore optimal for~$\bypass_{R^*}(G)$; this contradicts our choice of~$v_i$.
\end{claimproof}

We can now prove the desired bound on the size of~$Z$.

\begin{claim}
The set~$Z$ computed by the procedure has size~$\Oh(|T|^{s+2})$.
\end{claim}
\begin{claimproof}
Consider the guarantee of \cref{thm:multiway-cover} when applied to graph~$G'$ with terminal set~$T'$ and the chosen value of~$s$. The theorem states that there is an algorithm that, with positive probability, computes a set~$\hat{Z}$ of size~$|\hat{Z}| = \Oh(|T'|^{s+1})$ that contains a minimum multiway cut of every generalized $s$-partition of~$T'$. So in particular, the theorem guarantees there \emph{exists} a set with these properties. So consider a set~$\hat{Z}$ of size~$\Oh(|T'|^{s+1}) = \Oh((|T|^2)^{s+1}) = \Oh(|T|^{2s+2})$ that contains a minimum multiway cut of every generalized $s$-partition. We argue that~$\hat{Z}$ contains~$Z \setminus T$, so that~$|Z| \leq |\hat{Z}| + |T| \in \Oh(|T|^{2s+2})$ which will establish the desired bound on~$|Z|$. It therefore suffices to argue that~$\hat{Z}$ contains all vertices of~$Z \setminus T$.

Consider an arbitrary vertex~$v_i \in Z \setminus T$. By \cref{claim:z:essential:gprime}, there is a generalized $s$-partition~$\mathcal{T}'$ of~$T'$ such that~$z$ belongs to all minimum multiway cuts of~$\mathcal{T}'$ in~$G'$. By assumption on~$\hat{Z}$, there is a minimum multiway cut~$U'$ of~$\mathcal{T}'$ in~$G'$ with~$U' \subseteq \hat{Z}$. Since~$z \in U'$ by the guarantee of \cref{claim:z:essential:gprime}, we have~$z \in U' \subseteq \hat{Z}$, so that~$z \in \hat{Z}$. Since~$z$ was arbitrary, this establishes that~$Z \setminus T \subseteq \hat{Z}$.
\end{claimproof}

This concludes the proof of \cref{lem:covering:corollary}.
\end{proof}
	
  % \begin{corollary}\label{cor:multiway-cover}
  %   Let~$G$ be an undirected graph on $n$ vertices with a set~$T \subseteq V(G)$ of terminal vertices,
  %   and let~$s \in \mathbb{N}$ be a constant.
  %   There is a set~$Z \subseteq V(G) \setminus T$ with~$|Z| \leq \binom{3|T|}{s+1}$
  %   such that~$Z$ contains a minimum multiway cut \emph{with undeletable terminals} of every generalized $s$-partition~$\mathcal{T}^*$ of~$T$,
  %   and we can compute such a set in randomized polynomial time
  %   with failure probability~$\mathcal{O}(2^{-n})$.
  %   %
  %   Especially when $s=3$, we have $|Z|\leq \frac{27}{8}|T|^4$.
		
	% 	\bmp{I don't see how this statement follows from the previous. The previous statement only gives asymptotic bounds on~$|Z|$ while the statement gives an absolute bound. Also, in its current form I think this is false. Observe that when the terminals are not allowed to be deleted, the minimum size of a multiway cut can be arbitrarily much larger than~$|T|$. So even disregarding how you would compute such a set efficiently, you cannot even guarantee that there is a set~$Z$ that contains even \emph{one} multiway cut with undeletable terminals and whose size is bounded in terms of~$|T|$; let alone guarantee that there is a set that contains a multiway cut for all partitions, and whose size is bounded by a \emph{polynomial} function of~$|T|$.}
	% 	\yo{I was missing a condition: if a restricted multiway cut of size at most $|T|$ exists...}
  % \end{corollary}

\subsection{Universal sets: the deterministic version of color coding} \label{sec:color:coding}

For a set~$D$ of size~$n$ and integer~$k$ with~$n \geq k$, an~\emph{$(n,k)$-universal set} for~$D$ is a family~$\mathcal{U}$ of subsets of~$D$ such that for all~$S \subseteq D$ of size at most~$k$ we have~$\{S \cap U \mid U \in \mathcal{U}\} = 2^S$.

\begin{theorem}[{\cite[Theorem~6]{NaorSS95}, cf.~\cite[Theorem~5.20]{CyganFKLMPPS15}}] \label{thm:universal-set}
 For any set~$D$ and integers~$n$ and~$k$ with~$|D| = n \geq k$, an $(n,k)$-universal set~$\mathcal{U}$ for~$D$ with~$|\mathcal{U}| = 2^{\Oh(k)}\log n$ can be created in~$2^{\Oh(k)} n \log n$ time.
\end{theorem}

We extend the notion of universal set to the following notion of \emph{universal function family}.
For sets $A$ of size $n$ and $B$ of size $q$ and an integer $k$ with $n \geq k$,
an~\emph{$(n,k,q)$-universal function family} is a family~$\mathcal{F}$ of functions $A \to B$
such that for all~$S \subseteq A$ of size at most~$k$, the family $\{{f|}_{S} \colon f \in \mathcal{F} \}$ contains all $q^k$ functions from $A$ to $B$. Here we write~$f|_{S}$ for the restriction of~$f$ to domain~$S$.

\begin{corollary}\label{cor:universal-partition}
	There is an algorithm that, given sets $A,B$ with $n=|A|$, $q=|B|$,
	and a positive integer~$k \leq n$, constructs an~$(n,k,q)$-universal function family~$\mathcal{F}$ of size
  $2^{\Oh(k q')} \log^{q'} n$
  in $2^{\Oh(k q')} nq' \log^{q'} n$ time, where $q'=\lceil\log_2 q\rceil$.
  In particular, when $q$ is a constant, both the size and construction time are $2^{\Oh(k)} n^{\Oh(1)}$.
\end{corollary}

\begin{proof}
  Let $q'= \lceil \log_2 q \rceil$.
	The algorithm starts by computing an~$(n,k)$-universal set for $A$ in $2^{\Oh(k)} n \log n$ time via \cref{thm:universal-set}.
	We conceptually create~$q'$ copies of this universal set, denoted~$\mathcal{U}_i$ for~$i \in [q']$,
	where each $\mathcal{U}_i$ has size $2^{\Oh(k)} \log n$.
	Now consider the Cartesian product $\mathcal{U}_{\times} := \mathcal{U}_1 \times \ldots \times \mathcal{U}_{q'}$. 
  For a set $X \subseteq A$ we denote its indicator function by $I_X \colon A \to \{0,1\}$, so that $I_X(x) = 1$ if $x \in X$ and $I_X(x)=0$ if $x \not\in X$.
  Also, let $g \colon \{0,1\}^{q'} \to [2^{q'}]$ be a function that converts binary representations of
  integers between $0$ and ($2^{q'}-1$) into integers in the range~$\{1, \ldots, 2^{q'}\}$; note that~$2^{q'} \geq q$.
	Finally, define an arbitrary surjective function $h: [2^{q'}] \to B$.
  Now, we construct a family of functions $\mathcal{F}$ as follows:
  for each $(U_1,\ldots,U_{q'}) \in \mathcal{U}_\times$,
  create a function $f \colon A \to B$ such that
	$f(x)=h(g(I_{U_1}(x), \ldots, I_{U_{q'}}(x)))$.
	% $f(x)=h(\min\{q, g(f_{U_1}(x), \ldots, f_{U_{q'}}(x))\})$.
	% Note that we take the minimum with~$q$ to ensure that the input for the function $h$ is confined to~$[q]$,
	% even though some binary representations might encode a larger number when~$q' > \log_2 q$. 
  %
  It is straightforward to see that $|\mathcal{F}| \leq |\mathcal{U}_\times| \leq |\mathcal{U}_1|^{q'} \leq 2^{\Oh(kq')}\log^{q'}n$.
  This construction can be done in $\Oh(|\mathcal{F}|\cdot nq')$ time
  because for each $f \in \mathcal{F}$, its values~$f(x)$ can be computed by sequentially scanning $(U_1,\ldots,U_{q'})$ in time $\Oh(nq')$.

  We complete the proof by showing that $\mathcal{F}$ is an $(n,k,q)$-universal function family.
  First, observe that for all $S \subseteq A$ of size at most $k$, for each~$i \in [q']$, 
  the family $\{I_{U_i}|_S \colon U_i \in \mathcal{U}_i\}$ contains all functions from $S$ to $\{0,1\}$.
  Consequently, $\{ g(I_{U_1}|_S, \ldots, I_{U_{q'}}|_S) \colon (U_1,\ldots,U_{q'}) \in \mathcal{U}_\times \}$ contains all functions from $S$ to $[2^{q'}]$.
	Since the function $h$ is surjective,
  % Since $[q] \subseteq [2^{q'}]$,
	the family $\{{f|}_S : f \in \mathcal{F} \}$ includes all functions from $S$ to $B$.
\end{proof}

\section{Proofs for Odd Cycle Cuts} 
\label{sec:proofs:oct:cut}

\lemOCCARNSeparators*
\begin{proof}
	We start by proving the former of the two implications.
	
	\begin{claim} \label{clm:tight:occ:cut:optimal:AR:separtor}
		$A_C \cap X_B$ is a minimum-size $\{A, R\}$-separator in $G[X_B]$.
	\end{claim}
	\begin{claimproof}
		We start by showing that $A_C \cap X_B$ is a (not necessarily smallest) $\{A, R\}$-separator in $G[X_B]$. Consider the graph $\widetilde{G} := G[(X_C \cap A_B) \cup X_B]$, let $W_0 = f_A^{-1}(0) \cap X_C$, $W_1 = f_A^{-1}(1) \cap X_C$ and $c = f_X$.
		
		Using these parameters, we invoke \cref{lem:AR-separation}. Note in particular that the sets~$A$ and~$R$ in this lemma coincide with the sets~$A$ and~$R$ we obtained from \cref{def:imposed:separation:problem} at the start of the proof of \cref{lem:tight:occ:cut:optimal:ARN:separator}. Now, from \cref{lem:AR-separation}, we obtain that it suffices  to show that $\widetilde{G} - (A_C \cap X_B)$ admits a proper $2$-coloring with vertices from $W_0$ and $W_1$ receiving colors $0$ and $1$ respectively. This last property can be satisfied by defining our coloring $f: (V(\widetilde{G}) \setminus (A_C \cap X_B) \rightarrow \{0, 1\}$ to have $f(v) = f_A(v)$ for all $v \in A_B \setminus X_R$. We complete the coloring $f$ by setting $f(v) = f_X(v)$ for all $v \in X_B \cap A_R$. This is a proper $2$-coloring of $\widetilde{G} - (A_C \cap X_B)$ because:
		\begin{itemize}
			\item it is a proper $2$-coloring restricted to $A_B \setminus X_R$ (as $f_A$ is a proper coloring of $G[A_B]$), and
			\item it is a proper $2$-coloring restricted to $X_B \cap A_R$ (as $f_X$ is a proper coloring of $G[X_B]$), and
			\item the sets $A_B \setminus X_R$ and $X_B \cap A_R$ cover all vertices of $\widetilde{G} - (A_C \cap X_B)$ and are not adjacent to one another as they are subsets of $A_B$ and $A_R$ respectively.
		\end{itemize}
		This shows that $A_C \cap X_B$ is an $\{A, R\}$-separator in $G[X_B]$ and it remains to show that it is a smallest such separator. For contradiction, let $S \subseteq X_B$ be an $\{A, R\}$-separator in $G[X_B]$ with $|S| < |A_C \cap X_B|$. Given such a set $S$, we show that $\widehat{S} := (A_C \setminus X_B) \cup S$ is an OCT in $G[A_B \cup A_C]$. Since $\widehat{S}$ is then strictly smaller than $A_C$, this will contradict the assumption that $(A_B, A_C, A_R)$ is a tight OCC, which requires that $A_C$ is a smallest OCT of $G[A_B \cup A_C]$.
		
		To prove that $\widehat{S}$ is an OCT of $G[A_B \cup A_C]$, we show that $G[A_B \cup A_C \cup X_B] - \widehat{S}$ is bipartite of which $G[A_B \cup A_C] - \widehat{S}$ is a subgraph. To this end, we recall that $S$ is an $\{A, R\}$-separator in $G[X_B]$. Hence, by \cref{lem:AR-separation} with $c = f_X$ and $W_0$ and $W_1$ based on $f_A$ as defined above, there is a proper $2$-coloring $\widetilde{f}$ of $\widetilde{G} - S$ with $\widetilde{f}(v) = f_A(v)$ for each $v \in X_C \cap A_B$.
		
		Now, the subgraphs induced by $A_B \setminus X_B$ and $V(\widetilde{G}) \setminus S$ can be properly $2$-colored by $f_A$ and $\widetilde{f}$ respectively. These vertex sets intersect precisely in $X_C \cap A_B$ on which the two colorings agree and which separates the remaining vertices from the two sets. Hence, by \cref{lem:separated:colorings:make:bipartite}, the graph induced by the union of these sets, $G[A_B \cup A_C] - \widehat{S}$, is bipartite.
	\end{claimproof}
	Towards proving that $A_C \cap X_B$ is also a minimum-size $\{A, R, N\}$-separator in $G[X_B]$, we prove the following claim.
	
	\begin{claim} \label{clm:tight:occ:cut:separates:N}
		$A_C \cap X_B$ is a (not necessarily minimum-size) $\{A \cup R, N\}$-separator in $G[X_B]$.
	\end{claim}
	\begin{claimproof}		
		Note that $A \cup R$, being defined as a set of neighbors of $A_B \cap X_C$, is a subset of $A_B \cup A_C$, because $A_C$ separates $A_B$ from the remaining vertices of the graph. Likewise, as $N$ is defined as a set of neighbors of $A_R \cap X_C$, it is a subset of $A_R \cup A_C$, because $A_C$ also separates $A_R$ from the remaining vertices of the graph. As such, when removing $A_C$ from the graph, the remaining vertices of $A \cup R$ and $N$ live in $A_B$ and $A_R$ respectively. Since $A_C$ separates these latter two sets by definition of $(A_B, A_C, A_R)$ being an OCC, $A_C \cap X_B$ is an $\{A \cup R, N\}$-separator in $G[X_B]$.
	\end{claimproof}
	
	Now, the fact that $A_C \cap X_B$ is an $\{A, R, N\}$-separator follows from the claims that this set is both an$\{A, R\}$-separator (\cref{clm:tight:occ:cut:optimal:AR:separtor}) and an $\{A \cup R, N\}$-separator (\cref{clm:tight:occ:cut:separates:N}). The fact that it is even a \emph{minimum-size} $\{A, R, N\}$-separator follows from the observation that any $\{A, R, N\}$-separator must be at least as large as a minimum-size $\{A, R\}$-separator and from the claim that $A_C \cap X_B$ is itself a minimum-size $\{A, R\}$-separator (\cref{clm:tight:occ:cut:optimal:AR:separtor}).
	
	This proves the second implication from the lemma statement and concludes the proof.
\end{proof}

To facilitate a proof of \Cref{lem:num_components}, here we characterize the components in the bipartite part of an $X_C$-certificate that are safe to remove.

\begin{lemma}\label{lem:redundant-component}
  For a graph $G$ and its vertex set $X_C \subseteq V(G)$,
  let $H$ be an $X_C$-certificate of order $z$ in~$G$.
  Suppose there exists a connected component $D$ of $H - X_C$ that satisfies the following two criteria,
  where $H'$ is a connected component of $H$ that contains $D$, and $X'_C := X_C \cap V(H')$:
  \begin{enumerate}[(1)]
    \item For each vertex $x \in X'_C$ for which $H[\{x\} \cup V(D)]$ contains an odd cycle,
    there are at least $z$ connected components $D' \neq D$ of $H-X_C$ for which
    $H[\{x\} \cup V(D')]$ also contains an odd cycle, and 
    \item For each pair of distinct vertices $x,y \in X'_C$, for each parity $p \in \{0,1\}$,
    if the graph $H[\{x,y\} \cup V(D)]$ contains an $(x,y)$-path of parity $p$,
    then there are at least $z$ connected components $D' \neq D$ of $H-X_C$ for which
    $H[\{x,y\} \cup V(D')]$ also contains an $(x,y)$-path of parity $p$.
  \end{enumerate}
  Then $\oct(H) = \oct(H - V(D))$.
\end{lemma}

\begin{proof}
Suppose $D$ satisfies these criteria.
Assume for a contradiction that $\oct(H - V(D)) < \oct(H)$.
% Let $H'$ be the connected component of $H$ that contains $V(D)$.
Then $\oct(H' - V(D)) < \oct(H') \leq z$,
where the last inequality comes from the assumption that $H$ is an order-$z$ $X_C$-certificate.
Let $X$ be an optimal OCT for $\oct(H' - V(D))$, whose size is less than $\oct(H') \leq z$.
It follows that $H' - X$ has an odd cycle.
Out of all odd cycles of $H' - X$,
let $A$ be an odd cycle that minimizes the quantity $t := |V(A) \cap V(D)|$.
Note that $t > 0$, otherwise $A$ would also be an odd cycle in $H' - V(D)$
which contradicts $X$ being an OCT for that subgraph.
Since $X_C$ is an OCT in $H$ by the definition of $X_C$-certificate,
the component $D$ of $H - X_C$ is bipartite.
Hence the odd cycle $A$ is not fully contained in $V(D)$,
but contains at least one vertex $d$ of $D$ since $t \geq 1$.
Orient the odd cycle $A$ in an arbitrary direction and let $x$ be the first vertex after $d$ on $A$ that does not belong to $V(D)$.
Similarly, let $y$ be the first vertex before $d$ on $A$ that does not belong to $V(D)$.
We distinguish two cases, depending on whether $x=y$.

Case 1: If $x=y$.
In this case, $A$ is an odd cycle in $H[\{x\} \cup V(D)]$.
Since component $D$ satisfied criterion (1),
there are at least $z$ connected components $D' \neq D$ of $H-X_C$
for which $H[\{x\} \cup V(D')]$ contains an odd cycle.
Since $|X| < z$, at least one of these components $D^*$ does not contain any vertex of $X$.
Then, $H[\{x\} \cup V(D^*)]$ contains an odd cycle $A'$, but we have $V(A') \cap V(D) = \emptyset$.
This contradicts our choice of $A$ as an odd cycle minimizing $V(A) \cap V(D)$.

Case 2: If $x \neq y$.
Observe that the subpath of $C$ from $x$ to $y$ through $d$ forms an $(x,y)$-path $P$ in the graph $H[\{x,y\} \cup V(D)]$.
Let $p$ be the parity of path $P$.
Since component $D$ satisfied criterion (2),
there are at least $z$ connected components $D' \neq D$ of $H-X_C$
that provide an $(x,y)$-path of parity $p$.
Since $|X| < z$, at least one of these components $D^*$ does not contain any vertex of $X$.
Consider the closed walk $A'$ that is obtained by replacing the subpath $P$ of $A$ by
an $(x,y)$-path $P'$ of the same parity in the graph $H[\{x,y\} \cup V(D^*)]$.
Since the replacement preserves the total parity and provides an alternative connection between
the two vertices $x,y$ that appeared on $A$,
it follows that $A'$ is a closed odd walk.
Since $V(D^*) \cap X = \emptyset$, it follows that $A'$ is a closed odd walk in $H' - X$,
which implies that $H' - X$ contains an odd cycle whose vertex set is a subset of $V(A')$.
But note that $A'$ uses strictly fewer vertices from $V(D)$ than $A$ does,
since vertex $d \in V(D)$ does not appear in $V(A')$.
This contradicts our choice of $A$ as an odd cycle minimizing $|V(A) \cap V(D)|$
and proves the claim.
\end{proof}

The following lemma states that we can bound the number of components in an $X_C$-certificate.
This is inspired by and a variation of \cite[Lemma 4.5]{DonkersJ24}.

\begin{lemma}\label{lem:num-components-certificate}
	Let $G$ be a graph. For a set $X_C \subseteq V(G)$, let $H$ be an $X_C$-certificate of order $z$ in~$G$.
  Then $H$ contains an $X_C$-certificate $\hat{H}$ of order $z$ in $G$
	such that $\hat{H}-X_C$ has at most $z^2|X_C|$ components.
\end{lemma}

\begin{proof}
  If there is a connected component $D$ of $H-X_C$ that satisfies the two criteria of \Cref{lem:redundant-component},
  then, update $H$ to $H-V(D)$.
  From \Cref{lem:redundant-component}, we have $\oct(H) = \oct(H-V(D))=|X_C|$,
  and since we never remove vertices in $X_C$, $H-V(D)$ remains an \zocc.
  We repeat this process until the criteria become unsatisfied,
  and let $\hat{H}$ be the final graph.

  Now, we count the number of components in $\hat{H}-X_C$.
  For each vertex $x \in X_C$, there are at most $z$ components in $\hat{H}-X_C$ that violates criterion (1)
  because otherwise, one of those components will satisfy criterion (1).
  This results in that there exist at most $z \cdot |X_C|$ components in $\hat{H}-X_C$ violating criterion (1).

  For each component $H'$ of $H$, for each distinct vertex pair $x,y \in X_C \cap H'$, 
  and for each parity $p \in \{0,1\}$, there are at most $z$ components in $\hat{H}-X_C$ that violates criterion (2)
  because otherwise, one of those components will satisfy criterion (2).
  This results in that there exist at most
  $\sum_{H'} \binom{|X_C \cap V(H')|}{2} \cdot 2z
  = z \cdot \sum_{H'} (|H_C \cap V(H')|)(|H_C \cap V(H')|-1) \leq z\cdot (z-1)|X_C|$
  components in $\hat{H}-X_C$ violating criterion (2).

  Hence, the number of components in $\hat{H}-X_C$ is upper-bounded by
  $z \cdot |X_C| + z \cdot (z-1) |X_C|=z^2|X_C|$, as desired.
\end{proof}

Lastly, we restate and prove \Cref{lem:num_components}.

\restatenumcomponents*

\begin{proof}
Since $(A_B,A_C,A_R)$ is a \zocc, the graph $G[A_B \cup A_C]$ contains an $A_C$-certificate $H$ of order $z$.
From \Cref{lem:num-components-certificate} we know that $H$ contains an $A_C$-certificate $\hat{H}$
of order $z$ such that $\hat{H}-A_C$ has at most $z^2|A_C|$ components.

Let $A_B'$ be the union of the vertices in the components $B$ in $G[A_B]$ such that
$B$ includes at least one vertex from $\hat{H}$.
By definition, $(A_B',A_C,A_R \cup A_B \setminus A_B')$ is a \zocc such that
$G[A_B']$ has at most $z^2|A_C|$ components.
\end{proof}

\section{Finding an \occ given a coloring}
\label{sec:proofs:oct:finding}

In this section, we prove \Cref{lem:finding-occ-colored} using 
bipartite separations, as introduced by Jansen et al.~\cite{JansenKW22}. 

Specifically, define a \emph{$(\mathcal{H},k)$-separation} to be 
a pair $(C,S)$ of disjoint vertex sets for a graph class $\mathcal{H}$ and
an integer $k$ such that $G[C] \in \mathcal{H}$, $|S| \leq k$ and $N_G(C) \subseteq S$.
For $\mathcal{H}=\bip$, where $\bip$ denotes the class of bipartite graphs,
Jansen et al.~\cite{JansenKW22} gave a $2$-approximation for the problem of computing a $(\bip, k)$-separation covering a given connected vertex set~$Z$.

\begin{corollary}[{{\cite[Lemma 4.24]{JansenKW22}} (rephrased)}] \label{cor:bipartite-separation}
  There is a polynomial-time algorithm that, given a graph~$G$, an integer~$k$, and a non-empty vertex set~$Z$ such that~$G[Z]$ is connected, either:
	\begin{itemize}
		\item returns a~$(\bip, 2k)$-separation~$(C,S)$ with~$Z \subseteq C$, or
		\item concludes that~$G$ does not have a~$(\bip, k)$-separation~$(C',S')$ with~$Z \subseteq C'$.
	\end{itemize}
\end{corollary}

The original formulation of the corresponding result is in terms of a problem called \textsc{$(\mathcal{H}, k)$-Separation Finding}. In the definition of this problem, the input graph is required to have a bound on its $\mathcal{H}$-treewidth and the output separation is only guaranteed to \emph{weakly} cover~$Z$, meaning that~$Z \subseteq C \cup S$ rather than~$Z \subseteq C$. Jansen et al.~\cite{JansenKW22} used this formalism to treat a variety of different graph classes~$\mathcal{H}$. For the particular case of bipartite graphs that is of interest here, the assumption on the $\mathcal{H}$-treewidth of the input graph is never used (as is acknowledged in their paper), and from the specification of the algorithm it can be seen that~$Z \cap S$ is empty: the set~$S$ is constructed as a separator in an auxiliary graph, separating two copies of~$Z$ without intersecting these sets. Hence our formulation follows directly from their results.

\lemFindColoredOCC*
\begin{proof}
    Consider the following algorithm.
    Let $Z$ be the vertices in a connected component of $G[\chi^{-1}(\cB)]$.
    We iterate over all components.
    If $|Z|<\ell$ or $G[Z]$ is not bipartite, continue to the next component.
    Otherwise, we invoke \Cref{cor:bipartite-separation} to find a $(\bip, 2k)$-separation $(B,S)$ such that $Z \subseteq B$.
    If such a separation exists, output $(B, S, V(G) \setminus (B \cup S))$.
    Otherwise, continue to the next component.
  
  \subparagraph*{Correctness.}
  Let $G' = G[\chi^{-1}(\cB)]$.
  By assumption, the single-component \occ $(X_B, X_C, X_R)$ is $\ell$-properly colored by $\chi$. Hence $V(G') \cap X_C = \emptyset$ and $G'$ contains a connected component
  $G[Z]$ such that $Z \subseteq X_B$ and $|Z| \geq \ell$.
  Also, because $(X_B, X_C)$ is a $(\bip, |X_C|)$-separation in $G$ such that $Z \subseteq X_B$,
  the algorithm must output some solution $(B,S,V(G)\setminus (B \cup S))$ such that $Z \subseteq B$,
  which is by definition an \occ,
  and we know $|B| \geq |Z| \geq \ell$ and $|S| \leq 2|X_C| \leq 2k$.
  
  \subparagraph*{Running time.}
  From \Cref{cor:bipartite-separation}, we can,
  given a component $Z$, find a $(\bip,2k)$-separation $(B,S)$
  such that $Z \subseteq B$ in polynomial time.
  There are $\Oh(n)$ components in $G'$, so the overall runtime is also polynomial.
  \end{proof}
\section{Proofs regarding guarantees of \refmark}
In this section, we cover the various guarantees that were promised for \refmark. A visual example of the graph reduction that is part of it may be found in \cref{fig:find-central}.

\begin{figure*}[t]
  \centering
    % taken from: https://tex.stackexchange.com/questions/100724/tikz-node-with-multiple-lines-around
    \tikzset{
        old inner xsep/.estore in=\oldinnerxsep,
        old inner ysep/.estore in=\oldinnerysep,
        double circle/.style 2 args={
            circle,
            old inner xsep=\pgfkeysvalueof{/pgf/inner xsep},
            old inner ysep=\pgfkeysvalueof{/pgf/inner ysep},
            /pgf/inner xsep=\oldinnerxsep+#1,
            /pgf/inner ysep=\oldinnerysep+#1,
            alias=sourcenode,
            append after command={
            let     \p1 = (sourcenode.center),
                    \p2 = (sourcenode.east),
                    \n1 = {\x2-\x1-#1-0.5*\pgflinewidth}
            in
                node [inner sep=0pt, draw, circle, minimum width=2*\n1,at=(\p1),#2] {}
            }
        },
        double circle/.default={-3pt}{black}
    }
  \tikzstyle{large} = [circle, fill=white, text=black, draw, thick, scale=1, minimum size=0.8cm, inner sep=1.5pt]
  \tikzstyle{plain} = [circle, fill=white, text=black, draw, thick, scale=1, minimum size=0.5cm, inner sep=1.5pt]
  \tikzstyle{small} = [circle, fill=white, text=black, draw, thick, scale=1, minimum size=0.2cm, inner sep=1.5pt]
  \tikzmath{\xunit = 1.3; \yunit = 0.8;}

  \begin{minipage}[m]{.50\textwidth}
      \vspace{0pt}
      \centering
      \begin{tikzpicture}
          % vertices
          \node[small] (r5) at (0.2 * \xunit, 2.8 * \yunit) {};
          \node[small] (r6) at (0.2 * \xunit, 3.6 * \yunit) {};

          \node[small] (v1) at (1 * \xunit, 1 * \yunit) {};
          \node[small] (v2) at (1 * \xunit, 2 * \yunit) {};
          \node[small] (v3) at (1 * \xunit, 3.25 * \yunit) {};
          \node[small] (v4) at (1 * \xunit, 4.4 * \yunit) {};
          \node[small] (v5) at (1 * \xunit, 5.2 * \yunit) {};
          \node[small] (v6) at (1 * \xunit, 6 * \yunit) {};

          \node[small] (b11) at (2 * \xunit, 1 * \yunit) {};
          \node[small] (b12) at (2 * \xunit, 2 * \yunit) {\footnotesize $N$};
          % \node[small] (b13) at (2 * \xunit, 3 * \yunit) {};
          \node[large,double circle] (b14) at (2 * \xunit, 3.25 * \yunit) {\scriptsize $ARN$};
          \node[small] (b15) at (2 * \xunit, 4.4 * \yunit) {\footnotesize $A$};
          \node[small] (b16) at (2 * \xunit, 6 * \yunit) {\footnotesize $R$};
          % \node[small] (b21) at (3 * \xunit, 1 * \yunit) {};
          \node[small] (b22) at (3 * \xunit, 2 * \yunit) {\footnotesize $N$};
          % \node[small] (b23) at (3 * \xunit, 3 * \yunit) {};
          \node[small,double circle] (b24) at (3 * \xunit, 3.25 * \yunit) {};
          \node[small] (b25) at (3 * \xunit, 4.4 * \yunit) {\footnotesize $R$};
          \node[small] (b26) at (3 * \xunit, 6 * \yunit) {\footnotesize $A$};

          % edges
          \draw (r5)--(r6);
          \draw (r5)--(v3);
          \draw (r6)--(v3);

          \draw (v1)--(v2);
          \draw (v2)--(v3);
          \draw (v3)--(v4);
          \draw (v5)--(v6);
          % \draw (v4)--(v6);

          \draw (v6)--(b15);
          \draw (v6)--(b14);
          \draw (v6)--(b25);
          \draw (v4)--(b16);
          \draw (v4)--(b26);

          \draw (v5)--(b14);

          \draw (v3)--(b14);
          \draw (v3)--(b15);
          \draw (v3)--(b12);
          \draw (v2)--(b12);
          \draw (v2)--(b14);
          \draw (v1)--(b12);
          \draw (v1)--(b22);

          \draw (b16)--(b25);
          \draw (b15)--(b26);
          \draw (b15)--(b24);
          \draw (b14)--(b25);

          \draw (b16)--(b24);
          \draw (b12)--(b24);
          \draw (b14)--(b22);
          \draw (b11)--(b22);
          \draw (b11)--(b24);

          % borders
          \draw[dotted] (-0.2 * \xunit, 0.5 * \yunit) -- (3.5 * \xunit, 0.5 * \yunit);
          % \draw[dotted] (-0.2 * \xunit, 2.5 * \yunit) -- (3.5 * \xunit, 2.5 * \yunit);
          % \draw[dotted] (-0.2 * \xunit, 4.0 * \yunit) -- (3.5 * \xunit, 4.0 * \yunit);
          % \draw[dotted] (-0.2 * \xunit, 5.5 * \yunit) -- (3.5 * \xunit, 5.5 * \yunit);
          \draw[dotted] (-0.2 * \xunit, 6.5 * \yunit) -- (3.5 * \xunit, 6.5 * \yunit);
          \draw[dotted] (-0.2 * \xunit, 0.5 * \yunit) -- (-0.2 * \xunit, 6.5 * \yunit);
          \draw[dotted] (0.5 * \xunit, 0.5 * \yunit) -- (0.5 * \xunit, 6.5 * \yunit);
          \draw[dotted] (1.5 * \xunit, 0.5 * \yunit) -- (1.5 * \xunit, 6.5 * \yunit);
          \draw[dotted] (2.5 * \xunit, 0.5 * \yunit) -- (2.5 * \xunit, 6.5 * \yunit);
          \draw[dotted] (3.5 * \xunit, 0.5 * \yunit) -- (3.5 * \xunit, 6.5 * \yunit);

          \draw[darkgray, rounded corners] (0.6 * \xunit, 0.7 * \yunit) rectangle (1.4 * \xunit, 2.3 * \yunit) {};
          \draw[darkgray, rounded corners] (0.6 * \xunit, 5.7 * \yunit) rectangle (1.4 * \xunit, 6.3 * \yunit) {};
          \draw[darkgray, rounded corners] (0.6 * \xunit, 4.1 * \yunit) rectangle (1.4 * \xunit, 5.5 * \yunit) {};

          % annotation
          \node[] at (0.3 * \xunit, 1.5 * \yunit) {$C_2$};
          \node[] at (0.1 * \xunit, 4.9 * \yunit) {\footnotesize$f_C^{-1}(1)$};
          \node[] at (0.1 * \xunit, 6 * \yunit) {\footnotesize$f_C^{-1}(0)$};

          \draw [decorate,decoration={brace,amplitude=5pt}]
          (-0.3 * \xunit, 4.0 * \yunit) -- (-0.3 * \xunit, 6.5 * \yunit) node[midway,xshift=-14pt]{$C_1$};

          \node[] at (0.2, 7 * \yunit) {$X_R$};
          \node[] at (1 * \xunit, 7 * \yunit) {$X_C$};
          \node[] at (2 * \xunit, 7 * \yunit) {\footnotesize$f_B^{-1}(0)$};
          \node[] at (3 * \xunit, 7 * \yunit) {\footnotesize$f_B^{-1}(1)$};

          \draw [decorate,decoration={brace,amplitude=5pt}]
          (1.5 * \xunit, 7.2 * \yunit) -- (3.5* \xunit, 7.2 * \yunit) node[midway,yshift=12pt]{$X_B$};
          \node[] at (-0.5 * \xunit, 7.8 * \yunit) {$G$};
      \end{tikzpicture}
  \end{minipage}
  \hfill
  \begin{minipage}[m]{.44\textwidth}
    \vspace{0pt}
    \centering
    \begin{tikzpicture}
          % vertices
          \node[small] (v1) at (1 * \xunit, 1 * \yunit) {\footnotesize $N'$};
          \node[small] (v2) at (1 * \xunit, 2 * \yunit) {\footnotesize $N'$};
          \node[small,dashed] (v3) at (1 * \xunit, 3.25 * \yunit) {\footnotesize $T_X$};
          \node[small] (v4) at (1 * \xunit, 4.2 * \yunit) {\footnotesize $R'$};
          \node[small] (v5) at (1 * \xunit, 5.1 * \yunit) {\footnotesize $R'$};
          \node[small] (v6) at (1 * \xunit, 6 * \yunit) {\footnotesize $A'$};

          \node[small] (w1) at (4 * \xunit, 1 * \yunit) {\footnotesize $N'$};
          \node[small] (w2) at (4 * \xunit, 2 * \yunit) {\footnotesize $N'$};
          \node[small,dashed] (w3) at (4 * \xunit, 3.25 * \yunit) {\footnotesize $T_X$};
          \node[small] (w4) at (4 * \xunit, 4.2 * \yunit) {\footnotesize $A'$};
          \node[small] (w5) at (4 * \xunit, 5.1 * \yunit) {\footnotesize $A'$};
          \node[small] (w6) at (4 * \xunit, 6 * \yunit) {\footnotesize $R'$};

          \node[small] (b11) at (2 * \xunit, 1 * \yunit) {};
          \node[small] (b12) at (2 * \xunit, 2 * \yunit) {};
          % \node[small] (b13) at (2 * \xunit, 3 * \yunit) {};
          \node[small,double circle] (b14) at (2 * \xunit, 3.25 * \yunit) {};
          \node[small] (b15) at (2 * \xunit, 4.4 * \yunit) {};
          \node[small] (b16) at (2 * \xunit, 6 * \yunit) {};
          % \node[small] (b21) at (3 * \xunit, 1 * \yunit) {};
          \node[small] (b22) at (3 * \xunit, 2 * \yunit) {};
          % \node[small] (b23) at (3 * \xunit, 3 * \yunit) {};
          \node[small,double circle] (b24) at (3 * \xunit, 3.25 * \yunit) {};
          \node[small] (b25) at (3 * \xunit, 4.4 * \yunit) {};
          \node[small] (b26) at (3 * \xunit, 6 * \yunit) {};

          % edges
          \draw (v6)--(b15);
          \draw (v6)--(b14);
          \draw (w6)--(b25);
          \draw (v4)--(b16);
          \draw (w4)--(b26);

          \draw (v5)--(b14);

          \draw[dashed] (v3)--(b14);
          \draw[dashed] (v3)--(b12);
          \draw[dashed] (v3)--(b15);
          \draw (v2)--(b12);
          \draw (v2)--(b14);
          \draw (v1)--(b12);
          \draw (w1)--(b22);

          \draw (b16)--(b25);
          \draw (b15)--(b26);
          \draw (b15)--(b24);
          \draw (b14)--(b25);

          \draw (b16)--(b24);
          \draw (b12)--(b24);
          \draw (b14)--(b22);
          \draw (b11)--(b22);
          \draw (b11)--(b24);

          \node[] at (1 * \xunit, 7 * \yunit) {$\{v^{(0)}\}$};
          \node[] at (2 * \xunit, 7 * \yunit) {\footnotesize$f_B^{-1}(0)$};
          \node[] at (3 * \xunit, 7 * \yunit) {\footnotesize$f_B^{-1}(1)$};
          \node[] at (4 * \xunit, 7 * \yunit) {$\{v^{(1)}\}$};
          \draw [decorate,decoration={brace,amplitude=5pt}]
          (1.5 * \xunit, 7.2 * \yunit) -- (3.5* \xunit, 7.2 * \yunit) node[midway,yshift=12pt]{$X_B$};

          \node[] at (0.7 * \xunit, 7.8 * \yunit) {$G'$};
      \end{tikzpicture}
  \end{minipage}%
  \caption{%
  An illustration of the auxiliary graph used in the proof of \Cref{lem:marked:separators:for:all:ARN}.
  The left figure shows an example graph $G$ with the given
  \occ $(X_B,X_C,X_R)$, disjoint sets $C_1, C_2 \subseteq X_C$, a proper $2$-coloring
  $f_B:X_B \to \{0,1\}$, and another (not necessarily proper) $2$-coloring $f_C:C_1 \to \{0,1\}$.
  Possibly overlapping terminals $A,R,N$ are determined as described in the proof.
  The right shows the auxiliary graph $G'$ for $G$,
  constructed from a copy of $G[X_B]$ with additional $2|X_C|$ vertices.
  Terminals are partitioned into $(A',R',N',T_X)$, where
  $T_X$ is deleted when we examine restricted $3$-way cuts.
  For both figures, the minimum $3$-way separators (in $G[X_B]$ and $G'-T_X$ (restricted), resp.) are shaped in double circles.
  }
  \label{fig:find-central}
\end{figure*}
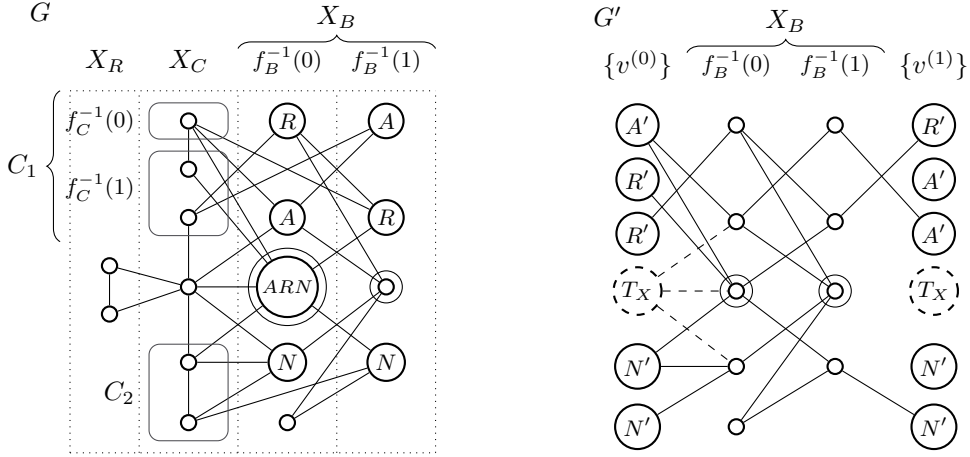

First, we briefly argue the following basic guarantees of \refmark.

\begin{lemma} \label{lem:mark:running:time}
	Let $G$ be a graph and let $(X_B, X_C, X_R)$ be an OCC in it. Then \refmark can be executed with these input parameters in $2^{\Oh(|X_C|)} n^{\Oh(1)}$ time and outputs a set~$B^\ast$ of size $|X_C|^{\Oh(1)}$.
\end{lemma}
\begin{proof}
	First we note that indeed $|B^\ast| = |X_C|^{\Oh(1)}$, since $B^\ast$ is constructed as a set of at most $c \cdot (2|X_C|)^8$ vertices. Next, it is clear to see that steps 1 and 2 can be done in polynomial time, and step 3 takes
	$2^{\Oh(|T|)} n^{\Oh(1)} = 2^{\Oh(|X_C|)} n^{\Oh(1)}$ time.
\end{proof}

Next, we formalize the main property we wished to hold for \refmark and proof that it is indeed satisfied.

\begin{lemma}
	\label{lem:marked:separators:for:all:ARN}
	Let $G$ be a graph, let $(X_B, X_C, X_R)$ be an OCC in it and let $B^\ast \subseteq X_B$ be the result of executing \refmark with these input parameters. Let $f_X : X_B \rightarrow \{0, 1\}$ be the proper $2$-coloring of $G[X_B]$ computed in the first step of \refmark. Let $C_1, C_2 \subseteq X_C$ be two arbitrary disjoint subsets of $X_C$ and let $f_C: C_1 \rightarrow \{0, 1\}$ be a (not necessarily proper) $2$-coloring of~$C_1$. If the~$\{A,R,N\}$-separation problem imposed onto~$G[X_B]$ by $(C_1, C_2, f_C, f_X)$ has a solution of size at most~$|X_C|$, then the set~$B^\ast$ contains an optimal solution to this separation problem.
\end{lemma}
\begin{proof}
	Recall that the auxiliary graph $G'$ consists of a copy of $G[X_B]$ and terminals
	$T = \{v^{(i)} \mid v \in X_C, i=0,1\}$.
	Given $C_1$, $C_2$ and $f_C$, consider the following generalized $3$-partition $\mathcal{T}$ of $T$
	into $(\emptyset, A',R',N',T_X)$, where $T_X$ represents deleted vertices.
	\begin{itemize}
		\item If $v \in C_1$ and $f_C(v)=0$, then let $v^{(0)} \in A'$ and $v^{(1)} \in R'$.
		\item If $v \in C_1$ and $f_C(v)=1$, then let $v^{(0)} \in R'$ and $v^{(1)} \in A'$.
		\item If $v \in C_2$, then let $v^{(0)},v^{(1)} \in N'$.
		\item If $v \in X_C \setminus (C_1 \cup C_2)$, then let $v^{(0)},v^{(1)} \in T_X$.
	\end{itemize}
	\Cref{fig:find-central} illustrates this construction.
	Note that for $v \in C_1$ we have $v^{(i)} \in A'$ ($v^{(i)} \in R'$, resp.)
	if and only if $f_C(v)=i$ ($f_C(v)=1-i$).
	We shall show the equivalence between restricted 3-way cuts for $\{A',R',N'\}$ in $G'-T_X$
	and $\{A,R,N\}$-separators in $G[X_B]$.
	
	First, let $S'$ be a restricted 3-way cut for $\{A',R',N'\}$ in $G'-T_X$, and
	assume that there is a path $sPt$ in $G[X_B]-S'$ where $s$ and $t$ belong to different sets in $\{A,R,N\}$.
	% \bmpr{This argument will have to be refined somewhat, since the statement of the used corollary was too strong. We do not keep \emph{all} separators with undeletable terminals; we can only keep them up to a certain size, such as~$|X_C|$.}
	%
	\begin{enumerate}[(1)]
		\item If $s \in A$ and $t \in R$,
		then there exist vertices $v_s \in N_G(s) \cap C_1$ and $v_t \in N_G(t) \cap C_1$
		such that $f_C(v_s)=f_X(s)$ and $f_C(v_t)\neq f_X(t)$.
		This implies $v_s^{(f_X(s))}s, v_t^{(f_X(t))}t \in E(G')$.
		By construction, $v_s^{(f_X(s))}=v_s^{(f_C(v_s))} \in A'$ and $v_t^{(f_X(t))}=v_t^{(1-f_C(v_t))} \in R'$,
		and these vertices are present in $G'-T_X$.
		Hence, $v_s^{(f_X(s))} s P t v_t^{(f_X(t))}$ is a path between $A'$ and $R'$ in $G'-T_X-S'$, a contradiction.
		
		\item If $s \in A$ and $t \in N$, then
		there exist vertices $v_s \in N_G(s) \cap C_1$ and $v_t \in N_G(t) \cap C_2$
		such that $f_C(v_s)=f_X(s)$.
		This implies $v_s^{(f_X(s))}s, v_t^{(f_X(t))}t \in E(G')$.
		By construction, $v_s^{(f_X(s))}=v_s^{(f_C(v_s))} \in A'$ and $v_t^{(f_X(t))} \in N'$.
		Hence, $v_s^{(f_X(s))} s P t v_t^{(f_X(t))}$ is a path between $A'$ and $N'$ in $G'-T_X-S'$, again a contradiction.
	\end{enumerate}
	Other cases are proven symmetrically.
	
	Conversely, let $S$ be an $\{A,R,N\}$-separator in $G[X_B]$, and
	assume that $S$ is not a restricted $3$-way cut for $\{A',R',N'\}$ in $G' - T_X$.
	Then, since terminals in $G'$ are independent,
	there is a path $s^{(i)} P t^{(j)}$ in $G'-T_X-S$ such that
	$s^{(i)}$ and $t^{(j)}$ belong to different sets in $\{A',R',N'\}$ and $P$ is a non-empty path in $G'[X_B]$.
	Let $s'$ and $t'$ be the first and the last vertices in $P$ respectively ($s'$ and $t'$ can be the same vertex).
	%
	% Note that $s',t' \in X_B$.
	We will show that $s'$ and $t'$ belong to different sets in $\{A,R,N\}$ in $G[X_B]$,
	which leads to a contradiction.
	% \bmp{I think there is a slight hole in this argument that needs patching. It may be that~$P$ contains vertices that do \emph{not} belong to~$G[X_B]$. In this case, having vertices~$s', t'$ belonging to different sets of~$G[X_B]$ does not lead to a contradiction because there is no corresponding path between those vertices in~$G[X_B] - S$ if~$P$ goes outside of~$X_B$. The latter happens when~$P$ contains \emph{more than two} terminals, which is possible in principle. But we can argue as follows: whenever you have a path in~$G'$ connecting two terminals belonging to different sets in the partition, either that path itself or some subpath of it connects two terminals belonging to different sets in the partition, while not containing any terminals in its interior. For paths with the latter property, you can indeed write them as~$sPt$ with~$P$ a path in~$G[X_B]$, for which the rest of the argument applies.}
	\begin{enumerate}[(1)]
		\item If $s^{(i)} \in A'$ and $t^{(j)} \in R'$, then $s, t \in C_1$
		and $f_C(s)=i$, $f_C(t)=1-j$.
		Also, knowing that $f_X(s')=i$, $f_X(t')=j$,
		we have $s' \in A$ and $t' \in R$ in $G[X_B]$.
		\item If $s^{(i)} \in A'$ and $t^{(j)} \in N'$, then $s \in C_1$,
		$t \in C_2$, and $f_C(s)=i$.
		Since $f_X(s')=i$, we have $s' \in A$.
		From $t't^{(j)} \in E(G')$, we have $t' \in N$.
	\end{enumerate}
	Other cases are symmetrical to these.
	
	Now, suppose there exists a smallest $\{A,R,N\}$-separator of size at most $|X_C|$ in $G[X_B]$.
	Then, there exists a restricted $3$-way cut of size at most $|X_C|$ for $\{A',R',N'\}$ in $G'-T_X$.
	From \Cref{lem:covering:corollary},
	we know that $B^\ast$ contains a minimum restricted $3$-way cut $Z$ of size at most $|X_C| \leq |T|$ for $\{A',R',N'\}$ in $G'-T_X$.
	Finally, $Z$ is indeed an optimal solution to the $\{A,R,N\}$-separation problem.
\end{proof}
Using this statement, we can prove \cref{lem:marked:tight:occ}.
\lemMarkedTightOcc*
\begin{proof}
	Suppose there exists a $z$-tight OCC $(A_B, A_C, A_R)$ in $G$. Let $f_X \colon X_B \rightarrow \{0,1\}$ and $f_A \colon A_B \rightarrow \{0,1\}$ be proper $2$-colorings of $G[X_B]$ and $G[A_B]$ respectively. To define $(A_B^\ast, A_C^\ast, A_R^\ast)$, we consider the separation problem imposed onto $G[X_B]$ by $(X_C \cap A_B, X_C \cap A_R, f_A, f_X)$ and let $A$, $R$ and $N$ be the three sets to be separated in this problem. 
	
	By \cref{lem:tight:occ:cut:optimal:ARN:separator}, $A_C \cap X_B$ is an $\{A, R, N\}$-separator in $G[X_B]$ and by \cref{lem:small:occ:intersection}, this intersection is at most as large as $X_C$. This satisfies the preconditions to \cref{lem:marked:separators:for:all:ARN}, from which we derive that $B^\ast$ contains some minimum-size $\{A, R, N\}$-separator $S$ in $G[X_B]$. We use this set $S$ to define $(A_B^\ast, A_C^\ast, A_R^\ast)$.
	
	First, let $A_C^\ast := (A_C \setminus X_B) \cup S$. To define $A_B^\ast$, let $U$ be the set of vertices from $X_B \setminus S$ that are not reachable from $N$ in $G[X_B - S]$. Now, we define $A_B^\ast := (A_B \setminus X_B) \cup U$. Finally, $A_R^\ast$ is simply defined as $V(G) \setminus (A_B^\ast \cup A_C^\ast)$. The remainder of the proof is used to show that the following five properties hold:
	\begin{enumerate}[(i)]
		\item \label{itm:marked:occ:1} $G[A_B^\ast]$ is bipartite.
		\item \label{itm:marked:occ:2} There are no edges between $A_B^\ast$ and $A_R^\ast$.
		\item \label{itm:marked:occ:3} $|A_C^\ast| = |A_C|$.
		\item \label{itm:marked:occ:4} $(A_B^\ast, A_C^\ast, A_R^\ast)$ has a certificate of order $z$.
		\item \label{itm:marked:occ:5} $A_C^\ast \cap X_B \subseteq B^\ast$.
	\end{enumerate}
	\bmpr{Clear proof setup, nice! I would personally do the simplest one (currently V) first, but that's a matter of taste.}The first two properties show that $(A_B^\ast, A_C^\ast, A_R^\ast)$ is an OCC. Together with properties \textbf{(\ref{itm:marked:occ:3})} and \textbf{(\ref{itm:marked:occ:4})}, they show that this decomposition is even a tight OCC whose order and width are equal to those of $(A_B, A_C, A_R)$. Finally, the lemma explicitly requires that $A_C^\ast \cap X_B \subseteq B^\ast$, which corresponds to property \textbf{(\ref{itm:marked:occ:5})}.
	\proofsubparagraph{(\ref{itm:marked:occ:1}) To prove: $G[A_B^\ast]$ is bipartite.} Since $A_B^\ast$ is a subset of $A_B \cup (X_B \setminus S)$, by showing that $G[A_B \cup (X_B \setminus S)]$ is bipartite, we show that $G[A_B^\ast]$ is bipartite.
	\begin{claim} \label{clm:ab:xb:bipartite}
		$G[A_B \cup (X_B \setminus S)]$ is bipartite.
	\end{claim}
	\begin{claimproof}
		Because $S$ is a $\{A, R, N\}$-separator in $G[X_B]$, it is in particular an $\{A, R\}$-separator in $G[X_B]$. Consider now the graph $G[(X_B \setminus S) \cup (X_C \cap A_B)]$, let $W_0 := f_A^{-1}(0)$ and $W_1 := f_A^{-1}(1)$ and let $c = f_X$. Then, it follows from \cref{lem:AR-separation} that $G[(X_B \setminus S) \cup (X_C \cap A_B)]$ admits a proper $2$-coloring $f_X^\ast \colon (X_B \setminus S) \cup (X_C \cap A_B) \rightarrow \{0,1\}$ with $f_X^\ast(v) = f_A(v)$ for every $v \in (X_C \cap A_B)$.
		
		As such, the graphs induced by $A_B \setminus X_B$ and and $(X_B \setminus S) \cup (X_C \cap A_B)$ are both bipartite as they admit proper $2$-colorings $f_A$ and $f_X^\ast$ respectively. These vertex sets intersect precisely in $X_C \cap A_B$, on which the two colorings agree and which separates the remaining vertices from the two sets. Hence, by \cref{lem:separated:colorings:make:bipartite}, the graph induced by the union of these sets, $G[A_B \cup (X_B \setminus S)]$, is bipartite.
	\end{claimproof}
	\proofsubparagraph{(\ref{itm:marked:occ:2}) To prove: there are no edges between $A_B^\ast$ and $A_R^\ast$.} We take an arbitrary vertex $v \in A_B^\ast$ and show that it has no neighbors in $A_R^\ast$. I.e.: all its neighbors live in $A_B^\ast \cup A_C^\ast$. To this end, we distinguish three cases.
	\begin{itemize}
		\item \textbf{Case 1:} Suppose $v \in A_B \cap X_R$. Because $A_C$ separates $A_B$ from $A_R$, vertex $v$ has no neighbors in $A_R$. Likewise, because $X_C$ separates $X_B$ from $X_R$, vertex $v$ has no neighbors in $X_B$. The neighborhood of $v$ is therefore contained in $V(G) \setminus (A_R \cup X_B)$. By construction, all these vertices are in $A_B^\ast \cup A_C^\ast$.
		
		\item \textbf{Case 2:} Suppose $v \in A_B \cap X_C$. Let $u$ be a neighbor of $v$. Since $A_C$ separates $A_B$ from $A_R$, vertex $u$ must be in $A_B \cup A_C$. We distinguish two subcases.
		\begin{itemize}
			\item \emph{Case 2a:} Suppose $u \in (A_B \cup A_C) \setminus X_B$. Then $v$ is in $A_B^\ast \cup A_C^\ast$ by construction.
			\item \emph{Case 2b:} Suppose $u \in (A_B \cup A_C) \cap X_B$. Recall that $A$ and $R$ are defined in such a way that all neighbors of $A_B \cap X_C$ in $X_B$ end up in $A \cup R$. In particular, this means that $u \in A \cup R$. If $u \notin S$, then $u$ is not reachable from $N$ in $G[X_B] - S$, because $S$ is an $\{A \cup R, N\}$-separator in $G[X_B]$. In that case, $u$ is in $A_B^\ast$. Otherwise, if $u \in S$, then $u$ is in $A_C^\ast$.
		\end{itemize}
		
		\item \textbf{Case 3:} Suppose $v \in X_B$. Let $u$ be a neighbor of $v$. Since $X_C$ separates $X_B$ from $X_R$, $u$ must be in $X_B \cup X_C$. We distinguish four subcases.
		\begin{itemize}
			\item \emph{Case 3a:} Suppose $u \in X_C \cap A_R$. Then, because $u$ and $v$ are neighbors, we find that $v \in N$ by definition of $N$. This contradicts the fact that $u \in A_B^\ast$, since $A_B^\ast \cap X_B$ only contains vertices that are not reachable from $N$ in $G[X_B] - S$. Hence, this case does not occur.
			\item \emph{Case 3b:} Suppose $u \in X_C \setminus A_R$. Then by construction $u \in A_B^\ast \cup A_C^\ast$.
			\item \emph{Case 3c:} Suppose $u \in S$. Then by construction $u \in A_C^\ast$.
			\item \emph{Case 3d:} Suppose $u \in X_B \setminus S$. Then because $v$ is not reachable from $N$ in $G[X_B] - S$, neither are its neighbors that also live in $X_B \setminus S$. In particular, this means that $u$ is not reachable from $N$ in $G[X_B] - S$, so $u \in A_B^\ast$.
		\end{itemize}
	\end{itemize}
	To see that the three cases above are exhaustive, recall that $A_B^\ast$ was defined as the union of $A_B \setminus X_B$ and $U$. The former is covered by $A_B \cap X_R$ and $A_B \cap X_C$ since ($A_B, A_C, A_R)$ is a partition of $V(G)$. The latter is defined as a subset of $X_B$.
	\proofsubparagraph{(\ref{itm:marked:occ:3}) To prove: $|A_C^\ast| = |A_C|$.} We start by observing that $|S| = |A_C \cap X_B|$, since $S$ and $A_C \cap X_B$ are both minimum-size $\{A, R, N\}$-separators in $G[X_B]$ (respectively by assumption and by \cref{lem:tight:occ:cut:optimal:ARN:separator}). Since $A_C^\ast = (A_C \setminus X_B) \cup S$, it follows that $|A_C^\ast| = |A_C|$.
	\proofsubparagraph{(\ref{itm:marked:occ:4}) To prove: $(A_B^\ast, A_C^\ast, A_R^\ast)$ has a certificate of order $z$.} By definition, $(A_B, A_C, A_R)$ has a certificate of order $z$. Let $D$ be such a certificate and recall that $D$ is then a subgraph of $G[A_B \cup A_C]$ with components $D_1, \ldots, D_q$ that each have a minimum size OCT of size at most~$z$.
	
	We prove property \textbf{(\ref{itm:marked:occ:4})} to hold by showing that $D$ is also an order-$z$ certificate for $(A_B^\ast, A_C^\ast, A_R^\ast)$. For this, we need to show that: 
	\begin{itemize}
		\item $D$ is a subgraph of $G[A_B^\ast \cup A_C^\ast]$,
		\item $A_C^\ast$ is a smallest OCT of $D$, and
		\item for every $i \in [q]$, $D_i$ has an OCT of size at most $z$.
	\end{itemize}
	We continue by proving each of these properties separately.
	\begin{claim}
		$D$ is a subgraph of $G[A_B^\ast \cup A_C^\ast$].
	\end{claim}
	\begin{claimproof}
		It suffices to show that $V(D) \subseteq A_B^\ast \cup A_C^\ast$, which we do by proving that $(A_B \cup A_C) \subseteq (A_B^\ast \cup A_C^\ast)$. For contradiction, suppose that there is some vertex $v \in A_B \cup A_C$ that is not in $A_B^\ast \cup A_C^\ast$. 
		
		Since $(A_B \cup A_C) \setminus X_B \subseteq (A_B^\ast \cup A_C^\ast)$ by construction, it follows that $v \in X_B$. Moreover, because $S \subseteq A_C^\ast$, we find that $v \in X_B \setminus S$. By construction of $A_B^\ast$, vertices in $X_B \setminus S$ are in $A_B^\ast$ if and only if they are not reachable from $N$ in $G[X_B] - S$. Since $v \notin A_B^\ast$ by assumption, there must be a path $P_v$ from $N$ to $v$ in $G[X_B] - S$.
		
		The remainder of the proof of this claim is used to show that the existence of $P_v$ implies the existence of a path from $N$ to $A \cup R$ in $G[X_B] - S$. This would contradict the fact that $S$ is an $\{A, R, N\}$-separator in $G[X_B]$. To this end, first note that $N$, which only contains neighbors of vertices in $A_R$, is a subset of $A_R \cup A_C$. As $P_v$ is a path from $N$ to $v \in A_B \cup A_C$, it contains at least one vertex from $A_C$: if none of its endpoints are in $A_C$ themselves, they lie in the sets $A_R$ and $A_B$ respectively, which are separated by $A_C$.
		
		Let $a$ be the first vertex from $A_C$ encountered when traversing $P$ from its endpoint in $N$ to $v$. Now, because $P_v$, which lives in $G[X_B] - S$, already contains a path from $N$ to $a$ as a subpath, it suffices to show that that there exists a path from $a$ to $A \cup R$ in $G[X_B] - S$.
		
		\bmpr{I commented out Menger's theorem and moved the sourcecode for it here since it was used only once in the paper, the entire research community knows it, and this proof will end up in the appendix anyway. We could resurrect the statement for the journal version, but the inline citation that I put in the text may suffice even in the journal.}
		% 	
		%\begin{theorem}[{\cite[Theorem 9.1]{Schrijver03}}] \label{thm:menger}
		%Let $G$ be an undirected graph and let $S, T \subseteq V(G)$. The maximum number of vertex-disjoint $(S, T)$-paths is equal to the minimum size of an $\{S, T\}$-separator.
		%\end{theorem}
		
		Recall that $A_C \cap X_B$ is a minimum-size $\{A, R\}$-separator in $G[X_B]$ by \cref{lem:tight:occ:cut:optimal:ARN:separator}. By Menger's theorem~\cite[Theorem 9.1]{Schrijver03}, the minimum size of an~$\{A,R\}$-separator in~$G[X_B]$ equals the maximum cardinality of a packing of pairwise vertex-disjoint~$\{A,R\}$ paths in~$G[X_B]$.  Therefore, there is a collection $\mathcal{P}$ of $|A_C \cap X_B|$ pairwise disjoint $(A, R)$-paths in $G[X_B]$. As such, every $\{A, R\}$-separator in $G[X_B]$ of minimum size consists of exactly one vertex from each path in $\mathcal{P}$. From this, it follows that every vertex in such a separator lies on a unique path in $\mathcal{P}$. Recall that the vertex $a$ lives in the minimum-size $\{A, R\}$-separator $A_C \cap X_B$, so let $P_a$ be the unique path from $\mathcal{P}$ that it lies on.
		
		Next, we recall that $|S| = |A_C \cap X_B|$, as observed while proving property \textbf{(\ref{itm:marked:occ:3})}. Because $S$ is an $\{A, R, N\}$-separator in $G[X_B]$, it is in particular an $\{A, R\}$-separator in $G[X_B]$. And because its size is equal to $A_C \cap X_B$, which is a minimum-size $\{A, R\}$-separator by \cref{clm:tight:occ:cut:optimal:AR:separtor}, we find that $S$ is a minimum-size $\{A, R\}$-separator as well. Therefore, it contains exactly one vertex from every path in $\mathcal{P}$ and in particular from $P_a$.
		
		However, we find that $a \notin S$ because $a$ was chosen as a vertex on the path $P_v$, which lives in $G[X_B] - S$. As $P_a$ is an $(A, R)$-path in $G[X_B]$ and $S$ contains only one vertex from it (which is not $a$) it follows that $P_a - S$ contains either an $(a, A)$-path or an $(a, R)$-path as subpath.
		
		By concatenating the $(N, a)$-subpath of $P_v$ with an $(a,A)$-subpath or $(a,R)$-subpath of $P_a - S$, we obtain an $(N, A \cup R)$-path in $G[X_B] - S$. This contradicts the fact that $S$ is an $\{A, R, N\}$-separator and concludes the proof.
	\end{claimproof}
	Next, we show that $A_C^\ast$ is a smallest OCT of the certificate $D$.
	\begin{claim}
		$A_C^\ast$ is a smallest OCT of $D$.
	\end{claim}
	\begin{claimproof}
		Clearly, $A_C^\ast$ is an OCT of $D$: because $D - A_C^\ast$ is a subgraph of $A_B^\ast$, which is bipartite by property \textbf{(\ref{itm:marked:occ:1})}. Furthermore, we know that $A_C^\ast$ and $A_C$ are equally large by property \textbf{(\ref{itm:marked:occ:3})}. Because $A_C$ is a smallest OCT of $D$ by definition of $D$, we find that $A_C^\ast$ is also a smallest OCT of $D$.
	\end{claimproof}
	Finally, for every $i \in [q]$, we know that $D_i$ has an OCT of size at most $z$, by definition of $D$ being an order-$z$ certificate of $(A_B, A_C, A_R)$. 
	As this was the third and final property required to show that $D$ is an order-$z$ certificate of $(A_B^\ast, A_C^\ast, A_R^\ast)$, this concludes the proof of property \textbf{(\ref{itm:marked:occ:4})}.
	\proofsubparagraph{(\ref{itm:marked:occ:5}) To prove: $A_C^\ast \cap X_B \subseteq B^\ast$.} We know that $A_C^\ast \cap X_B = S$ by construction of $A_C^\ast$ and we  have $S \subseteq B^\ast$ by definition of $S$.
\end{proof}
\section{Proofs regarding core properties of \refreduction} \label{sec:reduction:proofs}
First we prove that \refreduction can be executed in FPT time parameterized by $|X_C|$.
\begin{lemma} \label{lem:shrink:running:time}
	\Refreduction can be performed in $2^{\Oh(|X_C|)} \cdot n^{\Oh(1)}$ time.
\end{lemma}
\begin{proof}
	Executing \refmark in step~\ref{itm:shrink:mark:b:star} of the algorithm takes $2^{\Oh(|X_C|)} \cdot n^{\Oh(1)}$ time. We proceed by arguing that the remainder of the algorithm can be performed in $n^{\Oh(1)}$ time.
	
	The loop from step~\ref{itm:shrink:loop} forms the computational bottle-neck of the remaining part of the algorithm. It consists of $2 \cdot |X_C \cup B^\ast|^2$ iterations. As $|B^\ast| = |X_C|^{\Oh(1)}$ by construction, this results in $|X_C|^{\Oh(1)}$ iterations. In each of these iterations, we need to check whether $G[X_B \setminus B^\ast]$ connects two given vertices $u$ and $v$ using a path of a given parity $p$. 
	
	This check can be done by iterating over all components $H$ of $G[X_B \setminus B^\ast]$ and returning whether at least one of those connects $u$ and $v$ with a path of parity $p$. Now, it suffices to show that this can be computed for a single arbitrary such component $H$ in polynomial time.
	
	Because every component $H$ of $G[X_B \setminus B^\ast]$ is a connected bipartite subgraph, checking whether an arbitrary such component connects two given vertices $u$ and $v$ with a path of a given parity $p$ is in fact quite simple. If $p = \text{even}$, it suffices to return whether $u$ and $v$ each have a neighbor in the same partite set of $H$, which can clearly be done in polynomial time. To see that this suffices, note that any even path in $H$ starts and ends in the same partite set, since $H$ is bipartite. Moreover, because~$H$ is connected, any two vertices in it are connected by a path.
	
	By similar arguments, if $p = \text{odd}$, it suffices to return whether $u$ and $v$ have neighbors in opposite partite sets of $H$.
\end{proof}
Next, we show that \refreduction outputs a strictly smaller graph than its input graph when the OCC it is given is reducible. Recall that we call an OCC $(X_B, X_C, X_R)$ reducible if $|X_B| > g_r(|X_C|)$ for some polynomial $g_r(x)$. The exact definition of this polynomial had been postponed so far. Having reached the statement whose proof relies on this definition, we now specify this exact polynomial.

Its definition is based on \Cref{lem:covering:corollary} in which sets $Z$ and $T$ are specified. Setting the value of $s$ in this lemma to $3$ yields the existence of a constant $c \in \N$ such that $|Z| \leq c \cdot |T|^{8}$ for large enough $|T|$. Given this constant $c$, we define $g_r \colon \N \to \N$ as $g_r(x) = (6(2^8 c+1)^2+2^8 c) \cdot x^{16}$. With this definition, we can prove the following result. 
\begin{lemma} \label{lem:shrink:size}
	Let $G$ be a graph and let $(X_B, X_C, X_R)$ be an OCC of it. If $(X_B, X_C, X_R)$ is reducible, then running \refreduction with these input parameters yields an output graph~$G'$ with $|V(G')| < |V(G)|$.
\end{lemma}
\begin{proof}
	Note that by construction of $G'$, we have $X_C \cup X_R \subseteq V(G) \cap V(G')$, so that we can define $X_B' := V(G') \setminus (X_C \cup X_R)$. Since $X_B = V(G) \setminus (X_C \cup X_R)$, it suffices to show that $|X_B'| < |X_B|$. To do so, we separately bound the sizes of $B^\ast$ and $V(G') \setminus V(G)$ of which $X_B'$ is the union. 
	
	We express bounds in terms of $|X_C|$ and the values of $c$ and $g_r$ as defined above. Bounding the size of $B^\ast$ using these terms is simple: as it was constructed in \refmark via \Cref{lem:covering:corollary}, it has at most $c \cdot (2 |X_C|)^8$ vertices.
	
	To bound the size of $V(G') \setminus V(G)$, note that it includes at most $6$ vertices for every pair of vertices from $X_C \cup B^\ast$. Hence, $|V(G') \setminus V(G)| \leq 6 \cdot |X_C \cup B^\ast|^2 = 6 \cdot (|X_C| + |B^\ast|)^2$. Inserting the bound $|B^\ast| \leq c \cdot (2 |X_C|)^8$ into this expression, we obtain $|V(G') \setminus V(G)| \leq 6 \cdot \left((2^8 c+1) \cdot |X_C|^8 \right)^2 = 6 \cdot (2^8 c+1)^2 \cdot |X_C|^{16}$.
	
	Combining the bounds on $|B^\ast|$ and $|V(G') \setminus v(G)|$, we get $|X_B'| \leq \left(6(2^8 c+1)^2+2^8 c\right) \cdot |X_C|^{16} = g_r(|X_C|)$. As $(X_B, X_C, X_R)$ is a reducible OCC by assumption, it follows that $|X_B| > g_r(|X_C|) \geq |X_B'|$.
\end{proof}

\section{Safety proofs of \refreduction}
To prove the two safety lemmas from \cref{sec:reduce:occ}, we first introduce the two main ingredients that the proofs of these statements share. First, we show that the reduction ``preserves'' the structure of odd cycles in the graph.

\begin{lemma} \label{lem:odd:cycle:preservation}
	Let~$G$ be a graph, let $(X_B, X_C, X_R)$ be an OCC in it and let~$G'$ be the graph obtained by running \refreduction with these input parameters. Let $Y \subseteq V(G) \cap V(G')$ be arbitrary. Then, there is a closed odd walk~$F$ in~$G$ with $V(F) \cap V(G) \cap V(G') = Y$ if and only if there is a closed odd walk~$F'$ in~$G'$ with $V(F') \cap V(G) \cap V(G') = Y$.
\end{lemma}
\begin{proof}
	Let $B^\ast \subseteq X_B$ be the set that is marked in step~\ref{itm:shrink:mark:b:star} of the reduction.
	
	For one direction of the proof, let $F$ be a closed odd walk in $G$ such that $V(F) \cap V(G) \cap V(G') = Y$. If~$V(F) = Y$, then $F$ is completely contained in $V(G) \cap V(G')$, meaning that it is also a closed odd walk in $G'$. Suppose therefore that $V(F) \neq Y$. Consider any subwalk $\widehat{F}$ of $F$ such that $\widehat{F}$ is contained in $G - V(G') = G[X_B \setminus B^\ast]$, but such that the preceding vertex $u$ and succeeding vertex $v$ are in $V(G) \cap V(G')$. As all neighbors of $X_B \setminus B^\ast$ in $G$ are in $X_C \cap B^\ast$, $u$ and $v$ are in particular in this latter set.
	
	Since $G[X_B \setminus B^\ast]$ connects $u$ and $v$ with $\widehat{F}$ in $G$, the subgraph~$G' - V(G)$ connects $u$ and $v$ with a path $\widehat{F}'$ of the same parity as $\widehat{F}$, due to step~\ref{itm:shrink:loop} of \refreduction. We can use $\widehat{F}'$ as a replacement for $\widehat{F}$ in $G'$. Since $\widehat{F}$ was chosen as an arbitrary subwalk of $F$ that is contained in $G - V(G')$ and that connects two vertices from $V(G) \cap V(G')$, such a replacement exists for all these subwalks of $F$. Hence, we can obtain a closed odd walk $F'$ in $G'$ from $F$ by copying it and replacing all such subwalks in $G - V(G')$ by a replacement walk of the same parity. This yields a closed odd walk $F'$ in $G'$ with $V(F') \cap V(G) \cap V(G') = V(F) \cap V(G) \cap V(G') = Y$.
	
	The proof for the other direction is mostly analogous. Let $F'$ be a closed odd walk in $G'$ with $V(F') \cap V(G) \cap V(G') = Y$. Let $\widehat{F}'$ be an arbitrary subwalk of $F'$ in $G' - V(G)$ that connects two vertices $u$ and $v$ from $V(G') \cap V(G)$. Due to step~\ref{itm:shrink:loop} of \refreduction, this subwalk can only exist in $G' - V(G)$ if $G[X_B \setminus B^\ast]$ connects $u$ and $v$ with a path $\widehat{F}$ of the same parity as $\widehat{F}'$. Similar to the previous paragraph, $\widehat{F}$ can serve as a replacement for $\widehat{F}'$ when constructing a closed odd walk $F$ in $G$ from $F'$. Doing so for all such subwalks $\widehat{F}'$ yields a closed odd walk $F$ in $G$ with $V(F) \cap V(G) \cap V(G') = V(F') \cap V(G) \cap V(G') = Y$.
\end{proof}

The lemma below formalizes the idea that no minimum-size OCT of the graph $G'$ constructed by \refreduction contains vertices that were newly added to the construction. For ease of use in later proofs, the statement is more general than required, but to capture this original idea, one may think of $H$ as being $G'$ and $P_1$ and $P_2$ being two newly added paths in step~\ref{itm:shrink:loop} of the algorithm connecting two vertices $u$ and $v$ with the same parity.

\begin{lemma} \label{lem:new:vertices:not:optimal}
	Let $H$ be a graph and let $P_1$ and $P_2$ be two internally vertex-disjoint paths of the same parity that connect the same endpoints $u$ and $v$ and whose internal vertices all have degree $2$ in $H$. If $S$ is a minimum-size OCT of $H$, then $S \cap (V(P_1) \cup V(P_2)) \subseteq {u, v}$. 
\end{lemma}
\begin{proof}
	\newcommand{\Vint}{V_{\text{int}}}
	Equivalently, we prove that $S$ does not contain internal vertices from $P_1$ and $P_2$. To this end, we let $\Vint(P)$ denote the set of internal vertices of a path $P$, so that e.g. $\Vint(P_1) = V(P_1) \setminus \{u, v\}$. Now, suppose for contradiction that $S \cap (\Vint(P_1) \cup \Vint(P_2)) \neq \emptyset$. We distinguish two cases.
	
	Suppose that $|S \cap (\Vint(P_1) \cup \Vint(P_2))| \geq 2$. Since the internal vertices of $P_1$ and $P_2$ all have degree $2$, every odd cycle in $H$ that intersects $P_1$ or $P_2$ in particular contains $u$. Hence, we can construct a smaller OCT of $H$ as $(S \setminus (\Vint(P_1) \cup \Vint(P_2))) \cup \{u\}$. This contradicts the optimality of $S$.
	
	Suppose that $|S \cap (\Vint(P_1) \cup \Vint(P_2))| = 1$. Assume w.l.o.g.~that $S$ contains one vertex from $\Vint(P_1)$ and let $p_1$ be this vertex. Since $S$ is an OCT of $H$, but it does not contain vertices from $\Vint(P_2)$, there are no odd cycles in $H - S$ that intersect $\Vint(P_2)$. In particular, there are no odd cycles in $H - (S \setminus \{p_1\})$ that intersect $\Vint(P_2)$, because no odd cycle in $H$ can contain vertices from both $\Vint(P_1)$ and $\Vint(P_2)$. But every odd cycle in $H$ that intersects $P_1$ can be transformed into one that intersects $P_2$, by simply replacing the subpath $P_1$ by $P_2$. Hence, there are also no odd cycles in $H - (S \setminus \{p_1\})$ that intersect $\Vint(P_1)$ or $p_1$ in particular. As such, $S \setminus \{p_1\}$ is a strictly smaller OCT of $H$, contradicting the optimality of $S$.
\end{proof}

Using these two lemmas, we can prove the safety of our reduction.

\lemForwardSafety*
\begin{proof}
	Let~$G$ be a graph, let $(X_B, X_C, X_R)$ be an OCC in it and let~$G'$ be the graph obtained by running \refreduction with these input parameters. Let~$B^\ast \subseteq X_B$ be the set that is marked in step~\ref{itm:shrink:mark:b:star} of the reduction and let~$Q_{u, v, p}$ be the set of newly added vertices for a given choice of~$u$,~$v$ and~$p$ in step~\ref{itm:shrink:loop}. We allow such a set to be empty if $G[X_B \setminus B^\ast]$ does not connect $u$ and $v$ with a path of parity $p$.
	
	Now suppose there is a $z$-tight OCC $(A_B, A_C, A_R)$ in $G$. Then, by \cref{lem:marked:tight:occ}, there is in particular a $z$-tight OCC $(A_B^\ast, A_C^\ast, A_R^\ast)$ in $G$ with $|A_C^\ast| = |A_C|$ and $A_B^\ast \cap X_B \subseteq B^\ast$. To prove the forward safety of the reduction, we use this tight OCC in $G$ to construct a tight OCC $(A_B', A_C', A_R')$ in $G'$ with the same width and order.
	
	We start the construction of this tight OCC by setting $A_C' := A_C^\ast$. Next, we construct~$A_B'$ and initialize it as $A_B^\ast \cap V(G')$. Then, for every $u, v \in ((A_B \cup A_C) \setminus X_R) \cap V(G')$ and every $p \in \{\text{even, odd}\}$, we add $Q_{u,v,p}$ to~$A_B'$. Finally, we simply set $A_R' := V(G') \setminus (A_B' \cup A_C')$.
	
	To show that this construction yields a tight OCC in $G'$ with the same width and order as $(A_B, A_C, A_R)$, we show that the following four properties hold:
	\begin{enumerate} [(i)]
		\item \label{itm:reduced:occ:1} $G'[A_B']$ is bipartite.
		\item \label{itm:reduced:occ:2} There are no edges between $A_B'$ and $A_R'$.
		\item \label{itm:reduced:occ:3} $|A_C'| = |A_C|$.
		\item \label{itm:reduced:occ:4} $(A_B', A_C', A_R')$ has a certificate of order $z$.
	\end{enumerate}
	The first two properties show that $(A_B', A_C', A_R')$ is an OCC. Together with properties \textbf{(\ref{itm:reduced:occ:3})} and \textbf{(\ref{itm:reduced:occ:4})}, they show that it is even a tight OCC whose width and order are equal to those of $(A_B, A_C, A_R)$.\bmpr{The formatting of the counters looks different. In the LIPICS style, the roman numbers seem printed in grey in the original enumeration, but they become boldface black when referencing them. My suggestion would be to avoid manually formatting the referencing, leaving the job to the $\mathsf{cref}$ command or similar.} 
	\proofsubparagraph{(\ref{itm:reduced:occ:1}) To prove: $G[A_B']$ is bipartite.} Suppose for contradiction that $G'[A_B']$ is not bipartite. Then, it contains an odd cycle $F'$. By invocation of  \cref{lem:odd:cycle:preservation} with $Y = V(F') \cap V(G) \cap V(G')$, there exists a closed odd walk $F$ in $G$ with $V(F) \cap V(G) \cap V(G') = Y$. Note that $Y$ cannot be empty as this would imply the odd cycle $F'$ to live in $G' - V(G)$, which is bipartite by construction.
	
	Now, because $Y \subseteq A_B' \cap V(G)$, it follows from the construction of $A_B'$ that $Y \subseteq A_B^\ast$. Then, since $V(F) \cap V(G) \cap V(G') \subseteq A_B^\ast$, the entirety of $F$ is contained in $A_B^\ast$: otherwise, $F$ must intersect $A_C^\ast$, but $A_C^\ast$ is contained in $V(G')$, of which the intersection with $V(F)$ was established to be a subset of $A_B^\ast$. Since $F$ is a closed odd walk in $G[A_B^\ast]$, this contradicts the fact that $(A_B^\ast, A_C^\ast, A_R^\ast)$ is an OCC, which requires that $G[A_B^\ast]$ is bipartite.
	\proofsubparagraph{(\ref{itm:reduced:occ:2}) To prove: There are no edges between $A_B'$ and $A_R'$.} Suppose for contradiction that there is an edge $\{b, r\}$ for $b \in A_B'$ and $r \in A_R'$. If $\{b, r\}$ also exists in $E(G)$, then $b$ and $r$ also exist in $V(G)$. By construction of $(A_B', A_C', A_R')$, it follows that $b \in A_B^\ast$ and $r \in A_R^\ast$. This contradicts the assumption that there are no edges between $A_B^\ast$ and $A_R^\ast$, which follows from $(A_B^\ast, A_C^\ast, A_R^\ast)$ being an OCC.
	
	Hence, $\{b, r\}$ does not exist in $E(G)$. By construction of $G'$ every newly added edge has at least one newly added vertex as endpoint. If $b$ and $r$ are both newly added during the reduction, then the fact that they are adjacent implies that they live in the same set $Q_{u,v,p}$ for some $u$, $v$ and $p$. By construction of $(A_B', A_C', A_R')$, every such set is either completely added to $A_B'$ or to $A_R'$. This means that $b$ and $r$ are either both in $A_B'$ or both in $A_R'$, contradicting the assumption that $b \in A_B'$ and $r \in A_R'$.
	
	Hence, exactly one of $b$ and $r$ is in $V(G)$. Assume w.l.o.g.~that $b \in V(G)$. By construction of $(A_B', A_C', A_R')$ and the fact that $b \in A_B'$, it follows that $b \in A_B^\ast$. Furthermore, since $r$ is in $V(G') \setminus V(G)$ and it is adjacent to $b$, there exist $v \in X_C \cup B^\ast$ and $p \in \{\text{even}, \text{odd}\}$ for which~$r \in Q_{b,v,p}$. Next, recall that $A_R' \setminus V(G)$ contains exactly the sets $Q_{u', v', p'}$ for which at least one of $u', v'$ is in $A_R^\ast$. Since, $r \in Q_{b,v,p}$ and $b \notin A_R^\ast$, the fact that $r \in A_R' \setminus V(G)$ implies that~$v \in A_R^\ast$.
	
	Because $Q_{b,v,p}$ is non-empty (it contains $r$), $b$ and $v$ are connected in $G$ by the subgraph $G[X_B \setminus B^\ast]$ with some path $P$. However, since $A_C^\ast$ separates $A_B^\ast (\ni b)$ and $A_R^\ast (\ni v)$ and is itself disjoint from $X_B \setminus B^\ast$, $b$ and $v$ live in different components of $G[X_B \setminus B^\ast]$. a subset of $B^\ast$, no such component can exist. This contradicts the existence of $P$.
	\proofsubparagraph{(\ref{itm:reduced:occ:3}) To prove: $|A_C'| = |A_C|$.} By construction of $A_C'$ and $A_C^\ast$, we have $A_C' = A_C^\ast$ and $|A_C^\ast| = |A_C|$.
	\proofsubparagraph{(\ref{itm:reduced:occ:4}) To prove: $(A_B', A_C', A_R')$ has a certificate of order $z$.} Let $D^\ast$ be an order-$z$ certificate of $(A_B^\ast, A_C^\ast, A_R^\ast)$ and let $D_1^\ast, \ldots, D_q^\ast$ be the connected components of $D^\ast$. From this certificate, we derive an order-$z$ certificate $D'$ of $(A_B', A_C', A_R')$ with connected components $D_1', \ldots, D_q'$. We explain how to construct an arbitrary such component $D_i'$, such that $D'$ is constructed simply as the disjoint union of all such components.
	
	We start by explaining how to construct an arbitrary such component $D_i'$ from $D_i^\ast$. First, we define its vertex set $V(D_i')$. We initialize this vertex set as $V(D_i^\ast) \cap V(G) \cap V(G')$. Next, for every pair $u$ and $v$ from this set and for every $p \in \{\text{even}, \text{odd}\}$ we add $Q_{u,v,p}$ to $V(D_i')$. This concludes the construction of $V(D_i')$. Now, we define $D_i'$ to be the subgraph of $G'$ induced by this vertex set, i.e.: $D_i' := G'[V(D_i')]$.
	
	We briefly argue that this construction ensures that $D_i'$ is a connected subgraph. Starting from the connected subgraph $D_i^\ast$, the construction starts with the removal of the vertices of $X_B \setminus B^\ast$. However, for each path through these deleted vertices, we insert a corresponding path via some $Q_{u, v, p}$ to provide the same connections. As all such paths we insert are themselves connected and adjacent to at least one vertex of $D_i^\ast \setminus (X_B \setminus B^\ast)$, this yields connectivity for the resulting subgraph $D_i'$.
	
	Also note that this construction ensures that the subgraphs $D_1', \ldots, D_q'$ have pairwise disjoint vertex sets. As such, we can obtain the certificate $D'$ simply as their disjoint union, so that $D_1', \ldots, D_q'$ are indeed the connected components of the certificate. It remains to show that $D'$ is an order-$z$ certificate of $(A_B', A_C', A_R')$. For this, we need to show that:
	\begin{itemize}
		\item $D'$ is a subgraph of $G[A_B' \cup A_C']$,
		\item $A_C'$ is a smallest OCT of $D'$, and
		\item for every $i \in [q]$, $D_i'$ has an OCT of size at most $z$.
	\end{itemize}
	The first property trivially holds by construction. The remainder of the proof is used to show that the other two properties also hold.
	
	\begin{claim}
		$A_C'$ is a smallest OCT of $D'$.
	\end{claim}
	\begin{claimproof}
		Clearly, $A_C'$ is an OCT of $D'$ since $D' - A_C'$ is a subgraph of $A_B'$, which is bipartite by property \textbf{(\ref{itm:reduced:occ:1})}. We prove the optimality of $A_C'$ by contradiction, so suppose that $S'$ is a strictly smaller OCT of $D'$.
		
		Since the newly added vertices in $V(G') \setminus V(G)$ come in pairs that form degree-$2$ paths connecting the same endpoints, \cref{lem:new:vertices:not:optimal} implies that $S'$ does not contain any of these vertices. Hence, $S' \subseteq V(G)$. We know that $S'$ is not an OCT of $G[A_B^\ast \cup A_C^\ast]$, because $A_C^\ast$ is a smallest OCT of this induced subgraph and $|S'| < |A_C'| = |A_C^\ast|$.
		
		Therefore, $G[A_B^\ast \cup A_C^\ast] - S'$ contains an odd cycle $F$. Let $Y := V(F) \cap V(G) \cap V(G')$. By \cref{lem:odd:cycle:preservation}, $G'$ contains a closed odd walk $F'$ with $V(F') \cap V(G) = Y$. This implies the existence of an odd cycle $F''$ with $V(F'') \subseteq V(F')$. We proceed by showing that $F''$ lives in $G'[A_B' \cup A_C'] - S'$.
		
		Clearly, $F''$ does not contain vertices from $S'$: since $S' \subseteq V(G) \cap V(G')$, these vertices would have to be in $Y$, but $Y$ is disjoint from $S'$ by its definition. It remains to show that $F'$ does not contain vertices from $A_R'$. 
		
		Suppose for contradiction that $F''$ contains a vertex $r \in A_R'$. First note that $Y$ is disjoint from $A_R'$\rfar{@Bart: this place used to contain a comment of yours that the argument was hard to follow. Is the updated argument any better?} as it is a subset of $(A_B^\ast \cup A_C^\ast) \cap V(G) \cap V(G')$. This means that $r$ cannot be in $V(G') \cap V(G)$, as this would cause it to be in both $A_R'$ and in $V(F'') \cap V(G) \cap V(G') \subseteq Y$, which are disjoint sets.
		
		As such, $r$ is in $V(G') \setminus V(G)$, meaning that there are some $u, v \in X_C \cup B^\ast$ and $p \in \{\text{even}, \text{odd}\}$ for which $r \in Q_{u, v, p}$.		
		%Because $V(F') \cap V(G) = Y$ and $Y$ is disjoint from $A_R'$ by its definition\bmp{[I don't follow this argument. We defined~$Y$ as those vertices from~$V(G) \cap V(G')$ used on walk~$F$ in~$G[A_B^\ast \cup A_C^\ast] - S'$. So the definition of~$Y$ was in terms of things in~$G$, rather than things like~$A'_R$ in~$G'$. Can we make the narrative tighter? Do we mean: if~$r \in Y$, then~$r$ belongs to~$V(G) \cap V(G')$, but we have~$A'_R \cap V(G) \cap V(G') = A^\ast_R \cap V(G) \cap V(G')$ by construction, so that~$r$ would belong to~$A^\ast_R$ which contradicts~$r \in Y$ with~$Y$ the vertices of~$V(G) \cap V(G')$ visited by walk~$Y$ in~$G[A_B^\ast \cup A_C^\ast] - S'$, since that walk avoids~$A_R^\ast$?]}, we find $r \notin Y$. As such, $r$ is in $V(G') \setminus V(G)$, meaning that there are some $u, v \in X_C \cup B^\ast$ and $p \in \{\text{even, odd}\}$ for which $r \in Q_{u,v,p}$.
		By construction, any cycle in $G'$ that contains a vertex from $Q_{u,v,p}$, also contains $u$ and $v$. In particular: $u, v \in V(F'')$. This means that $u, v \in Y$, which in turn implies that $u, v \in A_B^\ast \cup A_C^\ast$. By construction of $A_B'$, it follows that $Q_{u, v, p} \subseteq A_B'$, which contradicts the assumption that $r \in A_R'$.
	\end{claimproof}
	
	\begin{claim}
		For every $i \in [q]$, $D_i'$ has an OCT of size at most $z$.
	\end{claim}
	\begin{claimproof}
		For every $i \in [q]$, the graph~$D_i' - A_C'$ is a subgraph of the bipartite graph $G'[A_B']$, witnessing that $A_C' \cap V(D_i')$ is an OCT of $D_i'$. By construction of $D_i'$ we have~$A_C' \cap V(D_i') = A_C^\ast \cap V(D_i^\ast)$, the size of which is at most $z$ by the fact that $D^\ast$ is an order-$z$ certificate of $(A_B^\ast, A_C^\ast, A_R^\ast)$. 
	\end{claimproof}
	This concludes the proof of \cref{lem:forward:safety}.
\end{proof}

\lemBackwardSafety*
\begin{proof}
	As in the previous proof, let~$G$ be a graph, let $(X_B, X_C, X_R)$ be an OCC in it and let~$G'$ be the graph obtained by running \refreduction with these input parameters. Let~$B^\ast \subseteq X_B$ be the set that is marked in step~\ref{itm:shrink:mark:b:star} of the reduction. Now, let $S'$ be a smallest OCT of $G'$. 
	
	Since the vertices of $V(G') \setminus V(G)$ come in pairs that form degree-$2$ paths connecting the same endpoints, \cref{lem:new:vertices:not:optimal} implies that $S'$ does not contain any of these newly added vertices. Hence, $S' \subseteq V(G) \cap V(G')$. It remains to prove that $S'$ is a smallest OCT of $G$.
	
	To see that $S'$ is an OCT of $G$, suppose for contradiction that $G - S'$ contains an odd cycle $F$. Then, by invocation of  \cref{lem:odd:cycle:preservation} with $Y = V(F) \cap V(G) \cap V(G')$, there is a closed odd walk $F'$ in $G'$ with $V(F') \cap V(G) \cap V(G') = Y$. This set is disjoint from $S'$ by definition of $F$. Since $S' \subseteq V(G) \cap V(G')$ and $V(F') \cap V(G) \cap V(G)$ is disjoint from $S'$, the walk~$F'$ does not contain vertices from $S$. Now, the assumption that $S'$ is an OCT of $G'$ is contradicted by $F'$ being a closed odd walk in $G' - S'$.
	
	To prove that $S'$ is even a smallest OCT of $G$, we start by proving the following claim about smallest OCTs in $G$.
	
	\begin{claim} \label{clm:optimal:oct:exists:with:untouched:vertices}
		There is a minimum-size OCT $S$ of $G$ with $S \subseteq V(G) \cap V(G')$. 
	\end{claim}
	\begin{claimproof}
		Let $T$ be an arbitrary minimum-size OCT of $G$. We show how to modify $T$ into an OCT $S$ of $G$ that is at most as large as $T$ with $S \subseteq V(G) \cap V(G')$. To this end, let $f \colon V(G) \setminus T \rightarrow \{0,1\}$ be a proper $2$-coloring of $G - T$ and let $f_X \colon X_B \rightarrow \{0,1\}$ be a proper $2$-coloring of $G[X_B]$. Consider now the separation problem imposed onto $G[X_B]$ by $(X_C \setminus T, \emptyset, f, f_B)$. Let $A, R, N$ be the three sets to be separated in this problem and let their names correspond to their role as in \cref{def:imposed:separation:problem}. 
		
		First we show that $T \cap X_B$ is an $\{A, R, N\}$-separator of size at most $|X_C|$ in $G[X_B]$. Because $N = \emptyset$, it suffices to show that $T \cap X_B$ is an $\{A, R\}$-separator in $G[X_B]$. By invoking \cref{lem:AR-separation} on $G[(X_C \setminus T) \cup X_B]$ with $W_0 = f^{-1}(0) \cap X_C$, $W_0 = f^{-1}(1) \cap X_C$ and $c = f_X$, we see that the coloring $f$ (restricted to $G[(X_C \setminus T) \cup (X_B \setminus T)]$) serves as a witness that $T \cap X_B$ is indeed an $\{A, R\}$-separator in $G[X_B]$.
		
		We prove $T \cap X_B$ to be of size at most $|X_C|$ by contradiction, so suppose that $|T \cap X_B| > |X_C|$. Then $T' := (T \setminus X_B) \cup X_C$ is a strictly smaller set than $T$. To see that $T'$ is also an OCT of $G$, we consider an arbitrary odd cycle $F$ in $G$ and show that it intersects $T'$. If $F$ lives in $G[X_R]$, then it intersects $T \cap X_R = T' \cap X_R$. Otherwise, since $G[X_B]$ is bipartite, $F$ intersects $X_C$, which in turn is a subset of $T'$.
		
		Having established that there is an $\{A, R, N\}$-separator of size at most $|X_C|$, it follows from \cref{lem:marked:separators:for:all:ARN}, that there is a minimum-size $\{A, R, N\}$-separator $T^\ast \subseteq B^\ast$. In particular, $T^\ast$ is an $\{A,R\}$-separator.
		
		Now, we invoke \cref{lem:AR-separation} on the graph $G[(X_C \setminus T) \cup X_B]$ with $W_0 = f^{-1}(0) \cap X_C$, $W_0 = f^{-1}(1) \cap X_C$, $c = f_X$ and with the sets $A$ and $R$ in this lemma coinciding with the sets $A$ and $R$ obtained from the imposed separation problem defined in the first paragraph this proof (\cref{clm:optimal:oct:exists:with:untouched:vertices}). From this, we obtain that $G[(X_C \setminus T) \cup (X_B \setminus T^\ast)]$ has a proper $2$-coloring in which the vertices from $X_C \setminus T$ receive the same color as in $f$. 
		
		Then, the invocation of \cref{lem:separated:colorings:make:bipartite} on $G - ((T \setminus X_B) \cup T^\ast)$ with $V_L = X_R \setminus T$, $V_0 = X_C \setminus T$, $V_R = X_B \setminus T^\ast$ and $f_L = f_R = f$ yields that $G - ((T \setminus X_B) \cup T^\ast)$ is bipartite. As such, $S := (T \setminus X_B) \cup T^\ast$ is an OCT of $G$.
		
		It remains to show that $|S| \leq |T|$. Equivalently, we show that $|T^\ast| \leq |T \cap X_B|$. Recall that $T \cap X_B$ was shown to be an $\{A, R, N\}$-separator in $G[X_B]$. Since $T^\ast$ was shown to be such a separator of minimum size, it follows that $|T^\ast| \leq |T \cap X_B|$.
	\end{claimproof}
	Now, showing that there is no smaller OCT of $G$ than $S'$ can be done using arguments that have mostly been used before. 
	
	Suppose for contradiction that $S$ is a strictly smaller OCT of $G$ than $S'$. By \cref{clm:optimal:oct:exists:with:untouched:vertices}, we can assume w.l.o.g.~that $S \subseteq V(G) \cap V(G')$. As $S'$ was defined to be a smallest OCT of $G'$, the strictly smaller set $S$ is not an OCT of $G'$ so there exists an odd cycle $F'$ in $G' - S$. By \cref{lem:odd:cycle:preservation}, there is a closed odd walk $F$ in $G$ with $V(F) \cap V(G') = V(F') \cap V(G)$. These intersections are disjoint from $S$ by definition of $F'$, so the closed odd walk $F$ lives in $G - S$. This contradicts the assumption that $S$ is an OCT of $G$.
\end{proof}
\section{Proofs for Finding and Removing Tight OCCs} \label{sec:tight:occs-proof}

The algorithm to extract a \zocc from a given coloring~$\chi$ of the vertices and edges of~$G$, is inspired by a corresponding step in previous work~\cite[Lemma 6.2]{DonkersJ24}. The main idea is to iteratively refine the coloring, by changing the colors of vertices and edges into~$\cR$ when the current coloring does not justify their membership in a \zocc. It turns out that after exhaustively refining the coloring in this way, we can ensure the vertices colored~$\cB$ and~$\cC$ actually form the bipartite part and head of a \zocc (assuming their union is nonempty). The most computationally expensive step of the algorithm comes from verifying the requirement that the coloring highlights an~$A_C$-certificate, consisting of connected components that each have an OCT of size at most~$z$. To verify this property, the algorithm will iterate over all sets~$C$ of at most~$z$ vertices colored~$\cC$ and test whether they form an optimal OCT for the subgraph induced by~$C$ together with the vertex sets of connected components colored~$\cB$ whose neighborhood is a subset of~$\cC$. This leads to the~$n^{\Oh(z)}$ term in the running time.

\lemFindAntlerColored*

\begin{proof}
  Let $G_\chi := G - \chi^{-1}(\cR)$. Note that the value of this expression will change when the algorithm updates the coloring~$\chi$ below. 
  We define a function $W_\chi\colon 2^{\chi_V^{-1}(\cC)} \to 2^{\chi_V^{-1}(\cB)}$ as follows:
  for any $C \subseteq \chi_V^{-1}(\cC)$ let $W_\chi(C)$ denote the set of all vertices that are in
  a component $B$ of $G_\chi[\chi_V^{-1}(\cB)]$ for which $N_{G_\chi}(B) \subseteq C$.
  Intuitively, the function $W_\chi$ maps $\cC$-colored vertices $C$ to $\cB$-colored vertices $B$
  such that $B$ could be used to build a $C$-certificate.
  The following algorithm updates the coloring $\chi$ and recolors any vertex or edge
  that is not part of a \zocc to color $\cR$.

  \begin{enumerate}
      \item Recolor all edges $uv$ to color $\cR$ if $\chi(u)=\cR$ or $\chi(v)=\cR$.
          \label{alg:find-antler-colored:1}

      \item For each component $B$ of $G[\chi_V^{-1}(\cB)]$,
          recolor all vertices of $B$ and their incident edges to $\cR$
          if $G[B]$ is not bipartite or $N_G(B) \not\subseteq \chi_V^{-1}(\cC)$.
          \label{alg:find-antler-colored:2}

      \item For each subset $C \subseteq \chi_V^{-1}(\cC)$ of size at most $z$,
          mark all vertices in $C$ if $\oct(G_\chi[C \cup W_\chi(C)]) = |C|$.
          \label{alg:find-antler-colored:3}

      \item If $\chi_V^{-1}(\cC)$ contains unmarked vertices we recolor them to $\cR$,
          clear markings made in Step \ref{alg:find-antler-colored:3}
          and repeat from Step \ref{alg:find-antler-colored:1}.
          \label{alg:find-antler-colored:4}

      \item At this point, all vertices in $\chi_V^{-1}(\cC)$ are marked in Step \ref{alg:find-antler-colored:3}.
          If $\chi_V^{-1}(\cB) \cup \chi_V^{-1}(\cC) \neq \emptyset$, then 
          return $(\chi_V^{-1}(\cB), \chi_V^{-1}(\cC), \chi_V^{-1}(\cR))$ as a \zocc.
          \label{alg:find-antler-colored:5}

      \item Otherwise, report that $\chi$ does not $z$-properly color any \zocc.
          \label{alg:find-antler-colored:6}
  \end{enumerate}

  % Notice that the algorithm returns $(\emptyset, \emptyset, V(G))$
  % if no \zoccs are $z$-properly colored.
  %
  % \bmpr{
  %     % I don't think it is acceptable to omit any correctness argument for this algorithm. We may be able to get away with it in the conference version (if reviewers fail to see it is missing), but definitely for a journal version I do not think we could get the paper published without saying anything about correctness of this step. 
  % I suggest giving a succinct version of the correctness proof here and referring to the journal version of the previous paper for details.}

  \subparagraph*{Running time.}
  There will be at most $n$ iterations (Steps \ref{alg:find-antler-colored:1}-\ref{alg:find-antler-colored:4}) since in every iteration the number of vertices in $\chi_V^{-1}(\cR)$ increases.
  Steps \ref{alg:find-antler-colored:1}, \ref{alg:find-antler-colored:2}, \ref{alg:find-antler-colored:4}, \ref{alg:find-antler-colored:5}, and \ref{alg:find-antler-colored:6}
  can be performed in no more than $\Oh(n^2)$ time.
  For Step \ref{alg:find-antler-colored:3} we solve \PrbOCTLong
  in time $3^z n^{\Oh(1)}$~\cite{ReedSV04} (see~\cite[Thm. 4.17]{CyganFKLMPPS15}) for all $\Oh(n^z)$ subsets $C \subseteq \chi_V^{-1}(\cC)$ of size at most $z$.
  Hence, the overall runtime is $n^{\Oh(z)}$.

  \subparagraph*{Correctness.}
  % Here we present a sketch of the correctness proof. Refer to \cite{DonkersJ24} for more details.
  We first show that the algorithm preserves the properness of recoloring,
  that is, if an arbitrary \zocc is $z$-properly colored prior to the recoloring,
  it is also $z$-properly colored after the recoloring.
  Second, we show that output in Step \ref{alg:find-antler-colored:5} is necessarily a \zocc in $G$.
  Under these assumptions, if $\chi$ $z$-properly colors a \zocc, then the algorithm must output
  a \zocc $(A_B,A_C,A_R)$ in $G$ such that for every $z$-properly colored \occ $(\hat{A}_B, \hat{A}_C, \hat{A}_R)$,
  we have $\hat{A}_B \subseteq A_B$ and $\hat{A}_C \subseteq A_C$ as it does not exclude any 
  $z$-properly colored \zoccs.
  And if $\chi$ does not $z$-properly color any \zoccs, the algorithm should correctly report the absence
  of $z$-properly colored \zoccs in Step \ref{alg:find-antler-colored:6}.

  \begin{claim}\label{clm:find-antler-colored-properness}
      All \zoccs $(\hat{A}_B, \hat{A}_C, \hat{A}_R)$ that are $z$-properly colored by $\chi$
      prior to executing the algorithm are also $z$-properly colored by $\chi$ after termination of the algorithm.
  \end{claim}
  \begin{claimproof}
      Suppose an arbitrary \zocc $(\hat{A}_B, \hat{A}_C, \hat{A}_R)$ is $z$-properly colored by $\chi$.
      Then,
      $\hat{A}_C \subseteq \chi_V^{-1}(\cC)$, 
      $\hat{A}_B \subseteq \chi_V^{-1}(\cB)$, and
      $G'=G[\hat{A}_B \cup \hat{A}_C] - \chi^{-1}(\cR)$ is an $\hat{A}_C$-certificate of order $z$.
      We will show that (1) if a vertex $v$ is recolored, then $v \not\in \hat{A}_B \cup \hat{A}_C$
      and (2) edges in $G[\hat{A}_B \cup \hat{A}_C]$ are never recolored.
      With these conditions, we see that $(\hat{A}_B, \hat{A}_C, \hat{A}_R)$ is $z$-properly colored by $\chi$ at any time during the algorithm.

      To show (1), analyze each step where a vertex is recolored.
      Suppose a vertex $v$ is recolored in Step \ref{alg:find-antler-colored:2}.
      Since $\chi$ $z$-properly colors $(\hat{A}_B, \hat{A}_C, \hat{A}_R)$,
      we have $N(\hat{A}_B) \subseteq \hat{A}_C \subseteq \chi_V^{-1}(\cC)$.
      We know that $\chi(v)=\cB$, so if $v \in \hat{A}_B$,
      then the component $B$ of $G[\chi_V^{-1}(\cB)]$ with $v \in B$ must be entirely included in $\hat{A}_B$.
      This implies that $G[B]$ is bipartite and $N(B) \subseteq N(\hat{A}_B) \subseteq \chi_V^{-1}(\cC)$, contradicting the recoloring condition in Step \ref{alg:find-antler-colored:2}.
      
      Next, suppose a vertex $v$ is recolored in Step \ref{alg:find-antler-colored:4}.
      We know that $\chi(v)=\cC$, and $v$ was not marked during Step \ref{alg:find-antler-colored:3}.
      Assume for the contradiction that $v \in \hat{A}_C$.
      Since $G'$ is an $\hat{A}_C$-certificate of order $z$,
      there exists a component $H$ of $G'$ such that $v \in H$ and $\oct(H)=|\hat{A}_C \cap V(H)| \leq z$.
      From $\hat{A}_C \cap V(H) \subseteq \hat{A}_C \subseteq \chi_V^{-1}(\cC)$,
      we know that in some iteration in Step \ref{alg:find-antler-colored:3}
      we have $C=\hat{A}_C \cap V(H)$.
      Vertex $v$ was marked if $\oct(G_\chi[C \cup W_\chi(C)]) = |C|$, which we want to show for a contradiction.
      First, notice that $W_\chi(C) \subseteq \chi^{-1}(\cB)$, and at this point $\chi^{-1}(\cB)$ is
      bipartite.
      Hence, $\oct(G_\chi[C \cup W_\chi(C)]) \leq |C|$.
      Next, we show that $H$ is a subgraph of $G_\chi[C \cup W_\chi(C)]$.
      It is clear to see that a component $H'$ in $H-C$ is also a component in $G_\chi[\chi_V^{-1}(\cB)]$.
      Since $H$ is a component of $G'$, $N_{G'}(H') \subseteq \hat{A}_C \cap V(H) = C$,
      and hence $N_{G'}(V(H-C)) \subseteq C$.
      Also, because $H$ is connected in $G'$, we have $V(H-C) \subseteq W_\chi(C)$,
      implying that $H$ is a subgraph of $G_\chi[C \cup W_\chi(C)]$.
      We have $\oct(G_\chi[C \cup W_\chi(C)]) \geq \oct(H) = |C|$.\looseness-1

      For (2), edge recoloring takes place only when the edge is incident to a vertex colored $\cR$.
      Every edge $e$ in $G[\hat{A}_B \cup \hat{A}_C]$ is not incident to
      any vertex in $\chi_V^{-1}(\cR)$, so $e$ cannot be recolored.
  \end{claimproof}

  \begin{claim}\label{clm:find-antler-colored-output}
      In Step \ref{alg:find-antler-colored:5}, $(\chi_V^{-1}(\cB), \chi_V^{-1}(\cC), \chi_V^{-1}(\cR))$ is a \zocc in $G$.
  \end{claim}
  \begin{claimproof}
      In this proof, for any family of sets $X_1,\ldots,X_\ell$ indexed by $\{1,\ldots,\ell\}$
      we define the following for all $1 \leq i \leq \ell$:
      $X_{<i} := \bigcup_{1 \leq j < i}X_j$ and $X_{\leq i} := \bigcup{1 \leq j \leq i}X_j$.

      From Step \ref{alg:find-antler-colored:2}, we know that $G[\chi_V^{-1}(\cB)]$ is bipartite
      and $\chi_V^{-1}(\cB) \subseteq \chi_V^{-1}(\cC)$.
      Also, Step \ref{alg:find-antler-colored:5} guarantees that $\chi_V^{-1}(\cB) \cup \chi_V^{-1}(\cC) \neq \emptyset$, so $(\chi_V^{-1}(\cB), \chi_V^{-1}(\cC), \chi_V^{-1}(\cR))$ is an \occ in $G$.
      It remains to show that $G[\chi_V^{-1}(\cB) \cup \chi_V^{-1}(\cC)]$ contains a $\chi_V^{-1}(\cC)$-certificate of order $z$.

      Let $\mathcal{C} \subseteq 2^{\chi_V^{-1}(\cC)}$ be the family of all subsets $C \subseteq \chi_V^{-1}(\cC)$ that have been considered and marked in Step \ref{alg:find-antler-colored:3}, i.e. $\oct(G_\chi[C \cup W_\chi(C)]) = |C| \leq z$ for every $C \in \mathcal{C}$.
      Let $C_1,\ldots,C_{|\mathcal{C}|}$ be the sets in $\mathcal{C}$ in an arbitrary order
      and define $D_i := C_i \setminus C_{<i}$ for all $1 \leq i \leq |\mathcal{C}|$.
      Also define vertex-disjoint subgraphs $G_i := G_\chi[D_i \cup (W_\chi(D_{\leq i})) \setminus W_\chi(D_{<i})]$ for all $1 \leq i \leq |\mathcal{C}|$.
      For any $i<j$, we have $D_i \cap D_j = \emptyset$ by definition, and 
      $(W_\chi(D_{\leq i}) \setminus W_\chi(D_{<i})) \cap (W_\chi(D_{\leq j}) \setminus W_\chi(D_{<j})) = \emptyset$ because $W_\chi(D_{\leq i}) \subseteq W_\chi(D_{<j})$.
      Hence, $V(G_i) \cap V(G_j) = \emptyset$.
      % The proof of that they are indeed vertex disjoint is in \cite[Lemma 6.2]{DonkersJ24}.

      From $\oct(G_\chi[C_i \cup W_\chi(C_i)]) = |C_i|$, we have
      $\oct(G_\chi[C_i \cup W_\chi(C_i)] - (C_i \setminus D_i)) = \oct(G_\chi[D_i \cup W_\chi(C_i)]) = |D_i|$.
      Furthermore, in $G_\chi[D_i \cup W_\chi(C_i)]$,
      the vertices $W_\chi(C_i) \cap W_\chi(C_{<i})$ induce a bipartite graph
      and are disconnected from $D_i \cup W_\chi(C_i) \setminus W_\chi(C_{<i})$;
      by construction, there are no edges between $D_i$ and $W_\chi(C_{<i})$
      because a vertex $v \in C$ adjacent to $W_\chi(C_{<i})$ must belong to $C_{<i}$,
      and there are no edges between $W_\chi(C_i)$ and $W_\chi(C_{<i})$ as they form different components in $G\chi[\chi_V^{-1}(\cB)]$.
      Hence, we have $\oct(G_\chi[D_i \cup W_\chi(C_i)]) = \oct(G_\chi[D_i \cup W_\chi(C_i)] - (W_\chi(C_{<i}))) = \oct(G_i) = |D_i|\leq |C_i|\leq z$.
      The disjoint union of $G_i$ contains a $(\bigcup_i D_i= \chi_V^{-1}(\cC)$)-certificate of order $z$.
  \end{claimproof}
	
This concludes the proof of \cref{lem:find-antler-colored}.
\end{proof}

\thmMain*

\begin{proof}
  Consider the following algorithm.

  \begin{enumerate}
    \item Use the algorithm from \Cref{lem:finding-occ} to either obtain a reducible \occ $(X_B,X_C,X_R)$ of width at most $k$
    or determine if there is no single-component \occ of width at most $k$ with $|X_B| > g_r(2k)$
    in $2^{\Oh(k^{16})}n^{\Oh(1)}$ time.
    \label{alg:main-thm:1}

    \item If we obtain a reducible \occ in Step \ref{alg:main-thm:1}, then apply reductions in $2^{\Oh(k)} n^{\Oh(1)}$ time
    as described in \Cref{lem:covering:corollary}, and continue to Step \ref{alg:main-thm:1}.
    \label{alg:main-thm:2}
    
    \item Let $G'$ be the reduced graph.
    Create a $(|V(G')|+|E(G')|,g(k,z),3)$-universal function family $\mathcal{F}$ for $V(G') \cup E(G') \to \{\cB, \cC, \cR\}$
    using \Cref{cor:universal-partition}, where we set $g(k,z) = 2kz^2 (g_r(2k))^2$.
    Each function $\chi \in \mathcal{F}$ represents a coloring of the vertices and edges in $G'$.\looseness-1
    %
    % Then, enumerate all functions $F \in \mathcal{F}$ such that $F: [n+m] \to [3]$,
    % with which we can construct a coloring $\chi\colon V(G) \cup E(G) \to \{\cB, \cC, \cR\}$
    % since there is a bijection between $\mathcal{F}$ and the family of $\chi$.
    \label{alg:main-thm:3}

    \item Iterate over all colorings $\chi \in \mathcal{F}$.
    For each $\chi$, call the algorithm from \Cref{lem:find-antler-colored} as a subroutine.
    If it outputs a \zocc $(A_B, A_C, A_R)$ with $|A_C| \geq k$, then output $A_C$ as the final result.
    \label{alg:main-thm:4}

    \item If all colorings result in a \zocc $(A_B, A_C, A_R)$ such that $|A_C| < k$,
    or all colorings are concluded as having no \zoccs, then report that $G$ does not contain a \zocc
    of width $k$.
    \label{alg:main-thm:5}
  \end{enumerate}

  \subparagraph*{Running time.}
  The reduction steps (Steps \ref{alg:main-thm:1}-\ref{alg:main-thm:2}) take $2^{\Oh(k^{16})} n^{\Oh(1)}$ time
  because whenever we find a reducible \occ, the number of vertices in $G$ decreases, resulting in at most $n$ iterations.
  For Step \ref{alg:main-thm:3}, from \Cref{cor:universal-partition} we can construct a family of colorings $\chi$
  that cover all colorings when restricted to $g(k,z)$ elements.
  Since the number of vertices never increases by the reduction steps, it holds that $|V(G')| \leq n$.
  Steps \ref{alg:main-thm:4}-\ref{alg:main-thm:5} take $n^{\Oh(z)}$ time for each coloring,
  and there are $2^{\Oh(g(k,z))} \log^2 n$ colorings.
  The overall runtime is $2^{\Oh(g(k,z))} n^{\Oh(z)} = 2^{\Oh(k^{33} z^2)} n^{\Oh(z)}$.

  \subparagraph*{Correctness.}
  The algorithm first reduces all reducible \occs in Steps \ref{alg:main-thm:1}-\ref{alg:main-thm:2}.
  Let $G'$ be the reduced graph when the algorithm enters Step \ref{alg:main-thm:3}.
	
	We start by arguing that the output of the algorithm is correct when it outputs the set~$A_C$ in Step~\ref{alg:main-thm:4} after having found a \zocc~$(A_B, A_C, A_R)$ in~$G'$. Observe that by definition of tight \occ, the set~$A_C$ is a subset of an optimal OCT $S$ for $G'$. Then, from the backward safety (\Cref{lem:backward:safety}) and its  transitivity, we have that 
  $S \subseteq V(G) \cap V(G')$ and $S$ is an optimal OCT of $G$. So the final result $A_C$ in Step \ref{alg:main-thm:4} belongs to an optimal solution for \PrbOCTLong in $G$.
	
	Suppose now that~$G$ contains a \zocc of width~$k$; we argue that the algorithm will produce a suitable output in Step~\ref{alg:main-thm:4}. From this, it will follow that the algorithm is correct in reporting that~$G$ does not have a \zocc of width~$k$ if it terminates in Step~\ref{alg:main-thm:5}. By invoking forward safety (\Cref{lem:forward:safety}) and its transitivity,
  we know that if there exists a \zocc $(A_B, A_C, A_R)$ of width $k$ in $G$, 
  then there exists a \zocc $(A'_B, A'_C, A'_R)$ of width $k$ in $G'$. By applying \Cref{lem:num_components}, we may assume without loss of generality that 
  $G'[A_B']$ has at most $z^2|A'_C|$ components.
	
	For each component $B'$ in $G'[A_B']$, the structure $(V(B'),A'_C,A'_R \cup (A'_B \setminus V(B')))$ is a single-component \occ in~$G'$ of width $k$. Since the reduction process stabilized with~$G'$, we have $|V(B')| \leq g_r(2k)$. Hence each of the at most~$z^2|A'_C|$ components of~$G'[A'_B]$ has at most~$g_r(2k)$ vertices, so that~$|A_B'| \leq k z^2 \cdot g_r(2k)$.

  % In the remainder of the proof, we argue that the algorithm is also correct when it terminates in Step~\ref{alg:main-thm:5}. To do so, we argue that if~$G$ had a \zocc of width~$k$, then there is a suitable coloring that causes the algorithm to terminate in Step \ref{alg:main-thm:4}. From this, it follows that the algorithm is correct in reporting that~$G$ does not have a \zocc of width~$k$ if it terminates in Step~\ref{alg:main-thm:5}.
	
  %Now, suppose that the all colorings\bmpr{Typo?} such that $g(k,z)$ elements in
  %$V(G') \cup E(G')$ are colored correctly in Step \ref{alg:main-thm:4},
  %we can find a \zocc of width at least $k$ in $G'$
  %or correctly conclude that $G'$ does not contain a \zocc of width $k$.
  %In the latter case, from the forward safety we can conclude that $G$ does not contain a \zocc of width $k$, which should be reported in Step \ref{alg:main-thm:5}.
  %Otherwise, the algorithm will output a \zocc $(A_B,A_C,A_R)$ with $|A_C| \geq k$ in $G'$.
  %First, observe that by definition, $A_C$ is a subset of an optimal OCT $S$ for $G'$.
  %Then, from the backward safety (\Cref{lem:backward:safety}), we have that 
  %$S \subseteq V(G) \cap V(G')$ and $S$ is an optimal OCT of $G$.
  %Hence, the final result $A_C$ in Step \ref{alg:main-thm:4} belongs to an optimal solution for \PrbOCTLong in $G$.

  Next, we show that there is a coloring of~$V(G') \cup E(G')$ that properly colors~$(A'_B, A'_C, A'_R)$. To see this, consider an~$A'_C$-certificate~$D'$ of order~$z$ in~$G'[A'_C \cup A'_B]$, which exists by \cref{def:zocc}; hence~$D'$ is a subgraph of~$G'[A'_C \cup A'_B]$ for which~$A'_C$ is an optimal OCT, while each component of~$D'$ contains at most~$z$ vertices from~$A'_C$. Then any coloring~$\chi$ that assigns~$A'_C$ color~$\cC$, assigns~$A'_B$ color~$\cB$, assigns all edges that belong to~$D'$ color~$\cB$ and all remaining edges of~$G'[A'_B \cup A'_C]$ color~$\cR$, will $z$-properly color~$(A'_B, A'_C, A'_R)$. To obtain a $z$-proper coloring, it therefore suffices for the coloring to act as prescribed on the vertex set~$A'_B \cup A'_C$ and on the edges of the subgraph~$G'[A'_B \cup A'_C]$. 

 	%To ensure that the algorithm constructs a coloring~$\chi$ that $z$-properly colors~$(A'_B, A'_C, A'_R)$ in Step~\ref{alg:main-thm:3}, it therefore suffices to argue that whether or not a coloring is proper is determined only by the colors of $g(k,z)$ entities.
  %
  %For each component $B$ in $G'[A_B']$, $(B,A_C,A_R \cup A_B \setminus B)$ is a single-component \occ of width $k$.
  %From the reduction steps, we have $|B| \leq g_r(2k)$ and consequently
  %$|A_B'| \leq z^2|A_C|\cdot |B| \leq kz^2 g_r(2k)$.

  %Now, consider a coloring of $(A_B', A_C, A_R \cup A_B \setminus A_B')$.
  %We obtain a $z$-proper coloring of this \zocc whenever
  %the set $A_B' \cup A_C \cup E(G[A_B' \cup A_C])$ is $z$-properly colored.
  %
  Notice that a bipartite graph of $n$ vertices may have at most $n^2/4$ edges.
  Hence:
	\begin{align*}
	|E(G'[A'_B \cup A'_C])| &= |E(G'[A'_B])| + |E_{G'}(A'_B, A'_C)| + |E(G'[A'_C])| \\
	  &\leq k z^2 (g_r(2k))^2/4 + k z^2 g_r(2k) \cdot k + \binom{k}{2} \\
		& \leq kz^2 (g_r(2k))^2,
	\end{align*}
	where we use the fact that~$|A'_C| \leq k$ and use~$E_{G'}(A'_B, A'_C)$ to denote the edges of~$G'$ that have one endpoint in~$A'_B$ and the other in~$A'_C$.
  This allows us to bound the number of elements whose color determines whether or not a coloring is $z$-proper by:
	\begin{align*}
	|A'_B \cup A'_C \cup E(G'[A'_B \cup A'_C])| &\leq k z^2 \cdot g_r(2k) + k + kz^2 (g_r(2k))^2 \\ &\leq 2kz^2 (g_r(2k))^2 = g(k,z).
	\end{align*}
	Since the algorithm constructs a universal function family with parameters~$(|V(G')| + |E(G')|, g(k,z), 3)$, it follows that the algorithm encounters at least one $z$-proper coloring of~$(A'_B, A'_C, A'_R)$ if it exists. Hence the algorithm is correct if it concludes there is no \zocc in~$G$ of width~$k$.
\end{proof}

\section{Hardness results}
\label{sec:hardness}

First, we consider the NP-hardness of the following problem.

\begin{problem}{\probname{$1$-Tight OCC Detection}}
  \Input & A graph $G$.\\
  \Prob & Output a non-empty $1$-tight OCC $(X_B,X_C,X_R)$ in $G$ or conclude that no non-empty $1$-tight OCC exists.
\end{problem}

To prove that this problem is NP-hard under Turing reductions,
we consider the following version of \PrbOCT.

\begin{problem}{\probname{Lower-Bounded OCT}}
  \Input & A graph $G$ and $k$ vertex-disjoint odd cycles $\{C_i\}$ in $G$.\\
  \Prob & Is there a set $S \subseteq V(G)$ of size at most $k$ such that $G-S$ is bipartite?
\end{problem}

We claim that \probname{$1$-Tight OCC Detection} is as hard as \probname{Lower-Bounded OCT}.

% \bmp{The exposition in this section overlooks the distinction between NP-hardness under Turing reductions versus under Karp reductions. Typically, when one sees '$X$ is NP-hard', it means NP-hard under Karp reductions: for each problem~$Y$ in NP, there is a polynomial-time reduction that maps each instance of~$Y$ to a single instance of~$X$ with the same YES/NO answer. That is not the type of NP-hardness we get here, since we have to make \emph{multiple} calls to an algorithm for 1-tight OCC detection to solve lower-bounded OCT. Our proof effectively shows hardness under Turing reductions, i.e., if there was a poly-time algorithm for 1-tight OCC detection, then there would be a polynomial-time algorithm for all problems in NP. I suggest staying closer to the formulation from the journal version of the previous paper, or explicitly mentioning 'NP-hard under Turing reductions', to make this clear.}

\begin{lemma}\label{lem:1-tight-occ-reduction}
  Assuming $\textnormal{\textsf{P}} \neq \textnormal{\textsf{NP}}$,
  if there is no polynomial-time algorithm for \probname{Lower-Bounded OCT}, then
  there is no polynomial-time algorithm for \probname{$1$-Tight OCC Detection}.
\end{lemma}

\begin{proof}
  Assume \probname{$1$-Tight OCC Detection} is solvable in polynomial time.
  Then, consider the following algorithm for \probname{Lower-Bounded OCT}.

  \begin{enumerate}
    \item Let $C \gets \{C_i\}$. Initially, $|C|=k$.
    \item Return \textsf{Yes} if $C=\emptyset$.
    \item Run an algorithm for \probname{$1$-Tight OCC Detection} on $G$.
    If it finds a non-empty $1$-tight OCC $(X_B,X_C,X_R)$, then
    let $G \gets X_R$ and $C \gets \{C_i \in C \mid X_C \cap V(C_i) \neq \emptyset\}$,
    i.e. remove odd cycles in $C$ that hit any vertex in $X_C$,
    and then continue to Step 2.
    Otherwise, return \textsf{No}.
  \end{enumerate}

  Let $S$ be the union of $X_C$ found in Step 3.
  If the algorithm returns \textsf{Yes}, then
  it is clear to see that $|S|=k$ and $G-S$ is bipartite.
  $(G,\{C_i\})$ is a yes-instance for \probname{Lower-Bounded OCT}.
  If the algorithm returns \textsf{No}, then
  there exists an induced subgraph $G'$ of $G$ such that
  for any OCT $S$ of $G'$, there exists a cycle $C_i$ such that $|V(C_i) \cap S|>1$.
  Hence, $\oct(G)>k$, and $(G,\{C_i\})$ is a no-instance.
  If \probname{$1$-Tight OCC Detection} is poly-time solvable,
  then so is \probname{Lower-Bounded OCT}.
\end{proof}

Now, let us show that \probname{Lower-Bounded OCT} is actually NP-hard.
We adopt an argument from \cite[Lemma 16]{DonkersJ19} but with simpler gadgets.

\begin{lemma}\label{lem:lb-oct-reduction}
  \probname{Lower-Bounded OCT} is NP-hard.
\end{lemma}

\begin{proof}
  We reduce from \probname{3-SAT}, a version of \probname{Satisfiability (SAT)}
  where every clause includes exactly $3$ literals, to \probname{Lower-Bounded OCT}.
  Let $X=x_1,\ldots,x_n$ be the variables and $\Phi=\phi_1,\ldots,\phi_m$ be the clauses
  appearing in a SAT instance.

  We construct a graph $G$ as follows:
  \begin{itemize}
    \item For each variable $x_i$, create a new triangle and let $v_i$ and $\overline{v_i}$
    be two of the triangle's vertices.
    These vertices represent SAT literals $x_i$ and $\overline{x_i}$ as shown in \Cref{fig:proof-np-hardness} (left).
    \item For each clause $\phi_j=\ell_{j1} \lor \ell_{j2} \lor \ell_{j3}$, introduce a gadget illustrated in \Cref{fig:proof-np-hardness} (right).
    We then identify vertices $s_1, s_2, s_3$ and the vertices representing $\ell_{j1}$, $\ell_{j2}$, $\ell_{j3}$,
    respectively.
  \end{itemize}

  It is clear to see that this construction can be done in polynomial time.
  Also observe that there are $n$ disjoint triangles for SAT variables
  and $2m$ disjoint triangles for clause gadgets.
  By letting $k=n+2m$, we have that $G$ contains $k$ vertex-disjoint triangles, which are odd cycles.
  We claim that $(X,\Phi)$ is satisfiable if and only if $G$ contains an OCT of size at most $k$.

  First, assume $(X,\Phi)$ is satisfiable.
  Let $U \subseteq V(G)$ be the vertices that correspond to a certificate,
  and consider all odd cycles present in $G-U$.
  For each variable gadget, we know that either $x_i$ or $\overline{x_i}$ is in $U$,
  so no triangles remain in $G-U$.
  For each clause gadget, we know that at least one of $s_1, s_2, s_3$ is in $U$.
  In any cases, we can choose $2$ vertices hitting all remaining triangles.
  Hence, $\oct(G-U) \leq 2m$ and $\oct(G) \leq n + 2m = k$.

  Conversely, assume $G$ has an OCT $S \subseteq V(G)$ of size at most $k$.
  Since $G$ has $k$ vertex-disjoint triangles, $S$ contains exactly one vertex from each of them.
  For every variable gadget, exactly one vertex is in $S$.
  For every clause gadget excluding $s_1$, $s_2$, $s_3$, exactly $2$ vertices must be in $S$,
  which leads to the fact that at least one of $s_1$, $s_2$, $s_3$ is in $S$.
  Finally, we construct a solution for \probname{3-SAT}.
  For each $1 \leq i \leq n$,
  we set $x_i$ to \textsf{true} if $v_i \in S$,
  and \textsf{false} otherwise.
  This by construction forms a certificate, and $(X,\Phi)$ is satisfiable.
\end{proof}

Combining \Cref{lem:1-tight-occ-reduction} and \Cref{lem:lb-oct-reduction},
we obtain the following result.

\begin{theorem} \label{thm:1tight:occ:nphard}
  \probname{$1$-Tight OCC Detection} is NP-hard under Turing reductions.
\end{theorem}

Next, we show that finding a non-empty 1-tight OCC of bounded width is W[1]-hard,
parameterized by the width.

\begin{problem}{\probname{Bounded-Width $1$-Tight OCC Detection}}
  \Input & A graph $G$ and an integer $k$.\\
  Parameter: & $k$. \\
  \Prob & Does $G$ admit a non-empty $1$-tight OCC $(X_B,X_C,X_R)$ with $|X_C|\leq k$?
\end{problem}

We will show the W[1]-hardness of this problem by a parameterized reduction from
the well-known \probname{Multicolored Clique} problem.

\begin{problem}{\probname{Multicolored Clique}}
  \Input & A graph $G$, an integer $k$, and a partition of $V(G)$ into sets $V_1,\ldots,V_k$.\\
  Parameter: & $k$. \\
  \Prob & Does $G$ include as a subgraph such a clique $S$ such that
  for each $1 \leq i \leq k$ we have $|S \cap V_i|=1$?
\end{problem}

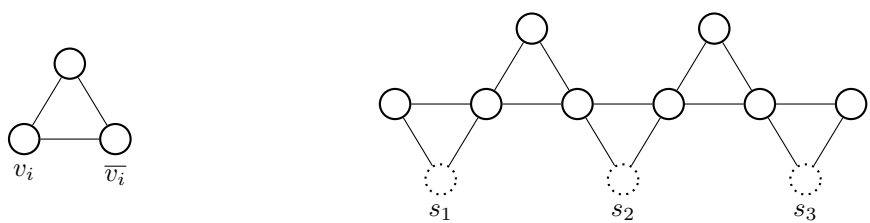
\begin{figure*}[t]
  \centering
  \tikzstyle{large} = [circle, fill=white, text=black, draw, thick, scale=1, minimum size=0.8cm, inner sep=1.5pt]
  \tikzstyle{plain} = [circle, fill=white, text=black, draw, thick, scale=1, minimum size=0.4cm, inner sep=1.5pt]
  \tikzstyle{small} = [circle, fill=white, text=black, draw, thick, scale=1, minimum size=0.2cm, inner sep=1.5pt]
  \tikzstyle{black} = [circle, fill=gray, text=black, draw, thick, scale=1, minimum size=0.2cm, inner sep=1.5pt]
  \tikzmath{\xunit = 0.6; \yunit = 1.0; }

  \begin{minipage}[m]{.30\textwidth}
    \vspace{0pt}
    \centering
    \begin{tikzpicture}
        % vertices
        \node[plain,label=below:$v_i$] (a2) at (2 * \xunit, 1 * \yunit) {};
        \node[plain,label=below:$\overline{v_i}$] (a3) at (4 * \xunit, 1 * \yunit) {};
        \node[plain] (b1) at (3 * \xunit, 2 * \yunit) {};
        
        \draw (b1)--(a2)--(a3)--(b1);
    \end{tikzpicture}
  \end{minipage}
  \hfill
  \begin{minipage}[m]{.66\textwidth}
      \vspace{0pt}
      \centering
      \begin{tikzpicture}
          % vertices
          \node[plain,dotted,label=below:$s_{1}$] (s1) at (1 * \xunit, 0) {};
          \node[plain,dotted,label=below:$s_{2}$] (s2) at (5 * \xunit, 0) {};
          \node[plain,dotted,label=below:$s_{3}$] (s3) at (9 * \xunit, 0) {};
          \node[plain] (a1) at (0 * \xunit, 1 * \yunit) {};
          \node[plain] (a2) at (2 * \xunit, 1 * \yunit) {};
          \node[plain] (a3) at (4 * \xunit, 1 * \yunit) {};
          \node[plain] (a4) at (6 * \xunit, 1 * \yunit) {};
          \node[plain] (a5) at (8 * \xunit, 1 * \yunit) {};
          \node[plain] (a6) at (10 * \xunit, 1 * \yunit) {};
          \node[plain] (b1) at (3 * \xunit, 2 * \yunit) {};
          \node[plain] (b2) at (7 * \xunit, 2 * \yunit) {};
          
          \draw (a1)--(a2)--(a3)--(a4)--(a5)--(a6);
          \draw (a1)--(s1)--(a2);
          \draw (a3)--(s2)--(a4);
          \draw (a5)--(s3)--(a6);
          \draw (a2)--(b1)--(a3);
          \draw (a4)--(b2)--(a5);
          
          % \draw (u23)--(u32);
          % \draw (u31)--(u32);
      \end{tikzpicture}
  \end{minipage}
  \caption{%
  Gadgets used for the proof of \Cref{lem:lb-oct-reduction}.
  The left shows a \emph{variable gadget}, including two vertices representing literals.
  The right shows a \emph{clause gadget}, where $s_1$, $s_2$, $s_3$ are connected to
  corresponding literals.
  }
  \label{fig:proof-np-hardness}
\end{figure*}

\begin{theorem} \label{thm:tightocc:w1hard}
  \probname{Bounded-Width $1$-Tight OCC Detection} is W[1]-hard.
\end{theorem}

\begin{proof}
  Consider an instance $(G,\chi,k)$ of \probname{Multicolored Clique},
  where $\chi\colon V(G) \to [k]$ denotes the \textit{color} of each vertex such that
  $\chi(v)=\ell$ if and only if $v \in V_\ell$.
  Without loss of generality, we assume $k \geq 2$ and $n > k + 2$.
  Then, we construct an input $(G',k')$ for \probname{Bounded-Width $1$-Tight OCC Detection} as follows
  (see \Cref{fig:proof-w1-hardness} for an illustration).
  \begin{enumerate}
    \item For each vertex $i \in V(G)$, introduce ($k-1$) vertices $U_{i}:=\{u_{i\ell} \mid \ell \in [k] \setminus \{\chi(i)\}\}$ to $G'$.
    \item For each edge $ij \in E(G)$, add $4$ vertices $x_{ij}$, $x_{ij}'$, $y_{ij}$, and $y_{ij}'$ to $G'$.
    \item Insert edges to turn $\bigcup_{i \in V(G)}U_i$ into a clique in $G'$.
    \item For each edge $ij \in E(G)$, add edges from $U_i \cup U_j$ to $x_{ij}$,
    $x_{ij}$ to $x_{ij}'$, and $x_{ij}'$ to $u_{i \chi(j)}$.
    Similarly, add edges from $U_i \cup U_j$ to $y_{ij}$,
    $y_{ij}$ to $y_{ij}'$, and $y_{ij}'$ to $u_{j \chi(i)}$.
  \end{enumerate}

  Let $k'=k(k-1)$, and we claim that $(G',k')$ for \probname{Bounded-Width $1$-Tight OCC Detection}
  \probname{Multicolored Clique}  is an equivalent instance to $(G,\chi,k)$ for \probname{Multicolored Clique}.

  \begin{claim}\label{clm:w1-hardness-forward}
    If $G$ has a multicolored clique of size $k$, then $G'$ has a non-empty $1$-tight OCC of width $k'$.
  \end{claim}

  \begin{claimproof}
    Suppose $S$ is a multicolored clique of size $k$ in $G$.
    Define a $1$-tight OCC $(X_B,X_C,X_R)$ as follows:
    \begin{itemize}
      \item $X_B = \bigcup_{i,j \in S:i \neq j} \{ x_{ij}, x_{ij}', y_{ij}, y_{ij}' \}$
      \item $X_C = \bigcup_{i \in S} U_i$
      \item $X_R = V(G') \setminus (X_B \cup X_C)$
    \end{itemize}

    It is clear to see that $|X_C|=k(k-1)=k'$.
    Also, by construction, $N(X_C)=\bigcup_{i,j \in S \mid i \neq j} (U_i \cup U_j \cup \{ u_{i \chi (j)}, u_{j \chi (i)} \}) = \bigcup_{i \in S} U_i = X_C$,
    which remains us to show that $G'[X_B \cup X_C]$ contains $k'$ vertex-disjoint odd cycles.
    Notice that $G'[X_B \cup X_C]$ contains the set of triangles
    $T = \bigcup_{i,j \in S \mid i \neq j} \{u_{i \chi(j)}x_{ij}x_{ij}', u_{j \chi(i)}y_{ij}y_{ij}'\}$.
    Since $|T|=\binom{k}{2}\cdot 2=k'$ and $x_{ij},x_{ij}',y_{ij},y_{ij}'$ are unique in $T$,
    it suffices to show that $\left|\bigcup_{i,j\in S \mid i\neq j} \{u_{i\chi(j)}, u_{j\chi(i)}\}\right|=k'$.
    Since $S$ is a multicolored clique, for every pairwise distinct $i,j\in S$, we have $\chi(i) \neq \chi(j)$.
    Hence, $\bigcup_{i,j\in S \mid i\neq j} \{u_{i\chi(j)}, u_{j\chi(i)}\}\supseteq
    \bigcup_{i\in S} \left(\bigcup_{j \in S \setminus \{i\}} u_{i \chi(j)}\right)
    =\bigcup_{i\in S} \bigcup_{\ell \in [k] \setminus \{\chi(i)\}}u_{i\ell}
    =\bigcup_{i\in S} U_i$.
    $T$ indeed consists of $\left| \bigcup_{i\in S} U_i \right|=k(k-1)=k'$ vertex-disjoint triangles,
    and $(X_B,X_C,X_R)$ is a $1$-tight OCC.
  \end{claimproof}

  Finally, we show the backward implication.

  \begin{claim}\label{clm:w1-hardness-backward}
    If $G'$ has a non-empty $1$-tight OCC of size $k'$, then $G$ has a multicolored clique of size $k$.
  \end{claim}

  \begin{claimproof}
    Let $(X_B,X_C,X_R)$ be a non-empty $1$-tight OCC of size $k'$ in $G'$.
    We write $U$ for $\bigcup_{i\in S}U_i$ and $W$ for $\bigcup_{e \in E(G)} \{x_e, x_e', y_e, y_e'\}$.

    First, we show that $X_B \subseteq W$.
    Assume not.
    Then, there exists a vertex $u \in X_B \cap U$.
    Since $X_B$ is bipartite and $U$ forms a clique in $G'$,
    we have $|X_B \cap U| \leq 2$.
    We know that $N_{G'}(u) \subseteq X_B \cup X_C$,
    and thus $|N_{G'}(u) \setminus X_B| \leq |X_C|$.
    We have:
    $|N_{G'}(u) \setminus X_B|\geq |N_{G'}(u)|-1
    \geq (|U|-1)-1 = n(k-1)-2 > ((k+2)(k-1)-2) = k^2 + k - 4 \geq k^2-k=k'$.
    This implies $|X_C|>k'$, a contradiction.

    Next, we show that $X_C \subseteq U$.
    For the sake of contradiction, assume there exists a vertex $x \in X_C \cap W$.
    Then, there must be an odd cycle in $G'[\{x\} \cup X_B]$, but because
    $\{x\} \cup X_B \subseteq W$ is bipartite, we reach a contradiction.

    Now that we know $X_B \subseteq W$ and $X_C \subseteq U$,
    for each $u \in X_B$, there is an odd cycle in $G'[\{u\} \cup X_B]$.
    This happens only when $u x_e x_e'$ or $u y_e y_e'$ forms a triangle for some $e \in E(G)$,
    which leads to $|X_B| \geq 2|X_C|$.

    Let $E' = \{e \in E(G) \mid \{x_e, y_e\} \cap X_B \neq \emptyset\}$.
    If $ij \in E'$, then all the vertices in $C_i \cup C_j$ are in $X_C$,
    so we can write $X_C=\bigcup_{i \in V'}C_i$ for some $V' \subseteq V(G)$.
    Furthermore, from $|X_C|=k'=k(k-1)$, we know $|V'|=k$.
    From $E' \subseteq \binom{V'}{2}$, we have $|X_B| \leq 4|E'| \leq 2|V'|(|V'|-1)=2k'$.
    On the other hand, $|X_B| \geq 2|X_C| = 2k'$.
    These inequalities hold only when $E'=\binom{V'}{2}$.

    We claim that $V'$ is a multicolored clique in $G$.
    Let $ij$ be an edge in $E'$.
    By construction, triangles $u_{i \chi(j)}x_{ij}x_{ij}'$ and $u_{j \chi(i)}y_{ij}y_{ij}'$
    are in $G'[X_B \cup X_C]$, and thus $\chi(i)\neq \chi(j)$.
    Since $V'$ forms a clique in $G$ such that any vertex pair has different colors,
    $V'$ is a multicolored clique.
  \end{claimproof}

  From Claims \ref{clm:w1-hardness-forward} and \ref{clm:w1-hardness-backward},
  we conclude that $(G,\chi,k)$ and $(G',k')$ are equivalent,
  and \probname{Bounded-Width $1$-Tight OCC Detection} is W[1]-hard.
\end{proof}

\begin{figure*}[t]
  \centering
  \tikzstyle{large} = [circle, fill=white, text=black, draw, thick, scale=1, minimum size=0.8cm, inner sep=1.5pt]
  \tikzstyle{plain} = [circle, fill=white, text=black, draw, thick, scale=1, minimum size=0.5cm, inner sep=1.5pt]
  \tikzstyle{small} = [circle, fill=white, text=black, draw, thick, scale=1, minimum size=0.2cm, inner sep=1.5pt]
  \tikzstyle{black} = [circle, fill=gray, text=black, draw, thick, scale=1, minimum size=0.2cm, inner sep=1.5pt]
  \tikzmath{\xunit = 0.8; \yunit = 0.8; \dx = 1.0; \dy = 0.4; \cx=0.8; }

  \begin{minipage}[m]{.96\textwidth}
      \vspace{0pt}
      \centering
      \begin{tikzpicture}
          % vertices
          \node[small,label=below:$u_{12}$] (u12) at (0 + 1 * \dx, 1 - 1 * \dy) {};
          \node[small,label=below:$u_{13}$] (u13) at (0 + 2 * \dx, 1 - 2 * \dy) {};

          \node[small,label={[shift={(-0.5, -0.1)}]$u_{21}$}] (u21) at (4, 1) {};
          \node[small,label=below:$u_{23}$] (u23) at (4 + 2 * \dx, 1 - 2 * \dy) {};

          \node[small,label={[shift={(-0.5, -0.1)}]$u_{31}$}] (u31) at (8, 1) {};
          \node[small,label=below:$u_{32}$] (u32) at (8 + 1 * \dx, 1 - 1 * \dy) {};

          \node[small,label=above:$x_{12}$] (x12) at (2.6, 3) {};
          \node[small,label=above:$x_{12}'$] (x12p) at (2.6 - \cx, 3.2) {};
          \node[small,label=above:$y_{12}$] (y12) at (3.4, 3) {};
          \node[small,label=above:$y_{12}'$] (y12p) at (3.4 + \cx, 3.2) {};

          \node[small,label=above:$x_{23}$] (x23) at (6.6, 3) {};
          \node[small,label=above:$x_{23}'$] (x23p) at (6.6 - \cx, 3.2) {};
          \node[small,label=above:$y_{23}$] (y23) at (7.4, 3) {};
          \node[small,label=above:$y_{23}'$] (y23p) at (7.4 + \cx, 3.2) {};

          %edges
          \draw (u12)--(x12);
          \draw (u13)--(x12);
          \draw (u21)--(x12);
          \draw (u23)--(x12);
          \draw (u12)--(y12);
          \draw (u13)--(y12);
          \draw (u21)--(y12);
          \draw (u23)--(y12);

          \draw (x12)--(x12p);
          \draw (y12)--(y12p);
          \draw (x23)--(x23p);
          \draw (y23)--(y23p);

          \draw (x12p)--(u12);
          \draw (y12p)--(u21);

          \draw (u21)--(x23);
          \draw (u23)--(x23);
          \draw (u31)--(x23);
          \draw (u32)--(x23);
          \draw (u21)--(y23);
          \draw (u23)--(y23);
          \draw (u31)--(y23);
          \draw (u32)--(y23);

          \draw (x23p)--(u23);
          \draw (y23p)--(u32);

          \draw (u12)--(u13);
          \draw (u12)--(u21);
          \draw (u12)--(u23);
          \draw (u12)--(u31);
          \draw (u12)--(u32);
          \draw (u13)--(u21);
          \draw (u13)--(u23);
          \draw (u13)--(u31);
          \draw (u13)--(u32);
          \draw (u21)--(u23);
          \draw (u21)--(u31);
          \draw (u21)--(u32);
          \draw (u23)--(u31);
          \draw (u23)--(u32);
          \draw (u31)--(u32);

          \node[] at (9, 3.1) {\large $\cdots$};
          \node[] at (11, 0.4) {\large $\cdots$};

          \draw[dashed,rounded corners] (0, 1.5) rectangle (-0.5 + 3 * \dx, 1 - 3.6 * \dy);
          \draw[dashed,rounded corners] (3.1, 1.5) rectangle (3.5 + 3 * \dx, 1 - 3.6 * \dy);
          \draw[dashed,rounded corners] (7.1, 1.5) rectangle (7.5 + 3 * \dx, 1 - 3.6 * \dy);
          \node[] at (0.4, -0.2) {$U_1$};
          \node[] at (3.5, -0.2) {$U_2$};
          \node[] at (7.5, -0.2) {$U_3$};

        \draw [decorate,decoration={brace,amplitude=5pt}]
        (11.5, 1.5) -- (11.5, -0.5) node[midway,xshift=20pt]{$U$};

        \draw [decorate,decoration={brace,amplitude=5pt}]
        (11.5, 3.5) -- (11.5, 2.5) node[midway,xshift=20pt]{$W$};
      \end{tikzpicture}
  \end{minipage}
  \caption{%
  A visualization of an auxiliary graph constructed from an instance $(G,k)$ of \probname{Multicolored Clique}.
  The set $U$ consists of $n(k-1)$ vertices representing the vertices in $G$ and
  their adjacent colors.
  The set $W$ consists of $4m$ vertices representing the edges in $G$ replaced by a gadget for 
  creating odd cycles.
  }
  \label{fig:proof-w1-hardness}
\end{figure*}
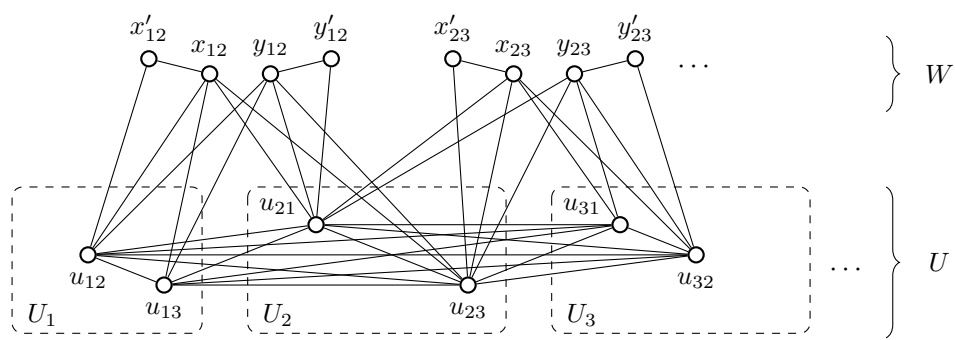

\section{Conclusion} \label{sec:conclusion}

Inspired by crown decompositions for \textsc{Vertex Cover} and antler decompositions for \textsc{Feedback Vertex Set}, we introduced the notion of (tight) odd cycle cuts to capture local regions of a graph in which a simple certificate exists for the membership of certain vertices in an optimal solution to \textsc{Odd Cycle Transversal}. In addition, we developed a fixed-parameter tractable algorithm to find a non-empty subset of vertices that belong to an optimal odd cycle transversal in input graphs admitting a tight odd cycle cut; the parameter~$k$ we employed is the \emph{width} of the tight \occ. Finding tight odd cycle cuts and removing the vertices certified to be in an 	optimal solution leads to search-space reduction for the natural parameterization of \textsc{Odd Cycle Transversal}. To obtain our results, one of the main technical ideas was to replace the use of minimum two-way separators that arise naturally when solving \textsc{Odd Cycle Transversal}, by minimum three-way separators that simultaneously handle breaking the odd cycles in a subgraph \emph{and} separating the resulting local bipartite subgraph from the remainder of the graph.

\subparagraph*{Theoretical challenges}
There are several interesting directions for follow-up work. We first discuss the theoretical challenges. The algorithm we presented runs in time~$2^{k^{\Oh(1)}} n^{\Oh(z)}$, where~$z$ is the order of the tight odd cycle cut in the output guarantee of \cref{thm:main}. The polynomial term in the exponent has a large degree, which is related to the size of the cut covering sets used to shrink the bipartite part of an odd cycle cut in terms of its width. While we expect that some improvements can be made by a more refined analysis, it would be more interesting to see whether an algorithmic approach that avoids color coding can lead to significantly faster algorithms. 

An odd cycle cut~$(X_B, X_C, X_R)$ of width~$|X_C| = k$ in a graph~$G$ gives rise to a $k$-secluded bipartite subgraph~$G[X_B]$; recall that a subgraph is called $k$-secluded if its open neighborhood has size~$k$. For enumerating inclusion-maximal \emph{connected} $k$-secluded subgraphs that satisfy a property~$\Pi$, a bounded-depth branching strategy was recently proposed~\cite{JansenKW23} that generalizes the enumeration of important separators. Can such branching techniques be used to improve the running time for the search-space reduction problem considered in this paper to~$2^{\Oh(k)} n^{\Oh(z)}$?

The dependence on the complexity~$z$ of the certificate is another topic for further investigation. The search-space reduction algorithm for \textsc{Feedback Vertex Set} by Donkers and Jansen~\cite{DonkersJ24} that inspired this work, also incurs a factor~$n^{\Oh(z)}$ in its running time. For \textsc{Feedback Vertex Set}, it is conjectured but not proven that such a dependence on~$z$ is unavoidable. The situation is the same for \textsc{Odd Cycle Transversal}. Is there a way to rule out the existence of an algorithm for the task of \cref{thm:main} that runs in time~$f(k, z) \cdot n^{\Oh(1)}$?

A last theoretical challenge concerns the definition of the substructures that are used to certify membership in an optimal odd cycle transversal. Our definition of an odd cycle cut~$(X_B, X_C, X_R)$ prohibits the existence of any edges between~$X_B$ and~$X_R$. Together with the requirement that~$G[X_B]$ is bipartite, this ensures that all odd cycles intersecting~$X_B$ are intersected by~$X_C$. In principle, one could also obtain the latter conclusion from a slightly less restricted graph decomposition. Since any odd cycle enters a bipartite subgraph on one edge and leaves via another, knowing that each connected component~$H$ of~$G[X_B]$ is connected to~$X_R$ by at most one edge is sufficient to guarantee that all odd cycles visiting~$X_B$ are intersected by~$X_C$. The prior work on antler structures for \textsc{Feedback Vertex Set} allows the existence of one pendant edge per component, and manages to detect such antler structures efficiently. It would be interesting to see whether our approach can be generalized for \emph{relaxed} odd cycle cuts in which each component of~$G[X_B]$ has at most one edge to~$X_R$. To adapt to this setting, one would have to refine the type of three-way separation problem that is used in the graph reduction step. 

For \textsc{Odd Cycle Transversal}, one could relax the definition of the graph decomposition even further: to ensure that odd cycles visiting~$X_B$ are intersected by~$X_C$, it would suffice for each connected component~$H$ of~$G[X_B]$ to have at most one neighbor~$v_H$ in~$X_R$, as long as all vertices of~$H$ adjacent to~$v_H$ belong to the \emph{same} side of a bipartition of~$H$.

\subparagraph*{Practical challenges}
Since the investigation of search-space reduction is inspired by practical considerations, we should not neglect to discuss practical aspects of this research direction. While we do not expect the algorithm as presented here to be practical, it serves as a proof of concept that rigorous guarantees on efficient search-space reduction can be formulated. Our work also helps to identify the types of substructures that can be used to reason locally about membership in an optimal solution. Apart from finding faster algorithms in theory and experimenting with their results, one could also target the development of specialized algorithms for concrete values of~$k$ and~$z$. 

For~$k=1$, a tight odd cycle cut of width~$1$ effectively consists of a cutvertex~$c$ of the graph whose removal splits off a bipartite connected component~$B$ but for which the subgraph induced by~$B \cup \{v\}$ contains an odd cycle. Preliminary investigations suggest that in this case, an algorithm that computes the block-cut tree, analyzes which blocks form non-bipartite subgraphs, and which cut vertices break all the odd cycles in their blocks, can be engineered to run in time~$\Oh(|V(G)| + |E(G)|)$ to find a vertex~$v$ belonging to an optimal odd cycle transversal when given a graph that has a tight odd cycle cut of width~$k=1$. Do linear-time algorithms exist for~$k>1$? These would form valuable reduction steps in algorithms solving \textsc{Odd Cycle Transversal} exactly, such as the one developed by Wernicke~\cite{Wernicke14}.

The~$k=1$ case of the \emph{relaxed} odd cycle cuts described above are in fact used as one of the reduction rules in Wernicke's algorithm~\cite[Rule 7]{Wernicke14}. His reduction applies whenever there is a triangle~$\{u,v,w\}$ in which~$w$ has degree two and~$v$ has degree at most three. Under these circumstances, there is an optimal solution that contains~$u$ while avoiding~$v$ and~$w$: since the removal of~$u$ decreases the degree of~$w$ to one, while~$w$ is one of the at most two remaining neighbors of~$v$, the removal of~$u$ breaks all odd cycles intersecting~$\{u,v,w\}$. This corresponds to the fact that the triple~$(X_B = \{v,w\}, X_C = \{u\}, X_R = V(G) \setminus \{u,v,w\})$ forms a tight \emph{relaxed} odd cycle cut. We interpret the fact that the~$k=1$ case was developed naturally in an existing algorithm as encouraging evidence that refined research into search-space reduction steps can eventually lead to impact in practice.

%%
%% Bibliography
%%

%% Please use bibtex, 
\newpage
\bibliography{main}

\begin{thebibliography}{10}

\bibitem{Abu-KhzamFLS07}
Faisal~N. Abu-Khzam, Michael~R. Fellows, Michael~A. Langston, and W.~Henry
  Suters.
\newblock Crown structures for vertex cover kernelization.
\newblock {\em Theory Comput. Syst.}, 41(3):411--430, 2007.
\newblock \href {https://doi.org/10.1007/s00224-007-1328-0}
  {\path{doi:10.1007/s00224-007-1328-0}}.

\bibitem{BumpusJK22}
Benjamin~Merlin Bumpus, Bart M.~P. Jansen, and Jari J.~H. de~Kroon.
\newblock Search-space reduction via essential vertices.
\newblock In Shiri Chechik, Gonzalo Navarro, Eva Rotenberg, and Grzegorz
  Herman, editors, {\em Proceedings of the 30th Annual European Symposium on
  Algorithms, {ESA} 2022}, volume 244 of {\em LIPIcs}, pages 30:1--30:15.
  Schloss Dagstuhl - Leibniz-Zentrum f{\"{u}}r Informatik, 2022.
\newblock URL: \url{https://doi.org/10.4230/LIPIcs.ESA.2022.30}, \href
  {https://doi.org/10.4230/LIPICS.ESA.2022.30}
  {\path{doi:10.4230/LIPICS.ESA.2022.30}}.

\bibitem{CyganFKLMPPS15}
Marek Cygan, Fedor~V. Fomin, Lukasz Kowalik, Daniel Lokshtanov, D{\'{a}}niel
  Marx, Marcin Pilipczuk, Michal Pilipczuk, and Saket Saurabh.
\newblock {\em Parameterized Algorithms}.
\newblock Springer, 2015.
\newblock \href {https://doi.org/10.1007/978-3-319-21275-3}
  {\path{doi:10.1007/978-3-319-21275-3}}.

\bibitem{CyganPPW13}
Marek Cygan, Marcin Pilipczuk, Michal Pilipczuk, and Jakub~Onufry Wojtaszczyk.
\newblock On multiway cut parameterized above lower bounds.
\newblock {\em {ACM} Trans. Comput. Theory}, 5(1):3:1--3:11, 2013.
\newblock \href {https://doi.org/10.1145/2462896.2462899}
  {\path{doi:10.1145/2462896.2462899}}.

\bibitem{DonkersJ19}
Huib Donkers and Bart M.~P. Jansen.
\newblock A {Turing} {Kernelization} {Dichotomy} for {Structural}
  {Parameterizations} of -{Minor}-{Free} {Deletion}.
\newblock In {\em Graph-{Theoretic} {Concepts} in {Computer} {Science}: 45th
  {International} {Workshop}, {WG} 2019, {Vall} de {Núria}, {Spain}, {June}
  19–21, 2019, {Revised} {Papers}}, pages 106--119, Berlin, Heidelberg, June
  2019. Springer-Verlag.
\newblock \href {https://doi.org/10.1007/978-3-030-30786-8_9}
  {\path{doi:10.1007/978-3-030-30786-8_9}}.

\bibitem{DonkersJ24}
Huib Donkers and Bart~M.P. Jansen.
\newblock Preprocessing to reduce the search space: Antler structures for
  feedback vertex set.
\newblock {\em Journal of Computer and System Sciences}, 144, 2024.
\newblock \href {https://doi.org/10.1016/j.jcss.2024.103532}
  {\path{doi:10.1016/j.jcss.2024.103532}}.

\bibitem{Fellows03}
Michael~R. Fellows.
\newblock Blow-ups, win/win's, and crown rules: Some new directions in {FPT}.
\newblock In Hans~L. Bodlaender, editor, {\em Proceedings of the 29th
  International Workshop on Graph-theoretic Concepts in Computer Science, WG
  2003}, volume 2880 of {\em Lecture Notes in Computer Science}, pages 1--12.
  Springer, 2003.
\newblock \href {https://doi.org/10.1007/978-3-540-39890-5_1}
  {\path{doi:10.1007/978-3-540-39890-5_1}}.

\bibitem{GoodrichHS21}
Timothy~D. Goodrich, Eric Horton, and Blair~D. Sullivan.
\newblock An updated experimental evaluation of graph bipartization methods.
\newblock {\em {ACM} J. Exp. Algorithmics}, 26:12:1--12:24, 2021.
\newblock \href {https://doi.org/10.1145/3467968} {\path{doi:10.1145/3467968}}.

\bibitem{Goodrich18}
Timothy~D. Goodrich, Blair~D. Sullivan, and Travis~S. Humble.
\newblock Optimizing adiabatic quantum program compilation using a
  graph-theoretic framework.
\newblock {\em Quantum Information Processing}, 17(5), 4 2018.
\newblock URL: \url{https://www.osti.gov/biblio/1460229}, \href
  {https://doi.org/10.1007/s11128-018-1863-4}
  {\path{doi:10.1007/s11128-018-1863-4}}.

\bibitem{Huffner09}
Falk H{\"u}ffner.
\newblock Algorithm engineering for optimal graph bipartization.
\newblock {\em J. Graph Algorithms Appl.}, 13(2):77--98, 2009.

\bibitem{JansenK21}
Bart M.~P. Jansen and Jari J.~H. de~Kroon.
\newblock {FPT} algorithms to compute the elimination distance to bipartite
  graphs and more.
\newblock In Lukasz Kowalik, Michal Pilipczuk, and Pawel Rzazewski, editors,
  {\em Graph-Theoretic Concepts in Computer Science - 47th International
  Workshop, {WG} 2021, Warsaw, Poland, June 23-25, 2021, Revised Selected
  Papers}, volume 12911 of {\em Lecture Notes in Computer Science}, pages
  80--93. Springer, 2021.
\newblock \href {https://doi.org/10.1007/978-3-030-86838-3\_6}
  {\path{doi:10.1007/978-3-030-86838-3\_6}}.

\bibitem{JansenKW23}
Bart M.~P. Jansen, Jari J.~H. de~Kroon, and Michal Wlodarczyk.
\newblock Single-exponential {FPT} algorithms for enumerating secluded
  {\(\mathscr{f}\)}-free subgraphs and deleting to scattered graph classes.
\newblock In Satoru Iwata and Naonori Kakimura, editors, {\em Proceedings of
  the 34th International Symposium on Algorithms and Computation, {ISAAC}
  2023}, volume 283 of {\em LIPIcs}, pages 42:1--42:18. Schloss Dagstuhl -
  Leibniz-Zentrum f{\"{u}}r Informatik, 2023.
\newblock \href {https://doi.org/10.4230/LIPICS.ISAAC.2023.42}
  {\path{doi:10.4230/LIPICS.ISAAC.2023.42}}.

\bibitem{JansenKW22}
Bart M.~P. Jansen, Jari J.~H. de~Kroon, and Michał Włodarczyk.
\newblock Vertex {Deletion} {Parameterized} by {Elimination} {Distance} and
  {Even} {Less}, 2022.
\newblock arXiv:2103.09715 [cs].
\newblock \href {https://doi.org/10.48550/arXiv.2103.09715}
  {\path{doi:10.48550/arXiv.2103.09715}}.

\bibitem{JansenV24}
Bart M.~P. Jansen and Ruben F.~A. Verhaegh.
\newblock Search-space reduction via essential vertices revisited: Vertex
  multicut and cograph deletion.
\newblock In Hans~L. Bodlaender, editor, {\em Proceedings of the 19th
  Scandinavian Symposium on Algorithm Theory, SWAT 2024}, LIPIcs. Schloss
  Dagstuhl - Leibniz-Zentrum f{\"{u}}r Informatik, 2024.
\newblock In press.
\newblock \href {https://arxiv.org/abs/2404.09769} {\path{arXiv:2404.09769}}.

\bibitem{KratschW14}
Stefan Kratsch and Magnus Wahlstr{\"{o}}m.
\newblock Compression via matroids: {A} randomized polynomial kernel for odd
  cycle transversal.
\newblock {\em {ACM} Trans. Algorithms}, 10(4):20:1--20:15, 2014.
\newblock \href {https://doi.org/10.1145/2635810} {\path{doi:10.1145/2635810}}.

\bibitem{KratschW20}
Stefan Kratsch and Magnus Wahlstr{\"{o}}m.
\newblock Representative sets and irrelevant vertices: New tools for
  kernelization.
\newblock {\em J. {ACM}}, 67(3):16:1--16:50, 2020.
\newblock \href {https://doi.org/10.1145/3390887} {\path{doi:10.1145/3390887}}.

\bibitem{LokshtanovNRRS14}
Daniel Lokshtanov, N.~S. Narayanaswamy, Venkatesh Raman, M.~S. Ramanujan, and
  Saket Saurabh.
\newblock Faster parameterized algorithms using linear programming.
\newblock {\em {ACM} Trans. Algorithms}, 11(2):15:1--15:31, 2014.
\newblock \href {https://doi.org/10.1145/2566616} {\path{doi:10.1145/2566616}}.

\bibitem{NaorSS95}
Moni Naor, Leonard~J. Schulman, and Aravind Srinivasan.
\newblock Splitters and near-optimal derandomization.
\newblock In {\em 36th Annual Symposium on Foundations of Computer Science,
  Milwaukee, Wisconsin, USA, 23-25 October 1995}, pages 182--191. {IEEE}
  Computer Society, 1995.
\newblock \href {https://doi.org/10.1109/SFCS.1995.492475}
  {\path{doi:10.1109/SFCS.1995.492475}}.

\bibitem{PilipczukZ18}
Marcin Pilipczuk and Michal Ziobro.
\newblock Experimental evaluation of parameterized algorithms for graph
  separation problems: Half-integral relaxations and matroid-based
  kernelization.
\newblock {\em CoRR}, abs/1811.07779, 2018.
\newblock \href {https://arxiv.org/abs/1811.07779} {\path{arXiv:1811.07779}}.

\bibitem{ReedSV04}
Bruce~A. Reed, Kaleigh Smith, and Adrian Vetta.
\newblock Finding odd cycle transversals.
\newblock {\em Oper. Res. Lett.}, 32(4):299--301, 2004.
\newblock \href {https://doi.org/10.1016/j.orl.2003.10.009}
  {\path{doi:10.1016/j.orl.2003.10.009}}.

\bibitem{Schrijver03}
Alexander Schrijver.
\newblock {\em Combinatorial optimization: polyhedra and efficiency},
  volume~24.
\newblock Springer, 2003.

\bibitem{Wernicke14}
Sebastian Wernicke.
\newblock {\em On the algorithmic tractability of single nucleotide
  polymorphism (SNP) analysis and related problems}.
\newblock diplom.de, 2014.

\end{thebibliography}

% \newpage
% \appendix
% \input{sections/a_appendix.tex}

\end{document}